\pdfoutput=1
\documentclass[
  aps,
  prx,
  reprint,
  amsmath,
  amssymb,
  superscriptaddress,
  nofootinbib,
  noeprint,
  longbibliography,
  floatfix,
  onecolumn,
  nobibnotes
]{revtex4-2}

\usepackage{xprintlen}
\usepackage[T1]{fontenc}
\usepackage{times}
\usepackage{graphicx}
\usepackage[table,usenames,dvipsnames]{xcolor}
\usepackage{multirow}

\usepackage{hyperref}
\hypersetup{
  colorlinks=true,
  hypertexnames=false,
  linktocpage,
  colorlinks=true, 
  urlcolor=magenta!90!black,    
  linkcolor=blue!60!black, 
  citecolor=black!60 
}

\usepackage{physics,siunitx}
\usepackage{bbm,mathtools}
\usepackage{amsthm,amssymb,amsmath,mathrsfs}
\usepackage{xfrac} 
\usepackage{MnSymbol} 
\usepackage{youngtab}
\usepackage{empheq}

\usepackage{booktabs}
\usepackage{enumitem}
\usepackage{placeins} 

\usepackage[capitalize,compress]{cleveref}

\usepackage{complexity}
\usepackage{pgfplots}
\usepackage{tikz}
\usepackage{orcidlink}
\usepackage{pgffor}
\usepackage[ruled,vlined,linesnumbered]{algorithm2e}
\SetKwInput{KwInit}{Initial Input}
\SetKwInput{KwQ}{Query}
\SetKwProg{Def}{def}{}{}

\usepackage{bm}

\pgfplotsset{compat=1.18}
\usepgfplotslibrary{
  fillbetween,
  }

\usetikzlibrary{
  pgfplots.groupplots,
  arrows.meta,
  backgrounds,
  chains,
  decorations.pathreplacing,
  decorations.pathmorphing,
  fit,
  intersections,
  through,
  positioning,
  shapes,
  shapes.geometric,
  shapes.arrows,
  calc,
  3d
  }

\makeatletter 
 \hypersetup{
  pdftitle = {Learning Noise-Robust Stabilizer Structure via Bell Sampling},
  pdfauthor = {Letian Tang, Thomas Steckmann},
  pdfkeywords = {
    Bell sampling, 
    stabilizer nullity, 
    magic, 
    Pauli noise, 
    stabilizer learning
    }
  }
\makeatother

\usepackage{xprintlen}
\usepackage{mathtools}

\newclass{\stoqma}{StoqMA}
\newclass{\classP}{P}
\newclass{\bqp}{BQP}
\newclass{\qcam}{QCAM}
\newclass{\postbqp}{postBQP}
\newclass{\posta}{postA}
\newclass{\postiqp}{postIQP}
\newclass{\classa}{A}
\newclass{\bpp}{BPP}
\newclass{\fbpp}{FBPP}
\newclass{\pp}{PP}
\newclass{\cocp}{coC_=P}
\newclass{\ph}{PH}
\newclass{\np}{NP}
\newclass{\conp}{coNP}
\newclass{\gapp}{GapP}
\newclass{\approxclass}{Apx}
\newclass{\gapclass}{Gap}
\newclass{\sharpP}{\#P}
\newclass{\ma}{MA}
\newclass{\am}{AM}
\newclass{\qma}{QMA}

\newclass{\hog}{HOG}
\newclass{\quath}{QUATH}
\newclass{\bog}{BOG}
\newclass{\xeb}{XEB}
\newclass{\xhog}{XHOG}
\newclass{\xquath}{XQUATH}
\newclass{\maxcut}{MAXCUT}
\newclass{\sat}{SAT}
\newclass{\maxtwosat}{MAX2SAT}
\newclass{\twosat}{2SAT}
\newclass{\threesat}{3SAT}
\newclass{\sharpsat}{\#SAT}
\newclass{\se}{Sign Easing}
\newclass{\classx}{X}

\newtheorem{theorem}{Theorem}
\newtheorem{proposition}[theorem]{Proposition}
\newtheorem{definition}[theorem]{Definition}
\newtheorem{lemma}[theorem]{Lemma}
\newtheorem{problem}[theorem]{Problem}

\newtheorem{assumption}[theorem]{Assumption}
\newtheorem{observation}[theorem]{Observation}
\newtheorem{notation}[theorem]{Notation}
\newtheorem{remark}[theorem]{Remark}

\newtheorem{corollary}[theorem]{Corollary}

\DeclareMathOperator{\supp}{supp}

\DeclareMathOperator{\spn}{span}
\DeclareMathOperator{\rad}{rad}

\newcommand{\ceil}[1]{\lceil #1 \rceil}

\newcommand{\strictFloor}[1]{\lfloor #1 \rfloor_<}

\newcommand{\proj}[1]{\ket{#1}\bra{#1}}

\DeclarePairedDelimiterX{\infdivx}[2]{(}{)}{%
  #1\;\delimsize\|\;#2%
}

\numberwithin{theorem}{section} 

\newcommand{\restatedname}{}
\newtheorem*{restatedthm}{\restatedname}
\newenvironment{restated}[1]{\renewcommand{\restatedname}{#1 (restated)}\begin{restatedthm}}{\end{restatedthm}}

\begin{document}

\title{Learning Noise-Robust Stabilizer Structure via Bell Sampling}

\author{Letian Tang\,\orcidlink{0000-0001-8882-9187}}
\affiliation{Joint Center for Quantum Information and Computer Science,
University of Maryland and NIST, College Park, Maryland 20742, USA}
\affiliation{Department of Physics,
University of Maryland, College Park, Maryland 20742, USA}
\affiliation{Department of Physics, Carnegie Mellon University, Pittsburgh, PA 15213, United States}
\affiliation{Department of Computer Science, Carnegie Mellon University, Pittsburgh, PA 15213, United States}
\author{Thomas Steckmann\,\orcidlink{0000-0001-6012-2948}}
\affiliation{Joint Center for Quantum Information and Computer Science,
University of Maryland and NIST, College Park, Maryland 20742, USA}
\affiliation{Joint Quantum Institute,
University of Maryland and NIST, College Park, Maryland 20742, USA}

\date{\today}

\begin{abstract}
Bell sampling efficiently learns the stabilizer group, and hence the stabilizer nullity, of pure states with little magic, a primitive underlying tomography and certification of low-magic states. Real devices, however, only prepare noisy images of the target, which we model as a Pauli channel acting on the ideal state; the resulting noisy state generically has no exact Pauli stabilizers. Existing structure learners for mixed states rely on such exact stabilizers. Agnostic tomography does not, but its guarantees scale exponentially in the nullity $t$. We show that the ideal stabilizer group remains encoded in the noisy state as the set of Pauli operators under which it is invariant. We give Bell-sampling algorithms that learn the group to accuracy $\varepsilon$ for states of nullity at most $t$. The difficulty is that noisy samples break the group structure pure-state learners rely on, so more samples can erase structure rather than reveal it. We resolve this with a finder--verifier scheme that extracts stabilizer structure from partially corrupted sample batches while certifying proposed updates. For sufficiently weak global depolarizing noise, both samples and runtime are $\mathrm{poly}(n,t,1/\varepsilon)$ when the channel's total error probability is at most of order $\varepsilon n/t$. Up to a constant total error probability ($\approx 0.16$), samples remain polynomial while classical processing costs $2^t$, and for any error probability below one a $4^t$ overhead suffices. Arbitrary Pauli channels exhibit the same three regimes, with more restrictive thresholds. For weak Pauli noise, this closes much of the gap between agnostic learners such as stabilizer bootstrapping and polynomial pure-state algorithms.
\end{abstract}

\maketitle

\section{Introduction}\label{eq:introduction}

Magic, or non-stabilizerness, describes how far a quantum state lies from the set of stabilizer states~\cite{veitch2014}, and it can be quantified by measures such as the stabilizer R\'enyi entropy that are estimable from Bell measurements~\cite{leone2022a, haug2023}. Stabilizer states are classically tractable in two senses: they can be efficiently simulated by the Gottesman--Knill theorem~\cite{gottesman1998,aaronson2004a} and they can be learned efficiently from $O(n)$ copies by Bell sampling~\cite{montanaro2017}. Both results extend to states with little magic, through simulation and learning algorithms whose cost is polynomial in the number of qubits and exponential only in a magic measure such as the stabilizer nullity~\cite{bravyi2016,bravyi2019, grewal2025, hangleiter2024}. Magic is also the resource that fault-tolerant quantum computation must account for. The Eastin--Knill theorem rules out a transversal universal gate set under standard assumptions~\cite{eastin2009,zeng2007,bravyi2013}, and for many codes of practical interest the gates that must be supplied non-transversally are precisely the non-Clifford gates, which are typically implemented by consuming distilled magic states~\cite{beverland2022, bravyi2005}. Measures of magic therefore quantify both the cost of preparing a state fault-tolerantly~\cite{beverland2020,howard2017} and the cost of simulating or learning it classically.

A growing body of work has sought learning-theoretic handles on magic~\cite{anshu2024}. Given copies of an unknown state, how efficiently can one exploit limited magic to learn the state itself? 
Work on these questions began with exact stabilizer states, which can be tested and learned efficiently from Bell measurements~\cite{gross2021, montanaro2017, chang2026}, and has since extended to progressively broader classes: states of low stabilizer nullity or prepared with few $T$ gates~\cite{leone2024, grewal2025, hangleiter2024, cho2026}, states of bounded stabilizer rank or extent~\cite{arunachalam2025, arunachalam2026}, Clifford-scrambled product states~\cite{haug2026scrambled, aguilar2026}, and, under query access, the corresponding few-$T$ circuits and Clifford-structured unitaries~\cite{lai2022, dutt2026}. In most of these works, learning the stabilizer structure (the stabilizer group, or a Clifford that compresses the magic onto few qubits) is the first step toward full tomography of a pure input. A smaller line of work treats arbitrary inputs, either by agnostic tomography, which returns the best approximation within a class~\cite{chen2025, grewal2026agnostic, bakshi2025, qi2026}, or by tolerant testing and witnessing of magic in mixed states~\cite{iyer2024,haug2025witness}.

These questions are no longer purely theoretical~\cite{miller2026}. In early fault-tolerant devices~\cite{bluvstein2024, salesrodriguez2025} and in the verification of logical state preparation, the logical states of interest are typically low-magic states for which existing learning algorithms are formulated, and certifying their non-Clifford content is a natural diagnostic~\cite{leitch2026, oliviero2022}. But these pure states are not experimentally accessible; instead, only noisy access is ever possible. This work studies a basic question in that practical setting: how efficiently can the stabilizer structure of a pure state be learned from measurements of its noisy image?

A natural, albeit fragile, measure of magic is the stabilizer nullity---defined by the stabilizer group of an $n$-qubit pure state $\ket{\psi}$, with density matrix $\psi = \proj{\psi}$,
\begin{align}
  \mathrm{Stab}(\psi) &= \{\sigma_i \in \mathcal{P}^n : \sigma_i\ket{\psi} = \pm\ket{\psi}\},\nonumber\\
  \qquad
  \nu(\psi) &= n - \dim \left(\mathrm{Stab}(\psi)\right).
\end{align}
Here, $\mathcal{P}^n$ is the $n$-qubit Pauli group up to phases, which we identify with $\mathbb{F}_2^{2n}$ through the binary symplectic representation; $\mathrm{Stab}(\psi)$ is then a subspace, and $\dim$ denotes its $\mathbb{F}_2$-dimension. The nullity $\nu$ is defined such that it is a magic monotone, i.e., invariant under Clifford operations and increasing with increasing magic. A single $T$ gate changes the nullity by at most one, so $\nu$ is a lower bound on the number of $T$ gates used to prepare $\ket{\psi}$. Throughout, $t$ denotes a supplied upper bound $\nu\leqslant t$; the $T$-count of any preparation circuit provides one~\cite{beverland2020}. Additionally, stabilizer nullity is equivalent to the notion of $\nu$-compressibility: there exists some Clifford $C$ that compresses the state to $\nu$ qubits, such that $C|\psi\rangle = |\phi\rangle \otimes |0\rangle^{\otimes (n-\nu)}$. This property is used in tomography of low-magic states, because once $\nu$ and $C$ are known, full state tomography need only be performed on $\nu$ qubits of the full $n$-qubit system~\cite{grewal2025, hangleiter2024, chia2024}.

Both the stabilizer group and the algorithms that learn it are most naturally phrased through the Pauli expansion of the state. Writing an $n$-qubit state as $\rho=\sum_{i=1}^{4^n} a_i\sigma_i$, its \textit{Pauli support} is the subspace spanned by the Paulis with nonzero coefficients,\footnote{Here the Pauli support is the span of the nonzero labels, not the set of labels itself.}
\begin{equation}
\supp(\rho) = \spn\{\sigma_i \in \mathcal{P}^n: a_i \neq 0\}.
\label{eq:pauli-support}
\end{equation}
For a pure state, the stabilizer group corresponds to the centralizer of the Pauli support,
\begin{equation}
\mathrm{Stab}(\psi)= C\left(\supp \psi\right),
\label{eq:stab-centralizer}
\end{equation}
so learning $\mathrm{Stab}(\psi)$ amounts to learning the subspace $\supp(\psi)$.

Stabilizer nullity is a fragile measure because exact stabilizer structure is discontinuous. An arbitrarily small rotation $R_Z(\theta)\ket{+}$ removes the stabilizer $X$ of $\ket{+}$, while moving the state by only $O(\theta)$ in trace distance. With finitely many copies, stabilizer structure can therefore be resolved only up to a finite accuracy. We make this precise in Sec.~\ref{sec:approx-support}.

Learning the stabilizer group is the natural first step in learning a low-magic state: its dimension certifies the nullity, and once it is known, full tomography reduces to a $\nu$-qubit problem. Montanaro first showed that transversal two-copy measurement in the Bell basis---Bell sampling---exposes the subgroup structure of stabilizer states, and that learning the sign-free description of the stabilizer group is efficient in $n$ in both sample and time complexity, requiring $O(n)$ samples~\cite{montanaro2017}. A sequence of works~\cite{grewal2025, hangleiter2024} extended Bell sampling beyond exact stabilizer states, independently arriving at algorithms that learn the approximate stabilizer structure of $t$-doped or $\nu$-compressible states using $\mathrm{poly}(n,\nu)$ samples and time at fixed accuracy.

The underlying mechanism is as follows. A Bell sample, the outcome of measuring two copies of $\ket{\psi}$ in the Bell basis, is a length-$2n$ bitstring encoding a Pauli operator in $\supp(\psi)$ shifted by an unknown hidden Pauli $\sigma_k$ (Proposition~\ref{prop:hidden-shift}). The bitwise XOR of two Bell samples, a \textit{Bell difference sample}, cancels the shift and lies entirely in $\supp(\psi)$. The span of sufficiently many difference samples therefore reproduces the support with high probability, and its centralizer, which is efficiently computable, is the stabilizer group. Components of the support that are sampled only rarely may never be observed, so these algorithms naturally recover the stabilizer structure only approximately. The same group-theoretic structure underlies a broader view of stabilizer learning as a hidden subgroup problem~\cite{hinsche2026}.

Beyond pure inputs, even the learning problem itself is less settled. The most developed line is agnostic tomography, which accepts an arbitrary mixed input and returns the low-magic state of nearly largest fidelity with it~\cite{chen2025, grewal2026agnostic}. Tolerant testers likewise accept mixed inputs, but decide whether the state is close to a stabilizer state rather than recovering one~\cite{iyer2024}. Structure learners for mixed states, by contrast, have so far required a strong promise: that a subgroup of Pauli operators stabilizes the observed state exactly, $\sigma_i\rho = \pm\rho$, or equivalently has expectation values of unit magnitude~\cite{grewal2025, hinsche2026, cho2026}.

The stabilizer group admits a second characterization that makes the effect of noise transparent. We call a Pauli $\sigma_i$ a \emph{strong} symmetry of a state $\rho$ if $\sigma_i\rho = \pm\rho$, and a \emph{weak} symmetry if it commutes with $\rho$, i.e., $\sigma_i\rho\sigma_i^\dagger = \rho$. The strong symmetries form the usual stabilizer group; the weak symmetries form a group containing it, which we call the \emph{weak stabilizer} of $\rho$. Conjugation by $\sigma_i$ flips the sign of exactly those Pauli operators that anticommute with it, so the weak stabilizer consists of the Paulis commuting with the entire support of $\rho$:
\begin{equation}
    S \equiv \{\sigma_i \in \mathcal{P}^n : [\sigma_i, \rho] = 0\} = C(\supp\rho).
    \label{eq:weak-symmetry}
\end{equation}
For a pure state the two notions coincide: $\sigma_i\psi\sigma_i^\dagger=\psi$ implies $\sigma_i\ket{\psi}\propto\ket{\psi}$, so $S=\mathrm{Stab}(\psi)$, as in Eq.~\eqref{eq:stab-centralizer}. Noise separates them. Even arbitrarily weak depolarizing noise makes $\rho=\mathcal{E}(\psi)$ full rank, and a full-rank state has no strong Pauli symmetries other than the identity; the stabilizer group, and the learning algorithms built on it, therefore no longer apply directly.

The weak symmetries, however, survive. A Pauli channel rescales each Pauli coefficient of the state, so it can shrink the support but never enlarge it, and hence $\mathrm{Stab}(\psi) \subseteq C(\supp\rho)$. When the total error probability $\lambda$ of the channel is below $1/2$, no coefficient is rescaled to zero (Sec.~\ref{sec:noise-model}) and the inclusion becomes an equality. For global depolarizing noise, this holds for any $\lambda<1$. Our target is thus the stabilizer group of the underlying pure state, which for such support-preserving channels is the weak stabilizer of the noisy state:
\begin{equation}
    \mathrm{Stab}(\psi) = S.
\end{equation}
Hinsche, Eisert, and Carrasco distinguish this weak invariance from the strong, unit-expectation condition required by their Fourier-sampling learner~\cite{hinsche2026}; here we target the weak stabilizer directly. This allows us to naturally generalize the notion of the stabilizer nullity by instead considering the dimension of the weak stabilizer,
\begin{equation}
\nu \equiv n - \dim\left(C(\supp\rho)\right)  = \nu(\psi).
\end{equation}
This quantity will be the central focus of our learning algorithm, as it allows one to learn, despite noise, an approximate weak stabilizer: an unsigned group containing the original stabilizer group, with nullity $\tilde\nu\leqslant\nu$.

Learning $C(\supp\rho)$ from Bell samples is not, however, a direct application of the pure-state algorithms. As noisy sampling produces erroneous difference samples lying outside the support---and these do not respect the group structure of the support---the pure-state stabilizers need not commute with them. Let $B$ be the span of a batch of $M$ noisy difference samples. Under nonzero global depolarizing noise, this noisy group converges to the full Pauli group as $M\rightarrow \infty$, with a trivial centralizer. Yet for sufficiently low noise and a carefully chosen batch size $M$, $B$ still retains useful information about $C(\supp\rho)$. The central observation of this work is that a partially corrupted batch still carries some of this structure, because each bad sample can remove only one target direction from the radical of the span. Verified directions can then be used to post-select later batches, which under depolarizing noise suppresses the bad-sample rate exponentially in the number of accepted directions. Making this rigorous requires a verifier that certifies every proposed update from one reusable set of Bell samples, and a bound connecting the purity the verifier certifies to the samples the finder observes.

\section{Summary of Results}\label{sec:summary-of-results}
We give an algorithm that learns an approximate weak stabilizer $\tilde S$, and with it the nullity $\tilde\nu=n-\dim\tilde S\leqslant\nu$, from noisy Bell samples. Our main case is global depolarizing noise. The algorithm alternates between a finder, which proposes updates to a working group from the span of Bell difference samples, and a verifier, which accepts an update only if it costs little purity. Depending on the noise strength $\lambda$ and the accuracy $\varepsilon$, the finder can use the whole span, part of it, or single samples, which gives three regimes with different sample and runtime costs in $n$ and $t$.

\subsection{Problem Setup}
Before discussing the algorithm in depth, we formalize the notation used throughout this manuscript and define the relevant scale quantities: the nullity $\nu$, the approximation error $\varepsilon$, and the noise strength $\lambda$.
\subsubsection{Approximate support}\label{sec:approx-support}
As discussed in Sec.~\ref{eq:introduction}, exact stabilizer structure is discontinuous: states with different stabilizer nullity can be arbitrarily close in trace distance.

This motivates an approximate learning problem. Rather than demanding exact recovery of $\mathrm{Stab}(\psi)$, we allow for additional stabilizer constraints that are compatible with the state up to accuracy $\varepsilon$. To apply this notion directly to a noisy mixed state, write
\begin{equation}
    \rho=\sum_{i\in\mathcal P^n} a_i\sigma_i,
\end{equation}
and, for a commuting Pauli group $\tilde S$ (in practice, a candidate update to the working group), define the Pauli twirl
\begin{equation}
    \mathcal T_{\tilde S}(\rho)
    \equiv \frac{1}{|\tilde S|}
    \sum_{g\in\tilde S} g\rho g
    =\sum_{i\in C(\tilde S)}a_i\sigma_i .
\end{equation}
The twirl projects $\rho$ onto the Pauli components which commute with $\tilde S$, i.e., onto its centralizer $C(\tilde S)$. We quantify the retained support using the \emph{projected purity}, defined as
\begin{equation}
    \gamma_{\tilde S}
    \equiv
    \tr\!\left[\mathcal T_{\tilde S}(\rho)^2\right]
    =
    2^n\sum_{i\in C(\tilde S)}a_i^2,
\end{equation}
and then define the corresponding purity loss function
\begin{equation}
    \Delta\gamma(\tilde S)
    \equiv
    \gamma-\gamma_{\tilde S}
    =
    2^n\sum_{i\notin C(\tilde S)}a_i^2,
    \qquad
    \gamma\equiv\tr(\rho^2).
\end{equation}

Our task is to find an \emph{$\varepsilon$-approximate weak stabilizer} $\tilde S$ satisfying
\begin{equation}
    S\subseteq\tilde S,
    \qquad
    \Delta\gamma(\tilde S)\leqslant\varepsilon .
    \label{eq:approx-support-centralizer}
\end{equation}
The first condition ensures that the learned group contains the exact algebraic structure of the target state, while the second limits how many additional constraints may be imposed: $C(\tilde S)$ must retain all but $\varepsilon$ of the purity. Equivalently, we search for a lower-nullity description of the underlying pure state $\psi$ while controlling the error incurred by that compression.

The accuracy parameter $\varepsilon$ has a direct state-space interpretation. If $\rho$ is obtained from a pure state $\psi$ by Pauli noise of total error probability $\lambda<1/2$, then for fixed commuting $\tilde S$,
\begin{equation}
    \sqrt{\frac{\Delta\gamma(\tilde S)}{2}}
    \;\leqslant\;
    \inf_{\phi:\,\tilde S\subseteq\mathrm{Stab}(\phi)}
    D_{\rm tr}(\psi,\phi)
    \;\leqslant\;
    \frac{\sqrt{\Delta\gamma(\tilde S)}}{1-2\lambda} .
    \label{eq:purity-trace-distance}
\end{equation}
Hence, for sufficiently small constant noise, $\varepsilon$ is, up to constant factors, the square of the trace-distance resolution at which stabilizer structure can be distinguished. A justification is given in Appendix~\ref{app:fourier}.

\subsubsection{Noise Model}\label{sec:noise-model}
We are given copies of $\rho = \mathcal{E}(\psi)$, where $\mathcal{E}(\rho)=\sum_Q p_Q\,Q\rho Q$ is a Pauli channel with total error probability $\lambda \equiv \sum_{Q\neq I} p_Q$. It rescales each Pauli coefficient by its fidelity,
\begin{equation}
    \mathcal{E}(\sigma_i) = f_i\,\sigma_i,
    \qquad
    f_i = \sum_Q p_Q (-1)^{\langle i, Q\rangle}
    \geqslant 1-2\lambda ,
    \label{eq:pauli-fidelities}
\end{equation}
so for $\lambda<1/2$ it preserves $\supp\rho$. Our expository case is global depolarizing noise,
\begin{align}
\mathcal{E}_\lambda (\psi)
= (1-\lambda)\,\psi + \frac{\lambda}{2^n}I,
\label{eq:depolarizing}
\end{align}
for which $f_i = 1-\lambda$ for all $\sigma_i \neq I$; its total error probability is $\lambda(1-4^{-n})\leqslant\lambda$, so bounds stated for Pauli channels apply to it with the same $\lambda$.

This noise model is salient in the early fault-tolerant regime. An encoded state in the code space of an $[[n,k,d]]$ stabilizer code has nullity at most $k$, and the code provides the stabilizer structure: every Pauli that stabilizes a code state commutes with the code's stabilizer group, which supplies $n-k$ independent stabilizers and leaves only the $k$ logical qubits to carry magic~\cite{nielsen2010,gottesmanbook}. Syndrome measurement collapses errors of weight below $d/2$, including coherent errors, into a Pauli channel on $\psi$: it decoheres any error into Pauli errors with a definite syndrome, and correctable Paulis that share a syndrome differ by a stabilizer, so they act identically on $\psi$~\cite{nielsen2010,gottesmanbook}. Other circuit operations, such as the Bell measurement, are Clifford, and their faults are Pauli under twirling~\cite{knill2005,wallman2016}. For families of unencoded Clifford$+T$ circuits and those forming unitary designs, we leave the particular extension to the Bell sampling regime for future work.

\subsubsection{Problem Definition}
We can now formalize the question explored in this work:

\begin{problem}[Main Problem]
\label{prob:main}
Given Bell-sampling access to a noisy state
\begin{equation}
\rho=\mathcal E(\psi),
\qquad
\psi=|\psi\rangle\langle\psi|,
\end{equation}
can we recover the stabilizer structure of the underlying pure state $\psi$, up to a prescribed accuracy?

\begin{description}[font=\normalfont\bfseries\itshape,leftmargin=2em]
\item[Promise]
The ideal $n$-qubit state $\psi$ has stabilizer nullity $\nu\leqslant t<n$, and the noise channel belongs to a specified class of global depolarizing or Pauli channels characterized by a noise-strength parameter $\lambda$.

\item[Input]
Copies of $\rho$, an accuracy parameter $\varepsilon>0$, and a confidence parameter $\delta\in(0,1)$.

\item[Output]
With probability at least $1-\delta$, return a commuting Pauli group $\tilde S$ that is an $\varepsilon$-approximate weak stabilizer of $\rho$, for which
\begin{equation}
    S\subseteq\tilde S,
    \qquad
    \Delta\gamma(\tilde S)
    =\gamma-\gamma_{\tilde S}
    \leqslant\varepsilon.
\end{equation}
Here, $S$ denotes the target weak stabilizer. Equivalently, $\tilde S$ defines an approximate compression of the underlying pure state $\psi$ to
\begin{equation}
    \tilde\nu=n-\dim\tilde S\leqslant\nu
\end{equation}
logical qubits, with approximation error controlled by $\varepsilon$.

\end{description}
\end{problem}

When required by the algorithm, $\lambda$ may either be supplied as part of the noise-model promise or bounded from the measurement data; in particular, for white noise it can be inferred from the purity of $\rho$.

\subsection{Main Results}
\subsubsection{Algorithm}

Our algorithm iteratively constructs an $\varepsilon$-approximate weak stabilizer by extracting the algebraic structure of the Bell difference samples and computing second-moment quantities using the Bell sample output distribution. It consists of two subroutines, shown together in Fig.~\ref{fig:pipeline}:
\begin{enumerate}
    \item A \textit{verifier} reuses one set of Bell samples to estimate the incremental purity loss, $\Delta \gamma$, associated with a candidate group $\tilde S$, with a guarantee that holds uniformly over all queries. If the candidate entails a sufficiently small $\Delta \gamma$, it is accepted to extend the current accepted group and consumes some of the allocated purity-loss budget.
    \item A \textit{finder} processes independent batches of Bell difference samples to learn useful structure from the span of each batch. It proposes candidate groups to the verifier based on this span by computing the radical (the intersection of the centralizer of the span with the span).
\end{enumerate}

\begin{figure*}[t]
\centering
\includegraphics{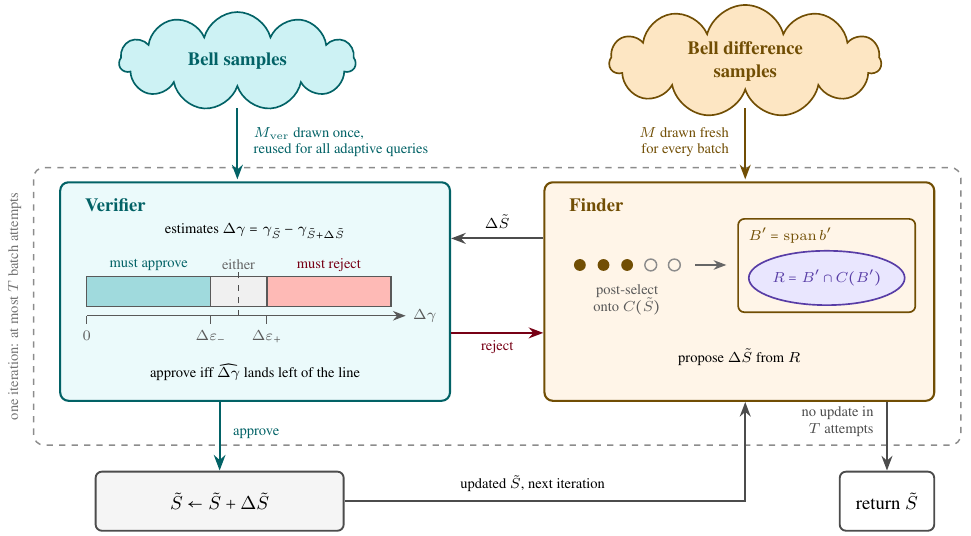}
\caption{The two subroutines and the loop they form. The verifier draws $M_{\mathrm{ver}}$ Bell samples once at initialization and reuses them for every query, so its sample cost is fixed before the finder starts; the finder draws a fresh batch of $M$ Bell difference samples each time. Each batch is post-selected onto $C(\tilde S)$, spanned, and reduced to its radical $R$, from which an increment $\Delta\tilde S$ is proposed. The verifier must approve when the true purity drop falls below $\Delta\varepsilon_-$ and must reject when it exceeds $\Delta\varepsilon_+$; between the two thresholds, which bracket the accuracy spent on a single update, either answer is allowed, and the algorithm decides by comparing its estimate $\widehat{\Delta\gamma}$ with the midpoint. An approved increment is applied and the iteration ends; otherwise the finder proposes again, for at most $T$ batches. The clean and intermediate modes of Sec.~\ref{sec:up-down} differ only in what is offered from $R$---the whole radical at once, or single elements in turn, so that in the intermediate mode several proposals may come from one batch. The loop is schematic: the final stopping checks and the noisy-regime fallback of Sec.~\ref{sec:top-down} are not shown.}
\label{fig:pipeline}
\end{figure*}
The central tension for the finder is the choice of batch size. Some Bell difference samples in a batch are \emph{good}, in the sense that they are consistent with the target weak stabilizer, while others are \emph{bad} samples introduced by noise. More samples are needed for the span to reveal stabilizer structure, but each bad sample adds a commutation constraint and can shrink the radical of the batch.

The radical degrades in a controlled way. Once the good samples in a batch span the target, each bad sample removes at most one target direction from the radical of the span, and the radical contains at most $t$ directions outside the target (Sec.~\ref{sec:bad-samples}). Under depolarizing noise, post-selecting later batches on the accepted group also suppresses bad samples: each accepted direction halves their rate, while the rate of good samples stays bounded below. The first update is therefore the bottleneck. How many good samples a batch needs is set by the purity--escape bound, Lemma~\ref{lemma:belldiff-and-gamma}: whenever a candidate group loses purity, Bell difference samples leave its centralizer at a proportional rate.

These facts give three algorithmic regimes:

\begin{itemize}
\item \textit{Clean regime.} A batch can be made large enough that its good samples ensure the whole radical has small purity loss, with few enough bad samples to preserve useful target directions. Whole-radical updates then keep both sample and classical costs polynomial.

\item \textit{Intermediate regime.} A batch can no longer be made large enough to certify the whole radical without admitting too many bad samples, but its radical still contains new target directions. The finder instead tests random elements of the radical. Because the radical contains at most $t$ directions outside the target, a fraction $\Omega(2^{-t})$ of these elements are useful, so about $2^t$ tests on the same verifier data find an update at polynomial sample cost.

\item \textit{Noisy regime.} Bad samples are too frequent for any batch guarantee, so the finder proposes individual Bell difference samples to the verifier, giving an \(O(4^t)\) overhead in the worst case.
\end{itemize}

These three situations lead to Theorems~\ref{thm:informal-clean-depo}--\ref{informal-arbi}.

\subsubsection{Informal Results}
We report our three results for depolarizing noise to leading order in $1\ll t\ll n$, and $\varepsilon\ll1$ for the clean regime, with $t$ the supplied bound $\nu\leqslant t$. The constant $0.159 \approx 1-2^{-1/4}$ is where $(1-\lambda)^4=\frac12$: beyond it, our bounds on the first batch allow bad samples to outnumber good ones. The linear clean-regime threshold $\frac{n}{4t}\varepsilon$ is the weak-noise ($\lambda\ll1$) form of the formal condition. The formal results in Sec.~\ref{sec:up-down-depo} retain the finite-size and accuracy margins.

In the case of global depolarizing noise, when the noise $\lambda$ is sufficiently small compared to accuracy $\varepsilon$, the problem can be solved efficiently in both sample and time complexity:
\begin{theorem}[Clean Regime, Informal]
\label{thm:informal-clean-depo}
If $\lambda < \min\{0.159, \frac{n}{4t} \varepsilon\}$, there is an algorithm that solves Problem~\ref{prob:main} in $\mathrm{poly}(\cdot)$ time and samples.
\end{theorem}
Here, $\mathrm{poly}(\cdot)$ denotes $\mathrm{poly}\left(n, t, \frac{1}{\varepsilon}, \ln \frac{1}{\delta} \right)$. Explicit resource bounds for the depolarizing schedules are given in Eq.~\eqref{eq:explicit-learning-costs} and Appendix~\ref{app:verifier-complexity}. 
When the noise is bounded by a constant independent of $\varepsilon$, we can still solve the problem with a $2^t$ factor in the time complexity.
\begin{theorem}[Intermediate Regime, Informal]
\label{thm:informal-intermediate-depo}
If $\lambda < 0.159$, there is an algorithm that solves Problem~\ref{prob:main} with  $\mathrm{poly}(\cdot)$ samples and $O\left(2^t \cdot \mathrm{poly}(\cdot)\right)$ time. 
\end{theorem}
Lastly, we can always solve the problem with time and sample complexity exponential in $t$:
\begin{theorem}[Noisy Regime, Informal]
\label{thm:informal-noisy-depo}
For any $\lambda<1$, there is an algorithm that gives an $\varepsilon$-approximate weak stabilizer with $O\left(4^t \cdot \mathrm{poly}\left(\cdot, \frac{1}{1-\lambda} \right)\right)$ time and samples, with success probability at least $1 - \delta$.     
\end{theorem}

\begin{figure}[ht]
  \centering
  \includegraphics{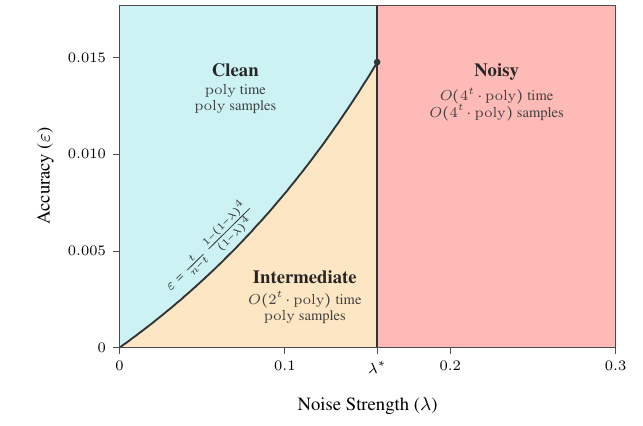}
  \caption{
   Leading-order illustration of the three regimes under global depolarizing noise, as a function of noise strength $\lambda$ and accuracy $\varepsilon$, for $n=2000$ and $t=30$. In the clean regime both time and samples are polynomial; in the intermediate regime samples stay polynomial but time carries a $2^t$ factor. The boundary between the two depends on $\varepsilon$, $\lambda$, $t$ and $n$, and to leading order in $1\ll t\ll n$ and weak noise $\lambda\ll1$ lies at $\lambda\approx n\varepsilon/4t$. Beyond the plotted boundary $\lambda^\star=1-\left(\frac{n+t}{2n}\right)^{1/4}\approx0.156$ only the noisy regime applies, with time and samples exponential in $t$. This threshold approaches $1-2^{-1/4}\approx0.159$ as $t/n\to0$. The boundaries are sufficient conditions for the analyzed algorithms, not lower bounds. Section~\ref{sec:numerics} compares them with simulations (Figs.~\ref{fig:numerics-threshold} and~\ref{fig:numerics-regimes}).}
  \label{fig:phase-diagram}
\end{figure}

These three algorithms partition the parameter space into three regimes, each with a different best-performing algorithm among those we analyze; we make no claim of optimality within these regimes, and the boundaries are sufficient conditions from our analysis rather than sharp algorithmic limits. Figure~\ref{fig:phase-diagram} gives an example. 

For arbitrary Pauli noise, with $\lambda$ the total error probability of Sec.~\ref{sec:noise-model}, the thresholds on $\lambda$ decrease:
\begin{theorem}[Arbitrary Pauli Noise, Informal]
\label{informal-arbi}
There is an algorithm that solves Problem~\ref{prob:main} in different regimes.
\begin{description}[font=\normalfont\bfseries\itshape,leftmargin=2em]
    \item[\textrm{Clean}] When $\lambda \leqslant \Theta\left(\min\left\{\frac{1}{t}, \frac{\varepsilon}{t \min\left\{t, \ln \frac{1}{\delta}\right\}}\right\}\right)$, the algorithm uses $\mathrm{poly}(\cdot)$ time and samples.
    \item[Intermediate] When $\lambda \leqslant \Theta\left(\frac{1}{t}\right)$, the algorithm uses $\mathrm{poly}(\cdot)$ samples and $O\left(2^t \cdot \mathrm{poly}(\cdot)\right)$ time.
    \item[Noisy] When $\lambda<\frac12$, the algorithm uses $O\left(4^t \cdot \mathrm{poly}\left(\cdot, \frac{1}{1-2\lambda} \right)\right)$ time and samples.
\end{description}

\end{theorem}

\subsubsection{Discussion}
\label{sec:comparison}

In considering the practical challenge brought by noise, we focused on the preserved weak symmetries of the noisy output state, rather than the strong stabilizer symmetries present for pure states~\cite{grewal2025, hangleiter2024} or considered in other work on stabilizer learning for mixed states~\cite{grewal2025, hinsche2026, cho2026}. Our work fills the gap between the polynomial-in-$t$ requirements of pure-state stabilizer structure learners and the exponential-in-$t$ cost of agnostic learning algorithms~\cite{chen2025}. The noise strength in both the global depolarizing and Pauli-channel settings defines a path between these two ends.

\begin{table}[ht]
\centering\small
\renewcommand{\arraystretch}{1.2}
\begin{tabular}{@{}llll@{}}
\toprule
Setting & Input promise & Output & Cost in $t$ \\
\midrule
Pure structure~\cite{grewal2025, hangleiter2024} & pure $\psi$, $\nu(\psi)\leqslant t$ & $\approx\mathrm{Stab}(\psi)$ & poly \\
Mixed, exact promise~\cite{grewal2025} & $\sigma_i\rho=\pm\rho$ on $2^{n-t}$ Paulis & $\approx$ strong stabilizers & poly \\
Agnostic~\cite{chen2025} & any $\rho$, fidelity $\geqslant\tau$ & best-fit nullity-$t$ state & $(2^t/\tau)^{O(\log 1/\varepsilon)}$ \\
\midrule
\multicolumn{4}{@{}l}{\textbf{This work}, $\rho=\mathcal{E}(\psi)$:} \\
\quad clean & \begin{tabular}[t]{@{}l@{}}$\lambda\lesssim\min\{0.159,\tfrac{n\varepsilon}{4t}\}$\\ (Pauli: $O(\varepsilon/(t\min\{t,\ln 1/\delta\}))$)\end{tabular} & $\tilde S\supseteq\mathrm{Stab}(\psi)$ & poly \\
\quad intermediate & $\lambda\lesssim 0.159$ \ (Pauli: $O(1/t)$) & same & poly copies, $2^t$ time \\
\quad noisy & $\lambda<1$ \ (Pauli: $<\tfrac12$) & same & $4^t\gamma^{-2}$ \\
\bottomrule
\end{tabular}
\caption{Learning stabilizer structure: prior guarantees and our three regimes (depolarizing thresholds; Pauli in parentheses). Accuracy metrics differ between rows (trace distance, fidelity loss, projected-purity loss), so costs are compared only in their $t$-dependence; exact statements are in the cited theorems and Eq.~\eqref{eq:explicit-learning-costs}.}
\label{tab:comparison}
\end{table}

Table~\ref{tab:comparison} places our guarantees beside the three closest settings: pure-state structure learning, mixed-state structure learning under an exact promise, and agnostic tomography. Exact-promise learners are efficient but apply only at $\lambda=0$, because depolarizing noise removes every strong stabilizer. Agnostic learners apply at every $\lambda$, since a depolarized state keeps fidelity $\tau\approx1-\lambda$ with $\psi$, but their cost is exponential in $t$ even at $\lambda=0$. In exchange, they need neither the pure pre-noise state nor the Pauli-noise model that our guarantees use.

As noise increases, our guarantees turn exponential in $t$ only at a finite noise strength, and in two stages. Under depolarizing noise, classical runtime becomes exponential first, once $\lambda\gtrapprox\frac{n}{4t}\varepsilon$ and the intermediate regime takes over with its $2^t$ overhead. The number of copies stays polynomial up to a constant $\lambda\approx0.159$, a total error probability of about $16\%$ per copy, and only past it does the noisy regime require $4^t$ scaling. For $n\gg t$, the intermediate regime is squeezed into very low $\varepsilon\lessapprox t/n$, and at coarser accuracy the clean regime reaches the constant threshold; in Fig.~\ref{fig:phase-diagram}, where $n=2000$ and $t=30$, the three regimes meet at $\varepsilon\approx0.015$. Arbitrary Pauli noise shows the same three regimes, with the constant threshold shrunk to order $1/t$.

\subsection{Organization}
\label{sec:organization}
Section~\ref{sec:preliminaries} collects notation and known facts about Bell sampling and noise-free stabilizer learning. Section~\ref{sec:tools} proves the properties of the projected purity and of Bell difference sampling on which our algorithms rest, including the purity--escape bound, Lemma~\ref{lemma:belldiff-and-gamma}. Section~\ref{sec:stab-ver} constructs the verifier, and Section~\ref{sec:learning-algorithms} gives the two finders. The top-down finder starts from $\tilde S=\{I\}$ and adds one verified Pauli at a time, so $C(\tilde S)$ shrinks from all of $\mathcal{P}^n$ toward $C(S)$; the pure-state learners of Sec.~\ref{eq:introduction} instead work bottom-up, growing a span toward $C(S)$; the up-down finder builds each batch span bottom-up and applies the resulting updates top-down through the verifier. Section~\ref{sec:up-down-per-batch} analyzes a single batch, including the radical bounds above, and Section~\ref{sec:end-to-end} chains batches into the guarantees of Theorems~\ref{thm:informal-clean-depo}--\ref{informal-arbi}. Section~\ref{sec:numerics} tests the first-batch thresholds numerically, including a sample-based verifier, entangled inputs and two-qubit gate noise, and runs complete intermediate-mode trajectories under gate noise. Table~\ref{tab:reading-guide} lists where each main statement is proved and what it builds on.

\begin{table}[ht]
    \centering
    \begin{tabular}{l l l l}
    \toprule
    Topic & Read & Main statements & Builds on \\
    \midrule
    Purity and Bell difference samples & Sec.~\ref{sec:tools} & Lemmas~\ref{lemma:prob-diff-lower}, \ref{lemma:belldiff-and-gamma}, \ref{lemma:compatibility} & Preliminaries \\
    The verifier & Sec.~\ref{sec:stab-ver} & Algorithm~\ref{alg:stab-ver}, Theorem~\ref{thm:stab-ver} & Lemma~\ref{lemma:projected-purity}; Lemma~\ref{lemma:uniform-bell-estimates} \\
    Top-down finder (noisy regime) & Sec.~\ref{sec:top-down} & Algorithm~\ref{alg:top-down}, Theorem~\ref{thm:top-down} & The verifier; Lemma~\ref{lemma:prob-diff-lower} \\
    Up-down finder and its $\mathsf{propose}$ rules & Secs.~\ref{sec:up-down}, \ref{sec:up-down-example} & Algorithms~\ref{alg:stab}, \ref{alg:propose-clean}, \ref{alg:propose-intermediate} & The verifier \\
    What one batch achieves & Sec.~\ref{sec:up-down-per-batch} & Lemmas~\ref{lemma:per-batch-intermediate}, \ref{lemma:per-batch-clean} & Lemmas~\ref{lemma:good-span-S}, \ref{lemma:good-span-C(S)}, \ref{lemma:belldiff-and-gamma} \\
    From batches to the full algorithm & Secs.~\ref{sec:batch-to-iter}, \ref{sec:success-modes} & Observations~\ref{obs:per-iter-intermediate}--\ref{obs:per-iter-clean-terminal}, Algorithm~\ref{alg:stab-clean} & Per-batch lemmas \\
    Guarantees, depolarizing noise & Sec.~\ref{sec:up-down-depo} & Theorems~\ref{thm:main-intermediate-depo}, \ref{thm:main-clean-depo} & Per-iteration observations \\
    Guarantees, arbitrary Pauli noise & Sec.~\ref{sec:up-down-arbi} & Theorems~\ref{thm:main-intermediate-arbi}, \ref{thm:main-clean-arbi} & Per-iteration observations \\
    \bottomrule
    \end{tabular}
    \caption{Reading guide and overview of the paper structure.}
    \label{tab:reading-guide}
\end{table}

\section{Preliminaries}
\label{sec:preliminaries}

\subsection{Paulis and Cliffords}
\label{sec:paulis-cliffords}

We work with $n$-qubit systems throughout. Let $\mathcal{P}^n$ denote the $n$-qubit Pauli group modulo global phases. Each of its $4^n$ elements is a tensor product of the single-qubit Pauli matrices $\{I, X, Y, Z\}$. We identify elements of $\mathcal{P}^n$ with binary-symplectic labels $i = (i_X, i_Z) \in \{0,1\}^{2n}$ and write the corresponding Pauli as $\sigma_i$. We fix the sign convention so that $\sigma_i \sigma_j \propto \sigma_{i+j}$ under addition in $\mathbb{F}_2^{2n}$. When it does not cause confusion, we use $i$ interchangeably for the label and the Pauli itself.

Two Paulis either commute or anticommute, as captured by the symplectic inner product
\begin{equation}
\label{eq:symplectic-form}
\expval{i, j} \equiv i_X \cdot j_Z + i_Z \cdot j_X \pmod{2},
\end{equation}
via the rule $\sigma_i \sigma_j = (-1)^{\expval{i,j}} \sigma_j \sigma_i$. In Bell sampling, the $Y$-parity $\pi_Y(i) = i_X \cdot i_Z \bmod 2$ commonly arises because $Y$ is the only Pauli not equal to its transpose. We summarize the four identities that combine $\pi_Y$ and $\expval{\cdot,\cdot}$ in Appendix~\ref{app:bell-sampling} (Lemma~\ref{lemma:useful-properties}).

For a subspace $H \subseteq \mathbb{F}_2^{2n}$, its centralizer $C(H) \equiv \{j : \expval{i,j} = 0 \ \forall i \in H\}$ is its symplectic complement, so $\dim C(H) = 2n - \dim H$ and $C(C(H)) = H$. A subspace $H$ is commuting if and only if $H \subseteq C(H)$, so a commuting subspace has dimension at most $n$. The radical of $H$,
\begin{equation}
\label{eq:radical}
\rad(H) \equiv H \cap C(H),
\end{equation}
consists of the elements of $H$ that commute with all of $H$; it is always commuting.

Any density matrix expands uniquely in the Pauli basis,
\begin{equation}
\label{eq:pauli-expansion}
\rho = \sum_{i \in \{0,1\}^{2n}} a_i\, \sigma_i,\qquad a_i \equiv \frac{1}{2^n}\tr(\rho\, \sigma_i).
\end{equation}
Trace normalization pins $a_0 = 2^{-n}$, and the purity is
\begin{equation}
\label{eq:purity-in-a}
\gamma \equiv \tr(\rho^2) = 2^n \sum_{i \in \{0,1\}^{2n}} a_i^2.
\end{equation}

The Clifford group $\mathcal{C}_n$ is the normalizer of $\mathcal{P}^n$ in the unitary group: $C \in \mathcal{C}_n$ iff $C\, \sigma_i\, C^\dagger = \pm\, \sigma_{\phi(i)}$ for some $\mathbb{F}_2$-linear symplectic permutation $\phi$ of $\{0,1\}^{2n}$. That is, Cliffords map Paulis to Paulis under conjugation. They therefore permute $\{a_i\}$ and map $\supp(\rho)$ and $S$ to their images under $\phi$, leaving their dimensions, and hence the nullity $\nu$, invariant.

\subsection{Bell Sampling}
\label{sec:bell-sampling}

Both subroutines of our algorithm are built on one measurement primitive, Bell sampling~\cite{montanaro2017}, in which two copies of $\rho$ are measured jointly in the $n$-qubit Bell basis. The verifier uses Bell samples to estimate purities, and the finder uses differences of Bell samples to reveal the Pauli support. The $n$-qubit Bell basis is

\begin{equation}
\label{eq:n-qubit-bell}
\ket{\sigma_i} \equiv (\sigma_i \otimes I)\, \ket{\Phi^+_n},\qquad \ket{\Phi^+_n} \equiv \frac{1}{\sqrt{2^n}} \sum_{x \in \{0,1\}^n} \ket{x}\ket{x},
\end{equation}
an orthonormal basis of $2n$ qubits labeled by the same binary-symplectic strings $i \in \{0,1\}^{2n}$ that index the Paulis themselves. Because distinct Paulis $\sigma_i \neq \pm\sigma_j$ give orthogonal Bell states, each outcome of the measurement is a single label $i$, which we call a \emph{Bell sample}. Figure~\ref{fig:bell-measurement} shows the measurement.

\begin{figure}[ht]
\centering
\includegraphics{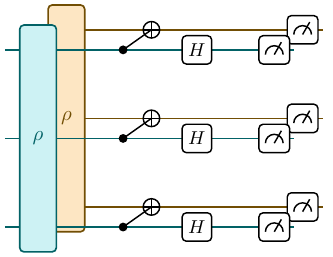}
\caption{The Bell sampling protocol. Two copies of $\rho$ are measured jointly in the transversal $n$-qubit Bell basis: a CNOT and a Hadamard act on each pair of corresponding qubits, followed by a computational-basis measurement. The $2n$ measured bits form the Bell-sample label $i$.}
\label{fig:bell-measurement}
\end{figure}

The outcome distribution is the diagonal of $\rho \otimes \rho$ in this basis. Expanding both copies of $\rho$ in Paulis collapses the double sum over $(j,k)$ to a single sum over $j$ (derivation in Appendix~\ref{app:bell-sampling}): 

\begin{equation}\label{eq:p-bell}
  P_{\mathrm{Bell}}\left( i \right) = \tr \left( \ket{\sigma_i}\bra{\sigma_i}\left( \rho \otimes \rho \right) \right)
  = \sum_{j \in \left\{ 0,1 \right\}^{2n}} (-1)^{\pi_Y(j) + \expval{i,j}}\, a_j^2 .
\end{equation}
A single Bell sample probes the \emph{quadratic} weights $\{a_j^2\}$ rather than the raw Pauli amplitudes $\{a_j\}$. The signs $(-1)^{\pi_Y(j) + \expval{i,j}}$ encode the interference between different Paulis in $\supp(\rho)$ and are the reason $P_{\mathrm{Bell}}$ carries stabilizer structure.

An immediate corollary is that the purity of $\rho$ can be computed as an expectation over Bell samples:

\begin{lemma}
\label{lemma:bell-sampling-known-quantites}
The purity of $\rho$ can be estimated from Bell samples via
\begin{equation}
\label{eq:purity-from-bell}
\gamma = \expval{(-1)^{\pi_Y(i)}}_{i \sim \mathrm{Bell}}.
\end{equation}
\end{lemma}
\noindent We defer the proof, which is a special case of Lemma~\ref{lemma:projected-purity} in Sec.~\ref{sec:stab-ver}.

As $(-1)^{\pi_Y(i)} \in \{-1, +1\}$ is bounded, Hoeffding gives $M = O(\varepsilon^{-2} \ln(1/\delta))$ Bell samples as sufficient to estimate $\gamma$ to additive error $\varepsilon$ with confidence $1 - \delta$.

Beyond purity estimation, Bell samples carry a structural constraint, which Bell difference sampling exploits. In the pure-state setting, $P_{\mathrm{Bell}}$ is supported on a shift of the Pauli support:

\begin{proposition}[Hidden-shift structure of Bell samples, \cite{montanaro2017, hangleiter2024}]
\label{prop:hidden-shift}
For a pure state $\psi$, there exists a hidden shift $k \in \{0,1\}^{2n}$ (depending on $\psi$) such that
\begin{equation}
\label{eq:hidden-shift}
P_{\mathrm{Bell}}(i) \neq 0 \;\Longrightarrow\; i \oplus k \in \supp(\psi).
\end{equation}
\end{proposition}
\begin{proof}
See Refs.~\cite{montanaro2017, hangleiter2024}.
\end{proof}

\subsection{Bell Difference Sampling}
\label{sec:bell-difference-sampling}

XORing two i.i.d.\ Bell samples $i_1, i_2 \sim P_{\mathrm{Bell}}$ cancels the hidden shift of Proposition~\ref{prop:hidden-shift}, so for pure states the difference sample
\begin{equation}
\label{eq:bell-diff-sample}
b \;\equiv\; i_1 \oplus i_2
\end{equation}
lies in the support span $C(S)$ (Proposition~\ref{prop:pure-state-bell-diff} below).
 Its distribution $P_{\mathrm{BellDiff}}$ is the convolution of $P_{\mathrm{Bell}}$ with itself; substituting Eq.~\eqref{eq:p-bell} and collapsing the inner sum via Property~3 of Lemma~\ref{lemma:useful-properties} gives a clean closed form:
\begin{equation}
\label{eq:bell-diff-dist}
P_{\text{BellDiff}}(i)  = \sum_{j \in \mathcal{P}^n} {P_{\mathrm{Bell}}\left( {j} \right)P_{\mathrm{Bell}}\left( {i + j} \right)}    = 4^n \sum_{j \in \mathcal{P}^n} (-1)^{\expval{i,j}} a_j^4.
\end{equation}
The full derivation is in Appendix~\ref{app:bell-sampling}.

Structurally, Eq.~\eqref{eq:bell-diff-dist} is the binary-symplectic Fourier transform of the quartic Pauli weights $\{a_j^4\}$. Where Bell sampling probes the quadratic structure $\{a_j^2\}$, Bell difference sampling probes the $\ell_4$ structure of the Pauli expansion. This shift from $\ell_2$ to $\ell_4$ is what makes the distribution sensitive to the finer group structure our algorithm exploits.

An immediate consequence of the Fourier form is that translating by any element of the weak stabilizer leaves $P_{\mathrm{BellDiff}}$ unchanged:

\begin{observation}
\label{obs:stab-inv-diff}
The Bell difference sampling distribution is invariant under translations by $S$: for any $x \in S, y \in \{0,1\}^{2n}$, $P_{\mathrm{BellDiff}}(y) = P_{\mathrm{BellDiff}}(x + y)$.
\end{observation}

\subsection{Pure State (Noise-Free) Stabilizer Learning}
\label{sec:pure-state-learning}

For a pure state $\psi$, the Pauli expansion is highly constrained: every Pauli with a nonzero coefficient is an element of the centralizer of the stabilizer group.

\begin{observation}
\label{obs:pure-state-stab}
For a pure state $\psi$, $a_i \neq 0$ only if $i \in C(S)$.
\end{observation}

Bell difference sampling inherits this constraint:

\begin{proposition}
\label{prop:pure-state-bell-diff}
For a pure state $\psi$, $P_{\mathrm{BellDiff}}(i) \neq 0$ only if $i \in C(S)$.
\end{proposition}
\begin{proof}
See Appendix~\ref{app:bell-sampling}.
\end{proof}

Proposition~\ref{prop:pure-state-bell-diff} is the core of noise-free stabilizer learning~\cite{montanaro2017, grewal2024, hangleiter2024}. Every Bell difference sample $b_j$ lies in the support span $C(S)$, and the support of $P_{\mathrm{BellDiff}}$ spans $C(S)$ as an $\mathbb{F}_2$-subspace of $\{0,1\}^{2n}$. Enough samples therefore span $C(S)$ itself, and taking its symplectic complement recovers $S = C(C(S)) = \mathrm{Stab}(\psi)$.

However, exact recovery may require observing very rare samples. We instead allow an $\varepsilon$-approximate weak stabilizer, where sufficiently small components remain unresolved.

\begin{algorithm}[ht]
  \caption{\label{alg:noise-free-learning}Noise-free stabilizer learning~\cite{hangleiter2024}}
  \KwIn{$M$ independent Bell difference samples $b_1, \ldots, b_M \sim P_{\mathrm{BellDiff}}$ from a pure low-nullity state $\psi$.}
  \KwOut{Weak stabilizer estimate $\tilde{S} \leqslant \mathcal{P}^n$.}
  Compute the $\mathbb{F}_2$-span $\hat{\mathcal{G}} = \spn\{b_1, \ldots, b_M\}$\;
  \Return{$\tilde{S} = \rad(\hat{\mathcal{G}}) = \hat{\mathcal{G}} \cap C(\hat{\mathcal{G}})$}\;
\end{algorithm}

Algorithm~\ref{alg:noise-free-learning} gives an $\varepsilon$-approximate weak stabilizer group with $O\left(n + \frac{t+\log \frac{1}{\delta}}{\varepsilon}\right)$ with failure probability at most $\delta$. The proof is deferred as Corollary~\ref{lemma:noise-free-samples} in Sec.~\ref{sec:good-samples}.

With access to $\tilde S$, measuring its commuting generators on additional copies fixes their signs, which Bell data leaves open, and projection onto the most likely joint eigenspace introduces trace-distance error at most $\sqrt{\varepsilon}$ by Eq.~\eqref{eq:purity-trace-distance}. Full-state tomography on the compressed $\tilde\nu = n - \dim\tilde S \leqslant t$-qubit subsystem then yields an $O(\sqrt{\varepsilon})$-accurate approximation to $\psi$ with $O(2^{\tilde\nu} / \varepsilon)$ additional measurements~\cite{grewal2025, hangleiter2024, chia2024}.

Under noise, Proposition~\ref{prop:pure-state-bell-diff} fails: even arbitrarily weak white noise places Bell difference samples outside $C(S)$, so Algorithm~\ref{alg:noise-free-learning} returns $\{I\}$ for large $M$ (Sec.~\ref{eq:introduction}). We therefore target the $\varepsilon$-approximate weak stabilizer of Sec.~\ref{sec:approx-support}.

\subsection{Asymptotic notation}
\label{sec:asymptotic-notation}

Most of our statements are presented to leading order to reduce their mathematical complexity. We read the symbols as follows throughout unless a local regime is specified. Each statement comes with conditions of the form $A\ll B$ or $A\gg B$, such as $n-t\gg1$, $t\gg\ln\frac1\delta$ or $\varepsilon\ll1$. These define the limit in which it holds: the ratio $A/B$ tends to $0$, or to $\infty$. The approximate comparisons $\approx$, $\lessapprox$ and $\gtrapprox$ keep the leading terms and suppress relative corrections that are small in this limit: $A\approx B$ means $A=B\,(1+o(1))$, and $A\lessapprox B$ means $A\leqslant B\,(1+o(1))$, with $\gtrapprox$ analogous. When a ratio is only fixed and small rather than vanishing, these comparisons hold up to a relative error of that size. A sufficient condition stated with $\lessapprox$ needs enough margin to absorb the suppressed corrections. By contrast, $O(\cdot)$, $\Omega(\cdot)$ and $\Theta(\cdot)$ are constant-factor statements with untracked constants, as in the theorems for arbitrary Pauli noise, and $o(1)$ is a quantity that vanishes in the limit. Where exact finite bounds are needed we state them separately, as in Eq.~\eqref{eq:good-span-exact-bound}; elsewhere we stay at leading order to keep the structure visible. Auxiliary choices such as target dimensions, per-update budgets and batch sizes stay symbolic for the same reason; for the clean mode under depolarizing noise, Appendix~\ref{app:clean-parameters} fixes one explicit set of them, with a finite check under which it provably works.

\section{Projected Purity and Bell Difference Sampling}
\label{sec:tools}

Section~\ref{sec:summary-of-results} defined the Pauli twirl $\mathcal T_{\tilde S}$, the projected purity $\gamma_{\tilde S}$ and the purity loss $\Delta\gamma(\tilde S)$ for a commuting Pauli group $\tilde S$. This section proves the properties of these quantities on which the algorithms rest. Section~\ref{sec:purity-norm} records elementary properties of the projected purity. Section~\ref{sec:belldiff-purity} connects the projected purity to Bell difference sampling; its main result, Lemma~\ref{lemma:belldiff-and-gamma}, relates the purity that the verifier certifies to the samples that the finder observes, and is proved in Appendix~\ref{app:density-matrix-property}. Section~\ref{sec:degenerate-support} shows when the weak stabilizer is commuting and why candidate groups with small purity loss cannot conflict with one another. 

\subsection{Projected Purity as a Norm}
\label{sec:purity-norm}

In terms of the Pauli expansion of Eq.~\eqref{eq:pauli-expansion}, the twirl keeps exactly the components $\sigma_i$ with $i \in C(\tilde S)$: conjugating $\sigma_i$ by $g \in \tilde S$ gives $(-1)^{\expval{g,i}}\sigma_i$, and averaging over $\tilde S$ kills every $\sigma_i$ that anticommutes with some element of $\tilde S$. Since $\{\sigma_i/\sqrt{2^n}\}$ is an orthonormal basis for the Hilbert--Schmidt inner product $\tr(A^\dagger B)$, $\mathcal T_{\tilde S}$ is the orthogonal projection onto $\spn\{\sigma_i : i \in C(\tilde S)\}$, the operators commuting with all of $\tilde S$. The projected purity and the purity loss are
\begin{equation}
\label{eq:purity-as-norm}
\gamma_{\tilde S} = \left\| \mathcal T_{\tilde S}(\rho) \right\|_{\mathrm{HS}}^2 = 2^n\sum_{i\in C(\tilde S)}a_i^2,
\qquad
\Delta\gamma(\tilde S) = \left\| \rho - \mathcal T_{\tilde S}(\rho) \right\|_{\mathrm{HS}}^2 = 2^n\sum_{i\notin C(\tilde S)}a_i^2 .
\end{equation}
Here $\gamma_{\tilde S}$ is the squared length of the component of $\rho$ compatible with $\tilde S$ and $\Delta\gamma(\tilde S)$ is the squared distance from $\rho$ to the nearest operator compatible with $\tilde S$.

In the trivial case $\tilde S = \{I\}$ nothing is filtered and $\gamma_{\{I\}} = \gamma$. The twirl and $\gamma_H$ are defined by the same formulas for any Pauli label subgroup $H$, including noncommuting ones. Later, we will use commutativity for the joint eigenspace description (Appendix~\ref{app:fourier}) and for the verifier's product estimator (Lemma~\ref{lemma:projected-purity}).

Two properties of the projected purity will be useful. First, bigger groups filter more, since $H \subseteq H'$ gives $C(H') \subseteq C(H)$:

\begin{observation}
\label{obs:subgroup-purity}
For groups $H, H'$ with $H \subseteq H'$, $\gamma_{H'} \leqslant \gamma_H$.
\end{observation}

\noindent Second, adding the true weak stabilizer $S = C(\supp(\rho))$ to any group costs no purity, since every $i$ with $a_i \neq 0$ lies in $C(S)$ by definition:

\begin{observation}
\label{obs:add-S-purity}
For any group $H \leqslant \mathcal{P}^n$, $\gamma_{H + S} = \gamma_H$.
\end{observation}

\subsection{Bell Difference Samples and Purity}
\label{sec:belldiff-purity}

The verifier certifies projected purities, while the finder observes only where Bell difference samples fall. The two lemmas of this subsection connect these quantities. Bell difference samples also fall in the centralizer of any commuting subgroup $H$ with probability at least the squared projected purity of $H$:

\begin{lemma}
\label{lemma:prob-diff-lower}
For a commuting subgroup $H \leqslant \mathcal{P}^n$, $P_{\mathrm{Bell Diff}}(C(H)) \equiv \Pr_{b \sim \mathrm{Bell Diff}}[b \in C(H)] \geqslant \gamma_H^2$.
\end{lemma}
\begin{proof}
See Proposition~\ref{prop:belldiff-lower-bound} in Appendix~\ref{app:fourier}.
\end{proof}

The next lemma refines this bound. For $H \subseteq H'$, it bounds the probability that a Bell difference sample in $C(H)$ also lies in $C(H')$, in terms of the purity drop $\gamma_H - \gamma_{H'}$. The finder relies on it each time it enlarges a candidate group:

\begin{lemma}[Purity--escape bound]
\label{lemma:belldiff-and-gamma}
Let $H$ be a commuting Pauli subgroup and let $H \subseteq H' \subseteq C(H)$, where $H'$ need not be commuting. Write $\Delta\gamma = \gamma_H - \gamma_{H'}$. For every $0 \leqslant \Delta\varepsilon \leqslant \Delta\gamma$,
\begin{equation}
\label{eq:exact-escape-bound}
\mathbb{P}_{b \sim \mathrm{BellDiff}}[b \in C(H') \mid b \in C(H)]
\leqslant 1 - 2\Delta\varepsilon + \frac{2(\Delta\varepsilon)^2}{\gamma_H}.
\end{equation}
There is no small-loss assumption on $\Delta\gamma$.
\end{lemma}
\begin{proof}
The proof uses the eigenspaces of $H$ to compare Bell difference probabilities with projected purity, and is deferred to Appendix~\ref{app:density-matrix-property}, which also shows that the leading constant $2$ is sharp.
\end{proof}

\subsection{Degenerate Supports and No Wrong Turns}
\label{sec:degenerate-support}

Up to this point it has been implied that $S = C(\supp(\rho))$ is a commuting group, so that $\nu = n - \dim S$ is a well-defined nullity. This is automatic for pure states, but need not hold for mixed states. Consider the maximally mixed state $I/2^n$: its support is $\{I\}$, whose centralizer is $\mathcal{P}^n$---a noncommuting group with $\mathbb{F}_2$-dimension $2n$ and $\nu = -n$. The projected purity $\gamma_S$ remains well defined, but $S$ can no longer serve as a commuting stabilizer target.

This failure has a specific form: $S$ is noncommuting only when, in some Clifford frame, one qubit of $\rho$ is maximally mixed and uncorrelated with the remaining qubits; we say such a qubit is \emph{lost}. The lemma below states this quantitatively. If two anticommuting Paulis each nearly commute with $\rho$, in the sense that twirling over either one removes little purity, then the purity of $\rho$ is at most $\frac12$ plus the purity removed by the two twirls. The same statement controls the approximate version we need, in which two low-loss candidate groups fail to commute.

\begin{lemma}[Lost qubit]
\label{lemma:lost-qubit}
Let $g, h \in \mathcal{P}^n$ anticommute, and suppose $2^n \sum_{i \notin C(g)} a_i^2 \leqslant \varepsilon_1$ and $2^n \sum_{i \notin C(h)} a_i^2 \leqslant \varepsilon_2$. Then the purity satisfies $\gamma = \tr(\rho^2) \leqslant \frac12 + \varepsilon_1 + \varepsilon_2$.
\end{lemma}
\begin{proof}
We split $\rho$ into a component commuting with both $g$ and $h$, which carries purity at most $\frac12$, and a remainder, which carries at most $\varepsilon_1+\varepsilon_2$.
Define $\rho_\parallel = \sum_{i \in C(g) \cap C(h)} a_i \sigma_i$. By Eq.~\eqref{eq:purity-in-a} and a union bound over the excluded labels,
\begin{equation}
\gamma - \|\rho_\parallel\|_{\mathrm{HS}}^2 = 2^n \sum_{i \notin C(g) \cap C(h)} a_i^2 \leqslant \varepsilon_1 + \varepsilon_2 .
\end{equation}
Now choose a Clifford $U$ such that $U \sigma_g U^\dagger = X_1$ and $U \sigma_h U^\dagger = Z_1$. Cliffords preserve commutation, so $U \rho_\parallel U^\dagger$ is exactly the component of $U \rho U^\dagger$ that commutes with both $X_1$ and $Z_1$, i.e.\ acts as the identity on qubit $1$:
\begin{equation}
U \rho_\parallel U^\dagger = \frac{I}{2} \otimes \tr_1\!\left(U \rho U^\dagger\right).
\end{equation}
The purity of this state is $\frac12 \tr[(\tr_1 U \rho U^\dagger)^2] \leqslant \frac12$. The Hilbert--Schmidt norm is unitarily invariant, so $\|\rho_\parallel\|_{\mathrm{HS}}^2 \leqslant \frac12$. Hence $\gamma \leqslant \frac12 + \varepsilon_1 + \varepsilon_2$.
\end{proof}

Two cases will be important. First, when $\varepsilon_1 = \varepsilon_2 = 0$, if $S$ contained an anticommuting pair $g, h$, both would commute with the entire support, so $\gamma \leqslant \frac12$. In fact, $S$ is noncommuting exactly when $\rho$ is Clifford-equivalent to $\frac{I}{2} \otimes \rho'$, a state with one maximally mixed qubit. Second, when $\varepsilon_1 = \varepsilon_2 = \varepsilon$, if $S_1, S_2$ are commuting groups with purity loss at most $\varepsilon$, and $g \in S_1$, $h \in S_2$ anticommute, then $C(S_1) \subseteq C(g)$ and $C(S_2) \subseteq C(h)$, so the hypotheses hold and $\gamma \leqslant \frac12 + 2\varepsilon$. Intuitively, $\rho$ is approximately invariant under two anticommuting Paulis. In the Clifford frame of the proof they become $X_1$ and $Z_1$, and the only operators commuting with both act as the identity on qubit $1$. So $\rho$ is close to a state carrying an identity factor on that logical qubit:
\begin{equation}
U\rho U^\dagger\approx\frac{I}{2}\otimes\rho_{\mathrm{rest}},
\qquad
\Bigl\|U\rho U^\dagger-\frac{I}{2}\otimes\rho_{\mathrm{rest}}\Bigr\|_{\mathrm{HS}}^2\leqslant2\varepsilon,
\end{equation}
with $\rho_{\mathrm{rest}}=\tr_1(U\rho U^\dagger)$. Thus, if two low-loss candidate groups fail to commute, then in some Clifford frame one qubit of $\rho$ is close to maximally mixed and uncorrelated with the remaining qubits.

For the noise models of Sec.~\ref{sec:noise-model} the target is well behaved. By Eq.~\eqref{eq:pauli-fidelities}, no Pauli coefficient is rescaled to zero when $\lambda<\frac12$, or for any $\lambda<1$ under depolarizing noise, so $S=\mathrm{Stab}(\psi)$ exactly.

What the finder analyses need is that low-loss candidates never conflict. We secure this through the accuracy budget.

\begin{assumption}[Accuracy budget]
\label{asp:accuracy}
The total accuracy budget satisfies $\varepsilon<(1-2\lambda)^2/4$ for a Pauli channel with $\lambda<1/2$, and $\varepsilon<(1-\lambda)^2/4$ for depolarizing noise with $\lambda<1$.
\end{assumption}

This restricts only the requested accuracy. A larger $\varepsilon$ can always be lowered below the cap, which strengthens the output guarantee, and the lowered value is then used consistently in the update budgets, the verifier gap, and the cost bounds. The cap applies to the total purity loss accumulated over all accepted updates, not to each update separately.

\begin{lemma}[Compatibility]
\label{lemma:compatibility}
Under Assumption~\ref{asp:accuracy}, any internally commuting Pauli groups $S_1, S_2$ with $\gamma - \gamma_{S_1} \leqslant \varepsilon$ and $\gamma - \gamma_{S_2} \leqslant \varepsilon$ commute with each other.
\end{lemma}
\begin{proof}
Let $\Delta\gamma_\psi(H)$ denote the purity loss of $H$ in the underlying pure state $\psi$. By Eq.~\eqref{eq:pauli-fidelities}, a Pauli channel scales each coefficient by $f_i\geqslant1-2\lambda$, so $\Delta\gamma(H)\geqslant(1-2\lambda)^2\Delta\gamma_\psi(H)$. For depolarizing noise, $\Delta\gamma(H)=(1-\lambda)^2\Delta\gamma_\psi(H)$ exactly, since only the identity coefficient escapes the common factor and it lies in every $C(H)$. Either way, Assumption~\ref{asp:accuracy} turns a loss of at most $\varepsilon$ in $\rho$ into a loss below $\frac14$ in $\psi$. If $g\in S_1$ and $h\in S_2$ anticommuted, then $C(S_1)\subseteq C(g)$ and $C(S_2)\subseteq C(h)$, and Lemma~\ref{lemma:lost-qubit} applied to $\psi$ would give $1=\gamma_\psi<\frac12+\frac14+\frac14$, a contradiction.
\end{proof}

Lemma~\ref{lemma:compatibility}, not the numerical cap, is what the finder analyses generally use, and the cap is only a sufficient condition. Another sufficient condition is $\gamma>\frac12+2\varepsilon$, from the approximate reading of Lemma~\ref{lemma:lost-qubit} applied to $\rho$ directly; under depolarizing noise it holds whenever $\lambda<1-\sqrt{\tfrac12+2\varepsilon}$, whose right-hand side approaches $1-1/\sqrt{2}\approx0.293$ as $\varepsilon\to0$. More generally, compatibility can only fail through the nearly erased logical qubit described above. The states we target are noisy images of pure states in which every encoded qubit still carries information at resolution $\varepsilon$, so we expect compatibility to hold well beyond the cap. This is a physical expectation, not a theorem, which is why the cap is the formal assumption. From here on, the finder analyses and recovery guarantees use Assumption~\ref{asp:accuracy}, and hence Lemma~\ref{lemma:compatibility}, unless stated otherwise.

\section{Candidate Verification}
\label{sec:stab-ver}

The first step towards finding the weak stabilizer is the ability to \textit{verify} whether some candidate group $\tilde{S}$, or a proposed extension of it, stays within its purity-loss budget, as in $\Delta\gamma = \gamma - \gamma_{\tilde{S}} \leqslant \varepsilon$. It does not by itself certify $S\subseteq\tilde S$, which an $\varepsilon$-approximate weak stabilizer also requires; that containment is what the finder analyses supply.

We build the verifier in two steps. First, we show that the projected purity $\gamma_{\tilde S}$ is estimable from Bell samples (Lemma~\ref{lemma:projected-purity}) and turn this into a one-shot verifier, Algorithm~\ref{alg:sim-stab-ver}, to fix ideas. Sec.~\ref{sec:adaptive-ver} then upgrades it to the verifier the finders actually call, Algorithm~\ref{alg:stab-ver}: it estimates the \emph{relative} drop $\gamma_{\tilde S} - \gamma_{\tilde S + \Delta \tilde S}$ of a proposed extension and reuses one dataset across adaptive queries (Theorem~\ref{thm:stab-ver}). Sec.~\ref{sec:ver-discussion} collects the conventions that the finder analyses rely on.

In Sec.~\ref{sec:bell-sampling}, we already showed that $\gamma$ can be efficiently estimated. Ref.~\cite{hangleiter2024} gives Bell-sample estimators for single-Pauli overlaps $\tr(\sigma_g\rho\sigma_g\rho)$. For a commuting $\tilde S$, the projected purity is their average, $\gamma_{\tilde S}=\frac{1}{|\tilde S|}\sum_{g\in\tilde S}\tr(\sigma_g\rho\sigma_g\rho)$, and with a bit of engineering the averaged estimator becomes a filter that only checks the generators of $\tilde S$:
\begin{lemma}[Projected Purity]
\label{lemma:projected-purity}
For a commuting candidate group $\tilde{S} = \expval{g_1, g_2, \ldots, g_m}$, Bell sampling gives
\begin{equation}
\label{eq:projected-purity-est}
\gamma_{\tilde{S}} = \expval{(-1)^{\pi_Y(i)} \prod_{k=1}^{m} \left[(-1)^{\pi_Y(g_k) + \expval{g_k,i}} = 1\right]}_{i \sim \mathrm{Bell}}.
\end{equation}
\end{lemma}
The proof can be found in Appendix~\ref{app:bell-sampling}. As a justification, recall from Lemma~\ref{lemma:bell-sampling-known-quantites} that the purity $\gamma$ can be estimated via $\gamma = \expval{(-1)^{\pi_Y(i)}}_{i \sim \text{Bell}}$. Eq.~\eqref{eq:projected-purity-est} is the usual purity estimator with an additional filter: outcomes contribute their usual sign when all conditions are satisfied, and zero otherwise.

By $\gamma = \expval{(-1)^{\pi_Y(i)}}_{i \sim \text{Bell}}$, we can also estimate the purity drop:
\begin{equation}
\label{eq:purity-drop-from-full}
\gamma - \gamma_{\tilde{S}} = \expval{(-1)^{\pi_Y(i)} \left(1 -\prod_{k = 1}^{s} {\left[ {(-1)^{{\pi_Y}(g_k) + \left\langle {{g_k},i} \right\rangle } = 1} \right]}\right)}_{i \sim \mathrm{Bell}}
\end{equation}
which allows us to \textit{verify} whether a candidate group keeps its purity loss within $\varepsilon$. A formal description is as follows:
\begin{algorithm}[ht]
  \caption{\label{alg:sim-stab-ver}Simple Candidate Verification}
  \KwIn{Generators $g_1, g_2, \ldots, g_s$ for a commuting candidate group $\tilde{S} = \langle g_1, \ldots, g_s\rangle$, accuracy parameter $\varepsilon$, confidence parameter $\delta$.}
  \KwOut{Accept or reject the candidate $\tilde{S}$ according to its estimated purity loss.}
  Set $M \leftarrow \left\lceil\frac{8}{\varepsilon^2} \ln \frac{2}{\delta}\right\rceil$\;
  Obtain $M$ Bell samples for the state, say ${b_1, b_2, \ldots, b_M}$\;
  Compute $\widehat{\Delta\gamma}\leftarrow\frac{1}{M}\sum_{i=1}^M(-1)^{\pi_Y(b_i)}\left(1-\prod_{k=1}^{s}\left[(-1)^{\pi_Y(g_k)+\expval{g_k,b_i}}=1\right]\right)$\;
  \lIf{$\widehat{\Delta \gamma} \leqslant \frac{\varepsilon}{2}$}{accept}
  \lElse{reject}
\end{algorithm}
\begin{proposition}
\label{prop:sim-stab-ver}
Let $\Delta \gamma = \gamma - \gamma_{\tilde{S}}$. With probability at least $1 - \delta$:
\begin{itemize}
    \item If $\gamma - \gamma_{\tilde{S}} \geqslant \varepsilon$, then Algorithm~\ref{alg:sim-stab-ver} is guaranteed to reject.
    \item If $\gamma - \gamma_{\tilde{S}} = 0$ (i.e., $\tilde{S} \subseteq S$ is part of the exact weak stabilizer), then the algorithm is guaranteed to accept.
\end{itemize}
\end{proposition}
\begin{proof}
Let $\varepsilon' = \frac{\varepsilon}{2}$, so $M = \left\lceil\frac{2}{\varepsilon'^2} \ln \frac{2}{\delta}\right\rceil$. Eq.~\eqref{eq:purity-drop-from-full} shows that $\Delta \gamma$ is estimated correctly. Formally, because each sample $i \sim \text{Bell}$ satisfies $(-1)^{\pi_Y(i)}\left(1 - \prod_{k = 1}^{s} {\left[ {(-1)^{{\pi_Y}(g_k) + \left\langle {{g_k},i} \right\rangle } = 1} \right]}\right) \in \{-1,0,1\}$, by Hoeffding's inequality, $P\left(|\widehat{\Delta \gamma} - \Delta \gamma| \geqslant \varepsilon'\right) \leqslant 2 e^{-\frac{1}{2} M \varepsilon'^2} \leqslant \delta$. So with probability at least $1 - \delta$, $|\widehat{\Delta \gamma} - \Delta \gamma|  < \frac{\varepsilon}{2}$. In this case, when $\Delta \gamma = 0$, $\widehat{\Delta \gamma} < \varepsilon' = \frac{\varepsilon}{2}$; thus, the algorithm accepts. Meanwhile, if $\Delta \gamma \geqslant \varepsilon = 2 \varepsilon'$, then $\widehat{\Delta \gamma} > \varepsilon' = \frac{\varepsilon}{2}$; thus, the algorithm rejects.
\end{proof}

\subsection{Adaptive Candidate Verification}
\label{sec:adaptive-ver}

Algorithm~\ref{alg:sim-stab-ver} is a proof of concept. This is not sufficient as a subroutine for the full stabilizer-finding algorithm. Here, a finder algorithm needs to build the group step by step. It will maintain a guess $\tilde{S}$ and request to add another group $\Delta \tilde{S}$ via $\tilde{S} \leftarrow \tilde{S} + \Delta \tilde{S}$ ($\tilde{S} \leftarrow \spn(\tilde{S} \cup \Delta \tilde{S})$). A verifier algorithm will see whether this addition is valid by checking $\Delta \gamma = \gamma_{\tilde{S}} - \gamma_{\tilde{S} + \Delta \tilde{S}} \leqslant \varepsilon$ for some $\varepsilon$. For this purpose, we'll need three upgrades to the simple candidate verifier in Algorithm~\ref{alg:sim-stab-ver}:
\begin{description}
\item[$\pmb{\Delta \gamma}$]  The verifier needs to estimate the relative purity drop $\gamma_{\tilde{S}} - \gamma_{\tilde{S} + \Delta \tilde{S}}$ instead of the full purity drop from the original state $\gamma - \gamma_{\tilde{S} + \Delta \tilde{S}}$.
\item[$\pmb{\Delta \varepsilon_-}$ and $\pmb{\Delta \varepsilon_+}$] It takes as input a pair of parameters $\Delta \varepsilon_-$ and $\Delta \varepsilon_+$ with $0 \leqslant \Delta \varepsilon_- < \Delta \varepsilon_+$. The verifier guarantees acceptance if $\Delta \gamma \leqslant \Delta \varepsilon_-$, and guarantees rejection if $\Delta \gamma > \Delta \varepsilon_+$.
\item[Adaptivity] If the number of requests from the finder is large, it is advantageous for the verifier to draw a large number of samples first and use this fixed set of samples to estimate $\Delta \gamma$ whenever needed. This saves sample complexity. This is adaptive because each future query to the verifier will depend on results from the previous queries.
\end{description}

For our queries, a single set of Bell samples suffices, as in adaptive data analysis~\cite{bassily2016}. The key is to estimate the purity drops for all nested commuting group pairs accurately at once; the finder can then choose its queries adaptively. A crude simultaneous guarantee could already be obtained by union-bounding over all commuting Pauli groups: since each is generated by at most $n$ Paulis, there are at most $4^{n^2}$ candidates, giving $O\!\left(\frac{n^2+\log(1/\delta)}{\Delta\varepsilon_{\rm gap}^2}\right)$ Bell samples.

For a commuting group $H=\langle g_1,\ldots,g_s\rangle$, write the filter in Eq.~\eqref{eq:projected-purity-est} as
\begin{equation}
Q_H(b)=\prod_{k=1}^{s}\left[(-1)^{\pi_Y(g_k)+\expval{g_k,b}}=1\right].
\end{equation}
The filter is independent of which generators we pick.\footnote{This follows from $Q_H(b)=\frac{1}{|H|}\sum_{g\in H}(-1)^{\pi_Y(g)+\expval{g,b}}$ in Eq.~\eqref{eq:projected-purity}.}

The following lemma sharpens the union-bound estimate above  by exploiting additional structure in the filters $Q_H$, reducing the dependence on $n$ from quadratic to essentially linear.
\begin{lemma}[Uniform Bell-sample estimates]
\label{lemma:uniform-bell-estimates}
Let $0<\Delta\varepsilon_{\mathrm{gap}}\leqslant1$ and $0<\delta<1$. From $M$ independent Bell samples, define
\begin{equation}
\widehat{\Delta\gamma}_{H,H'}=\frac1M\sum_{i=1}^M(-1)^{\pi_Y(b_i)}\bigl(Q_H(b_i)-Q_{H'}(b_i)\bigr)
\end{equation}
for every pair of commuting groups $H\subseteq H'$, with purity drop $\Delta\gamma_{H,H'}=\gamma_H-\gamma_{H'}$. A sample size
\begin{equation}
\label{eq:uniform-bell-samples}
M=O\!\left(\frac{n\ln(2/\Delta\varepsilon_{\mathrm{gap}})+\ln(1/\delta)}{(\Delta\varepsilon_{\mathrm{gap}})^2}\right)
\end{equation}
suffices to ensure $|\widehat{\Delta\gamma}_{H,H'}-\Delta\gamma_{H,H'}|\leqslant\Delta\varepsilon_{\mathrm{gap}}/2$ simultaneously for all such pairs, with probability at least $1-\delta$. An explicit choice of $M$ is given in Eq.~\eqref{eq:uniform-bell-sample-choice}.
\end{lemma}
\begin{proof}
The condition $Q_H(b)=1$ is a system of affine linear equations in the $2n$ bits of $b$. On any set of sample points, its nonempty intersection with an affine subspace is determined by an affine basis of at most $2n+1$ of those points. This bounds the number of distinct filter patterns on the samples. The number of patterns for $Q_H-Q_{H'}$ is at most the square of this bound. The fixed signs $(-1)^{\pi_Y(b)}$ do not change this count. Symmetrization and Hoeffding's inequality then give Eq.~\eqref{eq:uniform-bell-samples}; Appendix~\ref{app:uniform-verification} gives the details.
\end{proof}

Let $\Delta\varepsilon_{\mathrm{gap}}>0$ be a known lower bound on every query's gap $\Delta\varepsilon_+-\Delta\varepsilon_-$. Lemma~\ref{lemma:uniform-bell-estimates} makes every purity-drop estimate accurate to within $\Delta\varepsilon_{\mathrm{gap}}/2$. This guarantee holds for any number of adaptive queries, using the same samples throughout.

\begin{algorithm}[ht]
    \caption{\label{alg:stab-ver} Adaptive Candidate Verification}
    \KwInit{Minimum query gap $0<\Delta\varepsilon_{\mathrm{gap}}\leqslant1$; failure probability $\delta$.}
    Draw $M$ independent Bell samples $b_1,\ldots,b_M$, with $M$ chosen by Eq.~\eqref{eq:uniform-bell-sample-choice}\;
    \KwData{Current commuting guess $\tilde S$ supplied by the finder, initially $\{I\}$.}
    \KwQ{A candidate increment $\Delta\tilde S$ and thresholds $0\leqslant\Delta\varepsilon_-<\Delta\varepsilon_+\leqslant1$, with $\Delta\varepsilon_+-\Delta\varepsilon_-\geqslant\Delta\varepsilon_{\mathrm{gap}}$.}
    \KwOut{Approve or reject the candidate update $\tilde S'=\tilde S+\Delta\tilde S$.}
    \setcounter{AlgoLine}{0}
    Compute a basis of $\tilde S'$ extending a basis of $\tilde S$; reject if $\tilde S'$ is not commuting\;
    Compute $\widehat{\Delta\gamma}=\frac1M\sum_{i=1}^M(-1)^{\pi_Y(b_i)}\bigl(Q_{\tilde S}(b_i)-Q_{\tilde S'}(b_i)\bigr)$\;
    \lIf{$\widehat{\Delta\gamma}\leqslant\frac{\Delta\varepsilon_++\Delta\varepsilon_-}{2}$}{approve the candidate $\tilde S'$}
    \lElse{reject}
\end{algorithm}

\begin{theorem}
\label{thm:stab-ver}
Let $\Delta\gamma=\gamma_{\tilde S}-\gamma_{\tilde S'}$. With probability at least $1-\delta$, the following hold simultaneously for all queries satisfying Algorithm~\ref{alg:stab-ver}'s gap condition and having commuting $\tilde S'$:
\begin{itemize}
    \item If $\Delta\gamma\leqslant\Delta\varepsilon_-$, the algorithm approves the candidate update.
    \item If $\Delta\gamma>\Delta\varepsilon_+$, the algorithm rejects.
\end{itemize}
\end{theorem}
\begin{proof}
By Lemma~\ref{lemma:projected-purity}, the purity-drop estimate has mean
\begin{equation}
\Delta\gamma=\gamma_{\tilde S}-\gamma_{\tilde S'}=\expval{(-1)^{\pi_Y(b)}\bigl(Q_{\tilde S}(b)-Q_{\tilde S'}(b)\bigr)}_{b\sim\mathrm{Bell}}.
\end{equation}
On the uniform event from Lemma~\ref{lemma:uniform-bell-estimates}, every query satisfies
\begin{equation}
|\widehat{\Delta\gamma}-\Delta\gamma|
\leqslant\frac{\Delta\varepsilon_{\mathrm{gap}}}{2}
\leqslant\frac{\Delta\varepsilon_+-\Delta\varepsilon_-}{2}.
\end{equation}
The two claims follow by comparing with the midpoint of the thresholds, as Fig.~\ref{fig:verifier-band} illustrates. The uniform event already covers every possible nested commuting group pair, so adaptive choices require no further failure-probability allocation.
\end{proof}

\begin{figure}[ht]
\centering
\includegraphics{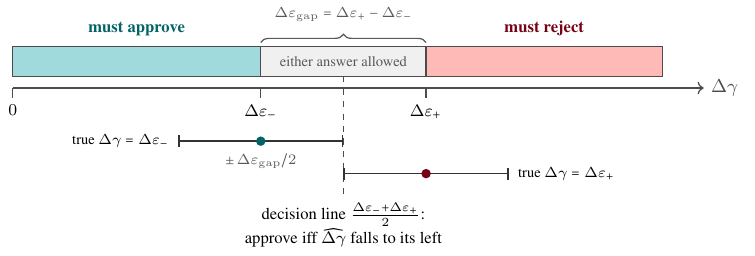}
\caption{The decision requirements of Algorithm~\ref{alg:stab-ver}, on the axis of the true purity drop $\Delta\gamma$. Below $\Delta\varepsilon_-$ the verifier must approve, above $\Delta\varepsilon_+$ it must reject, and in the gap between them either answer is allowed. The algorithm decides by comparing its estimate $\widehat{\Delta\gamma}$ with the midpoint of the two thresholds, and on the uniform event of Lemma~\ref{lemma:uniform-bell-estimates} that estimate is accurate to $\Delta\varepsilon_{\mathrm{gap}}/2$. The two ranges shown are therefore the worst cases: a true $\Delta\gamma$ sitting on either threshold can only move as far as the decision line, never across it, which is exactly Theorem~\ref{thm:stab-ver}.}
\label{fig:verifier-band}
\end{figure}

The verifier uses $M$ Bell samples, or $2M$ state copies, once at initialization. To evaluate queries efficiently, cache the signs $(-1)^{\pi_Y(b_i)}$ and the current filter values $Q_{\tilde S}(b_i)$. A proposal supplied by $r$ generators costs $O(nrM+n^2r)$ bit operations, including basis reduction and commutativity checks; redundant generators can be discarded during this reduction. Storage is $O(nM+n^2)$ bits. The finder decides whether to apply an approved update; the cached filter values change only when it does so.

The cost of Algorithm~\ref{alg:stab-ver} depends, among other things, on the gap $\Delta\varepsilon_+-\Delta\varepsilon_-$. In practice, it is often more convenient to think in terms of a single threshold $\Delta\varepsilon$. As an example, we can do this by choosing
\begin{equation}
\label{eq:verifier-relative-slack}
\Delta\varepsilon_-=(1-o(1))\Delta\varepsilon,
\qquad
\Delta\varepsilon_+=\Delta\varepsilon,
\end{equation}
where the $o(1)$ relative slack lies between $0$ and $1$. The verifier then effectively accepts updates with purity drop below $\Delta\varepsilon$ and rejects updates with purity drop above $\Delta\varepsilon$, up to a vanishingly small window around the threshold. This lets us work with the single parameter $\Delta\varepsilon$, at the cost of an additional factor polynomial in the inverse slack. We can let the slack tend to zero sufficiently slowly, e.g.\ inverse-logarithmically, that this overhead remains polynomial in the problem parameters. From here on, a clean-mode query with threshold $\Delta\varepsilon$ means this choice, and leading-order statements absorb the slack.

\subsection{Discussion}
\label{sec:ver-discussion}

Algorithm~\ref{alg:stab-ver}, the candidate verifier, will be the main subroutine inside a bigger stabilizer-finding algorithm. The verifier guards updates to the stabilizer guess by allowing at most $\Delta \varepsilon$ per step. Hence, after multiple updates, the total allowed purity drop is $\sum\Delta\varepsilon\leqslant\varepsilon$. The finder applies only nontrivial approved updates.

We may restrict algorithmic budgets to $0<\varepsilon\leqslant1$ and $0<\delta\leqslant1/2$: reducing larger requested values to these limits gives a stronger guarantee. The clean and intermediate noise thresholds involving $1/t$ are stated for $t\geqslant1$; when $t=0$, the top-down algorithm already has polynomial cost under the same purity lower bound.

There are two additional considerations:
\begin{description}
    \item[No Wrong Turns] The product-filter estimator above requires $\tilde S$ to be a commuting Pauli group. Accordingly, at each update step we check that $\tilde S+\Delta\tilde S$ is commuting. Our goal is eventually to reach a stage where $S \subseteq \tilde S$, where $S$ denotes the exact weak stabilizer. If the verifier were ever to accept an update $\tilde S \leftarrow \tilde S + \Delta \tilde S$ such that the resulting $\tilde S$ does not commute with $S$, then the commutativity requirement would preclude any future update from reaching $S \subseteq \tilde S$. In this sense, such an update is a ``wrong turn.''

    Lemma~\ref{lemma:compatibility} rules this out. On the uniform event of Lemma~\ref{lemma:uniform-bell-estimates}, every accepted update has true purity drop at most its $\Delta\varepsilon_+$, so the losses telescope:
    \begin{equation}
    \gamma-\gamma_{\tilde S}\leqslant\sum_{\text{accepted updates}}\Delta\varepsilon_+\leqslant\varepsilon,
    \end{equation}
    where the second inequality is how every finder below allocates its budget. The true $S$ has zero loss, so by the lemma the current $\tilde S$ commutes with $S$. A missing element of $S$ therefore passes the commutativity check, and adding it costs no purity (Observation~\ref{obs:add-S-purity}), so the verifier accepts it: no accepted update can block $S\subseteq\tilde S$.

    \item[Sample Reuse] The Bell samples used by the verifier are drawn once and reused across its queries. Bell difference samples used by the finder are drawn separately. Each finder batch still consists of fresh Bell difference samples, independent of the verifier data and of earlier batches.
\end{description}

In the finder analyses below, we assume the verifier is always correct: fix its dataset on the uniform-accuracy event of Lemma~\ref{lemma:uniform-bell-estimates}. This event depends only on the verifier dataset, so fresh finder batches keep their usual distribution conditional on it and on the preceding history. In these analyses and the finder pseudocode, $\delta$ denotes the finder's failure budget under this assumption. The verifier has a separate failure budget $\delta_{\mathrm{ver}}$.
The overall failure probability is at most
$\delta+\delta_{\mathrm{ver}}$.
For convenience, we leave their allocation free subject to
the requested overall confidence.

To summarize what the verifier hands to the finder: given the current commuting guess $\tilde S$ and a proposed commuting extension $\Delta \tilde S$, Algorithm~\ref{alg:stab-ver} accepts if the purity drop $\Delta\gamma \leqslant \Delta\varepsilon_-$ and rejects if $\Delta\gamma > \Delta\varepsilon_+$, using one Bell-sample dataset drawn up front and reused across arbitrarily many adaptive queries (Lemma~\ref{lemma:uniform-bell-estimates}). Its sample cost is therefore fixed before the finder starts; only its runtime grows with the number of queries. The finder, which we turn to next, never has to certify an update itself---its only job is to propose extensions that the verifier is likely to accept.

\section{Learning Algorithms}
\label{sec:learning-algorithms}

With the verifier in hand, what remains is to propose extensions $\Delta \tilde S$ that it will accept. We call this half of the pipeline the finder. This section develops the finders and their proposal rules, and deliberately stops short of proving that they work. Sec.~\ref{sec:up-down-per-batch} establishes what a single batch achieves, and Sec.~\ref{sec:end-to-end} specifies the full schedules and parameter choices and analyzes them.

We start with the top-down algorithm in Sec.~\ref{sec:top-down}, which feeds Bell difference samples to the verifier one at a time. It is a complete pipeline---pseudocode plus Theorem~\ref{thm:top-down}---and costs $O(4^t \cdot \mathrm{poly})$, which is what we use in the noisy regime. Sec.~\ref{sec:up-down} then introduces the up-down algorithm, Algorithm~\ref{alg:stab}: it collects a batch of Bell difference samples, post-selects them against the current guess, and proposes updates from the radical of their span. Its $\mathsf{propose}$ subroutine comes in two modes, clean and intermediate, and Sec.~\ref{sec:up-down-example} defines both through a small worked example.

\subsection{Top-Down Algorithm}
\label{sec:top-down}

\subsubsection{The Algorithm}

With the candidate verifier guarding updates to the guess $\tilde{S}$, we need a finder to propose these updates. The most straightforward idea is to propose updates directly from the Bell difference samples.

\begin{algorithm}[ht]
\caption{
\label{alg:top-down}Top-Down Stabilizer Finding Algorithm}
\KwIn{Accuracy parameter $\varepsilon$, confidence parameters $\delta$ and $\delta_{\mathrm{ver}}$}
\KwOut{An approximate weak stabilizer $\tilde{S}$ for the density matrix.}
Set sufficiently large $T$\;
Initialize $\tilde{S} \leftarrow I$\;
Initialize Algorithm~\ref{alg:stab-ver} with $\Delta\varepsilon_{\mathrm{gap}}=\varepsilon/n$ and failure probability $\delta_{\mathrm{ver}}$\;
\For{$i = 1$ \KwTo $T$}{
Draw Bell difference sample $b \sim \mathrm{BellDiff}$\;
\If{$b \not\in \tilde{S}$ and $b$ commutes with $\tilde{S}$}{
Propose to the verifier to add $\tilde{S} \leftarrow \tilde{S} + \{b\}$ with $\Delta \varepsilon_- = 0$, $ \Delta\varepsilon_+ = \frac{\varepsilon}{n}$; update if the verifier accepts\;
}
}
\Return{$\tilde{S}$}\;
\end{algorithm}

Intuitively, $\tilde{S}$ is like a sieve on the Pauli expansion $\rho = \sum_{i \in \mathcal{P}^n} a_i \sigma_i$: only terms commuting with $\tilde{S}$ contribute to $\gamma_{\tilde{S}}$. As we build up $\tilde{S}$, the sieve becomes finer---more constraints, fewer survivors, so the remaining set of Pauli terms becomes smaller. Because we repeatedly narrow down by adding more constraints, we call this approach the \textit{top-down} algorithm for stabilizer learning.

\subsubsection{Proof of Correctness}

We now show that Algorithm~\ref{alg:top-down} succeeds with $T \approx 4^t n / \gamma^2$ draws. The argument needs only two facts: the verifier always accepts a genuine element of $S$, and enough Bell difference samples land in $S$ to span it.

\begin{assumption}[Algorithm~\ref{alg:top-down}] 
\label{asp:top-down}
To simplify calculations, for the top-down algorithm, we assume
\begin{itemize}
    \item Number of qubits $n \gg 1$;
    \item The failure probability of the algorithm is at most $\delta$, with $\ln \frac{1}{\delta} \ll n$.
\end{itemize}
\end{assumption}
\begin{theorem}
\label{thm:top-down}
By taking $T \approx 4^t \frac{n}{\gamma^2}$, with probability at least $1-\delta$, Algorithm~\ref{alg:top-down} successfully gives an $\varepsilon$-approximate weak stabilizer $\tilde{S}$.
\end{theorem}
\begin{proof}
Since we set $\Delta \varepsilon_+ = \frac{\varepsilon}{n}$, and any commuting subgroup in $\mathcal{P}^n$ has dimension at most $n$, we would never run out of accuracy: $\gamma - \gamma_{\tilde{S}} \leqslant \sum \Delta \varepsilon \leqslant \varepsilon$. In the meantime, a Pauli in $S$ not yet in $\tilde{S}$ passes both checks of Algorithm~\ref{alg:top-down}: it commutes with $\tilde{S}$ by Lemma~\ref{lemma:compatibility}, and adding it causes no purity drop, so the verifier accepts it (Theorem~\ref{thm:stab-ver}). Therefore, as long as we get enough Bell difference samples to span $S$, $S \subseteq \tilde{S}$, which means $\tilde{S}$ will be an $\varepsilon$-approximate weak stabilizer. Say we need $M$ samples in $S$ for this to happen with high probability. Formally, we want to show:
\begin{enumerate}
    \item With probability at least $1-\delta_1$, we have drawn $M$ samples in $S$;
    \item Conditioned on having $M$ samples in $S$, with probability at least $1 - \delta_2$, these $M$ samples span $S$.
\end{enumerate}
We set $\delta_1 + \delta_2 \leqslant \delta$. For simplicity, say $\delta_1 = \delta_2 = \frac{\delta}{2}$.

\begin{description}
    \item[Proof of 1] By Eq.~\eqref{eq:bell-diff-dist}, $P_{\mathrm{BellDiff}}(i)$ is largest for $i \in S$. By Lemma~\ref{lemma:prob-diff-lower}, $P_{\mathrm{BellDiff}}(C(S)) \geqslant \gamma_{S}^2 = \gamma^2$. So, $P_{\mathrm{BellDiff}}(S) \geqslant \frac{|S|}{|C(S)|}P_{\mathrm{BellDiff}}(C(S)) \geqslant 4^{-\nu} \gamma^2$. Let $q = 4^{-\nu} \gamma^2$; by the Chernoff bound, we need $T \geqslant \frac{M + \sqrt{2M \ln \frac{1}{\delta_1}} + 2 \ln \frac{1}{\delta_1}}{q}$.
    \item[Proof of 2] By Observation~\ref{obs:stab-inv-diff}, when conditioned on being in $S$, the Bell difference distribution is uniform. So we bound success probability by bounding the failure events. If failure happens, then all the samples in $S$ must be contained in a proper subspace of $S$. Since $S$ is a Pauli subgroup of dimension $n - \nu$, there are $2^{n-\nu} - 1$ codimension-1 subspaces (each subspace commutes with an element of $\mathcal{P}^n / C(S)$). For each codimension-1 subspace, the probability of all the Bell difference samples in $S$ falling inside this subspace is bounded above by $\frac{1}{2^M}$. So, the probability of failure is bounded above by $2^{n-\nu-M} - 2^{-M} \leqslant 2^{n - \nu-M}$. So, taking $M \geqslant n - \nu + \log_2 \frac{1}{\delta_2}$ suffices.
\end{description}
Combined, to leading order, it suffices to take $T \gtrapprox \frac{n-\nu}{q} = 4^{\nu} \frac{n-\nu}{\gamma^2}$. Since $\nu$ is unknown, we should take $T \approx  4^t \frac{n}{\gamma^2}$, where we used $0 \leqslant \nu \leqslant t$ and $\ln \frac{1}{\delta_1} = \ln \frac{1}{\delta_2} = \ln \frac{2}{\delta} \ll n$.
\end{proof}
In practice, setting $T$ requires us to know purity $\gamma$, which can be estimated first via Bell sampling (Lemma~\ref{lemma:bell-sampling-known-quantites}).

The top-down algorithm also reaches noncommuting targets: when a Pauli channel with $\lambda\geqslant\frac12$ erases coefficients, $S$ can grow beyond $\mathrm{Stab}(\ket{\psi})$ and become noncommuting (Sec.~\ref{sec:degenerate-support}). Dropping the commutativity check handles this, since the projected purity of any Pauli group is the average of the overlaps $\tr(\sigma_g\rho\sigma_g\rho)$ (Sec.~\ref{sec:stab-ver}), which the verifier can estimate by drawing $g$ at random. We do not pursue this case, since such an $S$ is a symmetry created by the noise.

To summarize, when $\gamma$ is at least a constant, the top-down algorithm uses time and samples scaling as
\begin{equation*}
O\left(4^t \cdot \mathrm{poly}(n, \frac{1}{\varepsilon}, \ln \frac{1}{\delta}, \ln \frac{1}{\delta_{\mathrm{ver}}})\right).
\end{equation*}
It is simple, generic, and somewhat inefficient due to the $4^t$ factor.

\subsection{Up-Down Algorithm}
\label{sec:up-down}

For a pure state, the top-down algorithm takes time exponential in $t$, whereas Algorithm~\ref{alg:noise-free-learning} runs in polynomial time. The up-down algorithm interpolates between the two. Like Algorithm~\ref{alg:noise-free-learning}, it spans Bell difference samples; under noise, some samples lie outside $C(S)$ and shrink the radical of the span (Sec.~\ref{sec:summary-of-results}).

\begin{definition}[Good \& Bad Samples]
\label{def:good-bad}
A Bell difference sample in $C(S)$ is called a \textit{good} sample; conversely, a Bell difference sample not in $C(S)$ is called a \textit{bad} sample.
\end{definition}

The solution to the bad sample problem is to combine bottom-up aggregation with top-down verification:
\begin{itemize}
    \item We maintain a guess $\tilde{S}$ and use the verifier to guard its updates---the top-down idea.
    \item We collect a batch of Bell difference samples to span some subspace $B$ of the Paulis. Then, using $B$, we propose to the verifier to update $\tilde{S}$.
\end{itemize}
The algorithm is outlined in Algorithm~\ref{alg:stab}. It is thus called the up-down algorithm for stabilizer finding.
\begin{algorithm}[ht]
\caption{\label{alg:stab} Up-Down Stabilizer Finding Algorithm (Generic Form)}
\KwIn{Accuracy parameter $\varepsilon$, confidence parameters $\delta$ and $\delta_{\mathrm{ver}}$}
\KwOut{An approximate weak stabilizer $\tilde{S}$ for the density matrix.}
Initialize the adaptive verifier, Algorithm~\ref{alg:stab-ver}, with failure probability $\delta_{\mathrm{ver}}$ and minimum query gap $\Delta\varepsilon_{\mathrm{gap}}$, the smallest gap required by the accuracy schedule specified below; let its guess be $\tilde{S} = \{I\}$\;
Limit the following loop to $n+1$ iterations; report failure if its stopping condition is not reached\;
\Repeat{(no approved nontrivial update was applied and $\dim\tilde S\geqslant n-t$) or $\dim\tilde S=n$}{
    Set $\Delta \varepsilon$, $T$\;
    \For{$i = 1$ \KwTo $T$}{
      Set $M$\;
      Draw $M$ Bell difference samples, say $b = \{b_1, b_2, \ldots, b_M\}$\;
      \textbf{Post-selection:} keep all samples in $C(\tilde{S})$ and remove the others; now we have $b' = \{b_1', b_2', \ldots, b_{M'}'\}$\;
      Set $B' = \spn{b'}$, the span of the selected samples\;
      Run $\mathsf{propose}(B',\Delta \varepsilon)$, which uses $B'$ to propose updates on $\tilde{S}$ to the verifier with accuracy drop $\Delta \varepsilon$\;
      \lIf{verifier approves}{apply the approved update and break}
    }
}
\Return{$\tilde{S}$}\;
\end{algorithm}

Algorithm~\ref{alg:stab} is the common framework for both $\mathsf{propose}$ modes: every batch is fresh and post-selected, the finder proposes updates from the span of the surviving samples, and an update is applied only after the verifier approves it. Its knobs are the choices of $\Delta \varepsilon$, $T$, $M$, and the $\mathsf{propose}$ function. We will spell them out as we analyze the algorithm from Sec.~\ref{sec:up-down-per-batch} to Sec.~\ref{sec:up-down-arbi}, and Sec.~\ref{sec:batch-to-iter} specifies the full schedule for each mode. For now, our goal is to understand the structure of Algorithm~\ref{alg:stab}, which is illustrated in Fig.~\ref{fig:pipeline}.

\subsubsection{Post-Selection}

\label{sec:post-selection}

We first examine the inner loop. In each iteration we draw a batch of Bell difference samples and form their span $B$. The algorithm then uses (a structured subset of) $B$ to propose updates to the verifier.

The key new step is \textit{post-selection}. After drawing samples, we retain only those lying in $C(\tilde{S})$, i.e., those commuting with the current guess $\tilde{S}$. This is well motivated because every admissible update to $\tilde{S}$ must preserve commutativity with $\tilde{S}$; samples that anticommute with $\tilde{S}$ cannot directly support such updates. Moreover, by Lemma~\ref{lemma:compatibility}, the target weak stabilizer $S$ commutes with $\tilde{S}$, so post-selection does not discard the ``useful'' part of the distribution relevant for reaching $S \subseteq \tilde{S}$. At the same time, post-selection helps control \emph{bad samples}. Even if the raw BellDiff distribution contains many bad samples (those outside $C(S)$), restricting attention to $C(\tilde{S})$ can filter out a large fraction of samples that would otherwise contaminate the spanned subspace, making it substantially easier to extract stabilizer information from the remaining data.

Post-selection also changes the algebraic perspective. Once $\tilde{S}$ is fixed, the algorithm can only enlarge it, so what matters about a proposed update is not the specific Pauli representative but its class modulo $\tilde{S}$. For example, when proposing a single-generator extension $\tilde{S} \leftarrow \tilde{S} + \{b\}$, translating $b$ by any other element in $\tilde{S}$ yields the same enlarged subgroup.  More generally, when proposing $\tilde S \leftarrow \tilde S + A$, only the ``new directions'' contributed by $A$ beyond $\tilde{S}$ are relevant; this is captured by the quotient $(\tilde S + A)/\tilde S \;\cong\; A/(A\cap \tilde S)$ by the second isomorphism theorem.

In summary, after post-selection the effective sample space is $C(\tilde S)$ rather than $\mathcal P^n$, and proposed updates should be interpreted modulo $\tilde S$. To emphasize this shift, we use a superscript ${}^\star$ to denote quantities defined with post-selection in mind. 

\begin{notation}[Starred ($^\star$) Quantities]
\label{notation:starred}
A superscript ${}^\star$ indicates that the corresponding quantity is defined with post-selection onto $C(\tilde S)$ in mind, and hence is interpreted relative to the ambient space $C(\tilde S)$. A detailed list of these quantities can be found later in Table~\ref{tab:quantities-batch}.
\end{notation}

\subsubsection{Loop Structure}

\label{sec:loop-structure}
In the inner loop, we take a \emph{batch} of samples and propose an update to the verifier. When the noise is low enough, Sec.~\ref{sec:up-down-per-batch} shows that this yields a (useful) update with at least constant probability. In the end, we want $S \subseteq \tilde{S}$ with high probability.

To achieve such a goal, we group batches into \emph{iterations}. Each iteration consists of (at most) $T$ batch attempts (and the associated post-processing), until $\tilde{S}$ is updated. Within Algorithm~\ref{alg:stab}, one iteration corresponds to one execution of the repeat block, where we take at most $T$ batches via the for loop inside. This way, one of the following happens:
\begin{itemize}
    \item The verifier accepts one update and exits the loop, or
    \item no update is accepted, which we interpret as evidence that $S \subseteq \tilde{S}$, and we are successful. 
\end{itemize}
Suppose we showed that one batch yields a successful update with probability at least a constant when $S \not\subseteq \tilde{S}$. Then repetition boosts this event: we should pick $T = O(\ln \frac{n}{\delta})$ to achieve a per-iteration success probability at least $1 - \frac{\delta}{n+1}$. What counts as a successful update depends on the mode---an approved update need not be one---and Sec.~\ref{sec:batch-to-iter} makes this precise.

The simplest stopping rule is to run iterations until no accepted updates occur. In practice, we may use a better stopping condition, as in Algorithm~\ref{alg:stab}:
\begin{center}
\textbf{Repeat} $\dots$ \textbf{until} (\emph{no approved nontrivial update was applied} and $\dim \tilde{S} \geqslant n - t$) or $\dim \tilde{S} = n$
\end{center}
This stopping condition is justified by two facts:
\begin{itemize}
    \item $S \subseteq \tilde{S}$ only happens when $\dim \tilde{S} \geqslant \dim S \geqslant n - t$.
\item If $\dim \tilde{S} = n$, we have already reached a maximal commuting group, and $S$ commutes with it by Lemma~\ref{lemma:compatibility}, so $S \subseteq \tilde{S}$.
\end{itemize}
For the runtime bound, cap the total at $n+1$ iterations. There are at most $n$ nontrivial updates, and an iteration without an update after $S\subseteq\tilde S$ stops the algorithm. Thus either stopping before recovery or exhausting the cap requires an iteration without an update before recovery. The per-iteration bound above limits the probability of any such event to $\delta$.

The one piece of Algorithm~\ref{alg:stab} we have not described is $\mathsf{propose}$ itself. Rather than define it abstractly, we let a small example dictate its form.

\subsection{An Example}
\label{sec:up-down-example}

In this section we use a schematic example to motivate the construction of the $\mathsf{propose}$ function in its two modes---the \emph{clean} mode and the \emph{intermediate} mode---and to illustrate the intuition behind Algorithm~\ref{alg:stab}. The exact details---parameter choices, correctness proofs, and the noise thresholds that separate the regimes---are deferred to Sec.~\ref{sec:up-down-per-batch} and Sec.~\ref{sec:end-to-end}.

Consider a state $\rho$ on $n = 5$ qubits, with weak stabilizer:
\begin{equation}
S = \expval{Z_1, Z_2},
\qquad \dim S = 2, \quad \nu = 3.
\end{equation}
Qubitwise, $S$ and its centralizer $C(S)$ read
\begin{equation}
S \;=\; \boxed{\sfrac{I}{Z}} \; \boxed{\sfrac{I}{Z}} \; \boxed{I} \; \boxed{I} \; \boxed{I},
\qquad
C(S) \;=\; \boxed{\sfrac{I}{Z}} \; \boxed{\sfrac{I}{Z}} \; \boxed{\sfrac{\sfrac{I}{X}}{\sfrac{Y}{Z}}} \; \boxed{\sfrac{\sfrac{I}{X}}{\sfrac{Y}{Z}}} \; \boxed{\sfrac{\sfrac{I}{X}}{\sfrac{Y}{Z}}}.
\end{equation}
Each box ranges independently over the Paulis shown. We will use this box notation throughout without further comment.

Fix $\varepsilon \ll 1$, the purity-loss budget, at a small constant. We design the example $\rho$ such that including $Z_3$ in $\tilde{S}$ results in a purity loss $\Delta \gamma = O(\varepsilon^2)$, well below the $\varepsilon$ threshold. By contrast, including $Z_4$ causes a purity loss slightly above $\varepsilon$, say $2 \varepsilon$. In other words, $\tilde{S} = \langle Z_1, Z_2, Z_3\rangle$ is an \emph{acceptable} weak stabilizer, whereas $\tilde{S} = \langle Z_1, Z_2, Z_4\rangle$ is a \emph{fake} one. Having any other linearly independent Pauli in $\tilde{S}$ would result in a significant $O(1)$ purity drop and is thus unacceptable.
\begin{equation}
\underbrace{\boxed{\sfrac{I}{Z}} \quad \boxed{\sfrac{I}{Z}}}_{\text{true weak stabilizer}}
\qquad
\underbrace{\boxed{\sfrac{I}{Z}}}_{\substack{\downarrow\\\text{acceptable}}}
\qquad
\underbrace{\boxed{\sfrac{I}{Z}}}_{\text{fake}}
\qquad \boxed{I}
\label{eq:example-classification}
\end{equation}

As discussed in Sec.~\ref{sec:bell-difference-sampling}, the Bell difference distribution is tied to the Pauli coefficient structure of $\rho$.  In our case, this means that depending on the noise, most samples are good samples, living in $C(S) = \boxed{\sfrac{I}{Z}} \; \boxed{\sfrac{I}{Z}} \; \boxed{\sfrac{\sfrac{I}{X}}{\sfrac{Y}{Z}}} \; \boxed{\sfrac{\sfrac{I}{X}}{\sfrac{Y}{Z}}} \; \boxed{\sfrac{\sfrac{I}{X}}{\sfrac{Y}{Z}}}$. However, since imposing the additional directions $Z_3$ and $Z_4$ incurs only $O(\varepsilon^2)$ and $O(\varepsilon)$ purity loss, respectively, most good samples will be in $C(S + \langle Z_3, Z_4\rangle) = \boxed{\sfrac{I}{Z}} \; \boxed{\sfrac{I}{Z}} \; \boxed{\sfrac{I}{Z}} \; \boxed{\sfrac{I}{Z}} \; \boxed{\sfrac{\sfrac{I}{X}}{\sfrac{Y}{Z}}}$, varying only in the last qubit. Only with probability $O(\varepsilon)$ will we see a sample that is not $I$ or $Z$ on qubit 4 (but still $I$ or $Z$ on qubit 3). The probability that a good sample does not have $I$ or $Z$ on qubit 3 is $O(\varepsilon^2)$.
\begin{figure}[ht]
\centering
\includegraphics{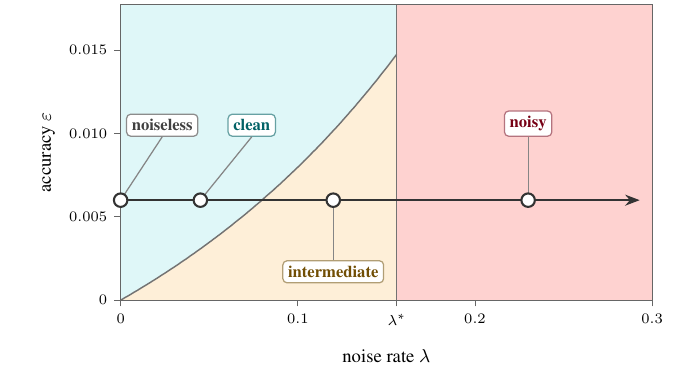}
\caption{We consider fixing $\varepsilon \ll 1$ and increasing $\lambda$ starting from $\lambda = 0$. Inside the diagram of Fig.~\ref{fig:phase-diagram}, this corresponds to traversing a horizontal line from left to right: starting from no noise at $\lambda = 0$, we pass through the clean regime, the intermediate regime, and finally the noisy regime. In each regime, the behavior of Bell difference samples changes qualitatively, requiring a different $\mathsf{propose}$ strategy if applicable.} 
\label{fig:phase-line-series}
\end{figure}

To better understand our example, we hold the accuracy parameter $\varepsilon \ll 1$ fixed and increase the noise rate $\lambda$ from 0, following the trajectory of Fig.~\ref{fig:phase-line-series} through the clean, intermediate, and noisy regimes. Each regime requires a different $\mathsf{propose}$ strategy, if applicable.

\begin{description}[leftmargin=*]
\item[Noiseless ($\lambda = 0$)]

Keep taking Bell difference samples and let $B$ be the group they span. Since the noise rate is 0, all samples are good and live in $C(S) = \boxed{\sfrac{I}{Z}}\;\boxed{\sfrac{I}{Z}}\;\boxed{\sfrac{\sfrac{I}{X}}{\sfrac{Y}{Z}}}\;\boxed{\sfrac{\sfrac{I}{X}}{\sfrac{Y}{Z}}}\;\boxed{\sfrac{\sfrac{I}{X}}{\sfrac{Y}{Z}}}$. However, as mentioned, since imposing both $Z_3$ and $Z_4$ incurs only $O(\varepsilon)$ purity loss, most samples concentrate in $C(S + \langle Z_3, Z_4\rangle) = \boxed{\sfrac{I}{Z}}\;\boxed{\sfrac{I}{Z}}\;\boxed{\sfrac{I}{Z}}\;\boxed{\sfrac{I}{Z}}\;\boxed{\sfrac{\sfrac{I}{X}}{\sfrac{Y}{Z}}}$. If we take $O(1)$ samples, at most we can get $B$ to span $C(S+\langle Z_3, Z_4\rangle)$, which is not useful.

However, if we take $\Omega(\frac{1}{\varepsilon})$ samples, we will eventually see a sample that is not $I$ or $Z$ on qubit 4 (but still $I$ or $Z$ on qubit 3). In this way, we can get $B$ to span $C(S+\langle Z_3\rangle)$. So, we simply output $\tilde{S} = C(B) =  \langle Z_1, Z_2, Z_3 \rangle = S + \langle Z_3 \rangle$, which solves the problem. There is no need for interactive updates protected by the verifier. This is exactly the idea of Ref.~\cite{hangleiter2024}.

\item[Clean regime]

With low noise (bounded above by some value $\Theta(\varepsilon)$), by the time we have drawn enough samples to span $C(S_{\text{acc}})$ for an acceptable $S_{\text{acc}}$, we may also get a few bad samples. However, if noise is sufficiently low, then there are very few bad samples, such that $B$ is just $C(S_{\text{acc}}) + \text{some bad samples}$. We can still compute $\Delta \tilde{S} = C(B)$, except now the bad samples will make $\Delta \tilde{S}$ smaller, so $\Delta \tilde{S} \subseteq S_{\text{acc}}$. This still leads to progress towards $S_{\text{acc}}$ and motivates Algorithm~\ref{alg:stab} with $\mathsf{propose}$ defined in Algorithm~\ref{alg:propose-clean}. Note that instead of proposing $\Delta \tilde{S} = C(B)$, we propose $\Delta \tilde{S} = \rad(B) = B \cap C(B)$. This choice is not necessary, but it ensures that $\Delta \tilde{S}$ is a commuting group and will be convenient later. We call this the \emph{clean mode} of the up-down algorithm.

\item[Intermediate regime]

If the noise rate is higher, say $\omega(\varepsilon)$, then by the time we have drawn enough samples to span $C(S_{\text{acc}})$ for an acceptable $S_{\text{acc}}$, we will have drawn so many bad samples that $B$ spans the whole Pauli group. In this case, $\Delta \tilde{S} = C(B) = \{I\}$ is trivial and useless. 

Therefore, we can no longer wait for $B$ to span $C(S_{\text{acc}})$. Instead, with $O(n)$ samples, we can get $B$ to span $S$, the true weak stabilizer. If the noise rate is less than some value $\Theta(1)$, then there will be a sufficiently large portion of $\rad B$ that is in $S$. Therefore, we'll need to iterate over $\rad(B)$ and propose to the verifier to add each Pauli in $\rad(B)$ to $\tilde{S}$, until the verifier accepts. This is the idea of $\mathsf{propose}$ in Algorithm~\ref{alg:propose-intermediate}. Each batch of samples could yield multiple linearly independent Paulis for successful updates, but for simplicity we stop once one is found. Here, Algorithm~\ref{alg:stab} operates in the \emph{intermediate mode}.

\item[Noisy regime]
If the noise rate is too large, bad samples are too frequent and contaminate any attempts to build up $B$. We no longer have the same progress guarantee, so we turn to the top-down algorithm in Algorithm~\ref{alg:top-down}.
\end{description}

The example thus suggests two $\mathsf{propose}$ rules, written out in Algorithms~\ref{alg:propose-clean} and~\ref{alg:propose-intermediate}. In clean mode, we hand the verifier the whole radical $R = \rad(B')$ at once. In intermediate mode, we draw elements of $R$ at random and hand them over one at a time, up to $L$ attempts per batch.

\begin{algorithm}[ht]
\caption{$\mathsf{propose}$ -- clean mode\label{alg:propose-clean}}
\KwIn{post-selected sample span $B'$; accuracy budget $\Delta\varepsilon$}
\Def{$\mathsf{propose}(B', \Delta\varepsilon)$}{
  Compute $R \leftarrow \rad(B')$\;
  \lIf{$R\subseteq\tilde S$}{reject}
  Query Algorithm~\ref{alg:stab-ver} to propose incrementing $\tilde{S}$ by $\Delta\tilde{S} = R$ with threshold $\Delta \varepsilon$; output the result\;
}
\end{algorithm}

\begin{algorithm}[ht]
\caption{$\mathsf{propose}$ -- intermediate mode\label{alg:propose-intermediate}}
\KwIn{post-selected sample span $B'$; accuracy budget $\Delta\varepsilon$}
\Def{$\mathsf{propose}(B', \Delta\varepsilon)$}{
  Set $L$\;
  Compute $R \leftarrow \rad(B')$\;
  \lIf{$R\subseteq\tilde S$}{reject}
  \For{$i = 1$ \KwTo $L$}{
    Draw $r$ uniformly at random from $R\setminus\tilde{S}$\;
    Query Algorithm~\ref{alg:stab-ver} to increment $\tilde{S}$ by $\Delta\tilde{S} = \{r\}$ with budget $\Delta\varepsilon_- = 0$, $\Delta \varepsilon_+ = \Delta \varepsilon$\;
    \lIf{the verifier approves}{return the candidate update $\tilde{S} \leftarrow \tilde{S} + \{r\}$}
  }
  \textbf{reject}\;
}
\end{algorithm}

At this point the finder--verifier framework is in place: the verifier (Algorithm~\ref{alg:stab-ver}), the top-down finder (Algorithm~\ref{alg:top-down}), the generic up-down finder (Algorithm~\ref{alg:stab}), and its two $\mathsf{propose}$ modes (Algorithms~\ref{alg:propose-clean} and~\ref{alg:propose-intermediate}). The next two sections choose the knobs---the batch size $M$, the number of attempts $T$, the per-update budget $\Delta\varepsilon$---and find the noise thresholds under which each mode succeeds, from the bottom up. Sec.~\ref{sec:up-down-per-batch} asks when a \emph{single} batch makes progress; Sec.~\ref{sec:end-to-end} chains batches into iterations and iterations into the final guarantees, specifying the full schedule for each mode along the way. A reader who only wants the end results can skip ahead to Theorems~\ref{thm:main-intermediate-depo}--\ref{thm:main-clean-arbi}.

\section{Single-Batch Progress}
\label{sec:up-down-per-batch}

Having established the structure of Algorithm~\ref{alg:stab}, we now prove its correctness and specify the parameter choices. The proof spans this section and the next; here we focus on a single batch and establish the corresponding success conditions for the clean and intermediate modes. Sec.~\ref{sec:batch-setup} fixes the notation for one batch, Sec.~\ref{sec:good-samples} lower-bounds the number of good samples needed before the radical carries useful structure, and Sec.~\ref{sec:bad-samples} upper-bounds the number of bad samples the radical can tolerate. The outputs of this section are Lemma~\ref{lemma:per-batch-intermediate} and Lemma~\ref{lemma:per-batch-clean}, the per-batch success conditions of the two modes.
\subsection{Setup}
\label{sec:batch-setup}

At the start of each batch, the algorithm has a current guess $\tilde{S}$. The goal is to enlarge $\tilde{S}$ until it contains the true weak stabilizer $S$. For this purpose, we define the target group for the current batch as $S^\star \equiv S + \tilde{S}$. Thus, progress means updating $\tilde{S}$ in a way that moves it closer to containing $S^\star$. The verifier guards each attempted update by allowing at most a purity drop of $\Delta \varepsilon$. We would like each batch to have a constant probability of producing a successful update toward $S^\star$. Define $\nu = n - \dim S$, $\tilde{\nu} = n - \dim \tilde{S}$, and $\nu^\star = n - \dim S^\star$ to be the nullities of $S$, $\tilde{S}$, and $S^\star$ respectively. Initially, $\dim \tilde{S} = 0$ and $\tilde{\nu} = n$. As the algorithm goes on, $\dim \tilde{S}$ increases and $\tilde{\nu}$ decreases. At the same time, $\nu^\star$ also decreases, but at a much slower pace, until eventually $\tilde{S}$ contains $S$, which means $S^\star = \tilde{S}$, so $\tilde{\nu} = \nu^\star$.

A Bell difference sample surviving post-selection lies in $C(\tilde{S})$ by construction. Following the principle from Notation~\ref{notation:starred}, define good$^\star$ and bad$^\star$ samples as good and bad samples among these surviving samples. In other words, good$^\star$ samples lie in $C(S^\star) = C(S) \cap C(\tilde{S})$, and bad$^\star$ samples lie in $C(\tilde{S}) \setminus C(S^\star)$. The post-selected span $B'$ (as in Algorithm~\ref{alg:stab}) is therefore generated by a mixture of good$^\star$ and bad$^\star$ samples. As we show below, good$^\star$ samples carry information useful for updating $\tilde{S}$ toward $S^\star$, while bad$^\star$ samples can corrupt it. Figure~\ref{fig:batch-groups} shows how these groups nest.

\begin{figure}[ht]
\centering
\includegraphics{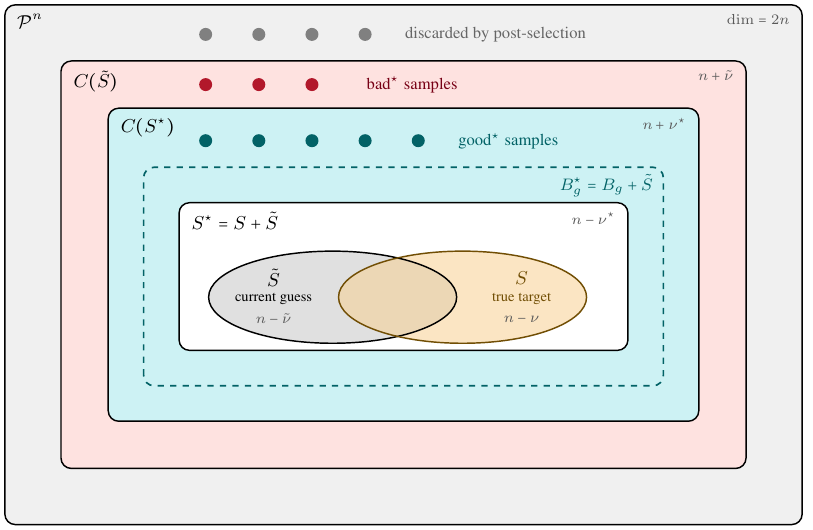}
\caption{The groups of a single batch, drawn by nesting, and where a Bell difference sample can land. Post-selection keeps only the samples inside $C(\tilde S)$; among those, good$^\star$ samples lie in $C(S^\star)$ and bad$^\star$ samples in $C(\tilde S) \setminus C(S^\star)$. $S$ and $\tilde S$ need not contain one another, which is why the target for the batch is their join $S^\star = S + \tilde S$. Each region is labelled with its $\mathbb{F}_2$-dimension; these are symmetric about $n$ because $\dim C(H) = 2n - \dim H$. The gaps used in the proofs follow by subtraction: $\dim C(S^\star)/S^\star = 2\nu^\star$, and $\dim S^\star/\tilde S = \dim C(\tilde S)/C(S^\star) = \tilde\nu - \nu^\star$. The dashed box illustrates the span $B_g^\star$ of good$^\star$ samples that we seek (Sec.~\ref{sec:good-samples}): it contains $S^\star$ and, in the clean mode, captures nearly all the projected purity.}
\label{fig:batch-groups}
\end{figure}

The per-batch success probability is thus controlled by two quantities:
\begin{itemize}
    \item $p_g^\star$, the probability that a Bell difference sample is good$^\star$;
    \item $p_b^\star$, the probability that a Bell difference sample is bad$^\star$.
\end{itemize}

Suppose we draw a batch of $M$ Bell difference samples. As a proof device, we first reveal each sample's type: good$^\star$, bad$^\star$, or filtered out by post-selection. Let $M_g^\star$ and $M_b^\star$ denote the counts of good$^\star$ and bad$^\star$ samples; the remaining $M - M_g^\star - M_b^\star$ are filtered out and do not contribute to $B'$.

Conditioned on this classification, we prove two conditions:
\begin{description}
    \item[Good$^\star$ sample condition]
    If $M_g^\star \geqslant K_g$ for a suitable lower bound $K_g$, then with high probability the good$^\star$ samples contain useful information for updating $\tilde{S}$ toward $S^\star$.

    \item[Bad$^\star$ sample condition]
    If $M_b^\star \leqslant K_b$ for a suitable upper bound $K_b$, then with high probability the bad$^\star$ samples do not corrupt this useful information too much.
\end{description}

We will then use a lower bound $q_g^\star$ on $p_g^\star$ and an upper bound $q_b^\star$ on $p_b^\star$, namely
\begin{equation}
    p_g^\star \geqslant q_g^\star,
    \qquad
    p_b^\star \leqslant q_b^\star.
\end{equation}
The quantities $q_g^\star$ and $q_b^\star$ will be computed in Secs.~\ref{sec:up-down-depo} and~\ref{sec:up-down-arbi} from the error model. If $q_g^\star$ is sufficiently large relative to $q_b^\star$, then we can choose a batch size $M$ such that both the good$^\star$ and bad$^\star$ sample conditions hold with high probability. This yields the desired per-batch correctness guarantee.

Sec.~\ref{sec:good-samples} proves the good$^\star$ sample condition, while Sec.~\ref{sec:bad-samples} proves the bad$^\star$ sample condition and combines the two conditions into a per-batch success guarantee. The relevant quantities for these sections are summarized in Table~\ref{tab:quantities-batch}.

\begin{table}[h]
    \centering
    \renewcommand{\arraystretch}{1.25}
    \begin{tabular}{c c c}
    \hline
      Where used & Concept  &  Interpretation  \\
      \hline
      \multirow{10}{*}{Both sections}
      & $S^\star = S+\tilde S$ & Target group of the batch \\
      & $B'$ & Span of the surviving (post-selected) samples \\
      & $B_g$, $B_b$ & Spans of the good$^\star$ and bad$^\star$ samples \\
      & $(B')^\star = B'+\tilde S$ & $B'$ with the accepted directions \\
      & $B_g^\star = B_g+\tilde S$ & Good$^\star$ span with the accepted directions \\
      & $\mathfrak{B}_g^\star = B_g+S^\star$ & Good$^\star$ span with the full target (analysis only) \\
      \cline{2-3}
      & $\Delta \varepsilon$  & Maximum allowed purity drop \\
      & $\nu, \tilde{\nu}, \nu^\star$ & $n - \dim S, n - \dim \tilde{S}, n - \dim S^\star$ \\
      & good$^\star$ samples & Useful surviving samples \\
      & bad$^\star$ samples & Corrupting surviving samples \\
      \hline
      Sec.~\ref{sec:good-samples} 
      & $M_g^\star$ & Number of good$^\star$ samples in a batch \\
      \hline
      \multirow{7}{*}{Sec.~\ref{sec:bad-samples}}
      & $R = \rad(B')$ & Radical of the surviving-sample span \\
      & $R^\star = R+\tilde S = \rad((B')^\star)$ & The radical with the accepted directions \\
      \cline{2-3}
      & $M_b^\star$ & Number of bad$^\star$ samples in a batch \\
      & $p_g^\star$ & Probability of a good$^\star$ sample in $\mathrm{BellDiff}$ \\
      & $p_b^\star$ & Probability of a bad$^\star$ sample in $\mathrm{BellDiff}$ \\
      & $q_g^\star$ & Lower bound on $p_g^\star$ \\
      & $q_b^\star$ & Upper bound on $p_b^\star$ \\
      \hline
    \end{tabular}
    \caption{Relevant quantities for the per-batch success analysis. On groups and spans, a star adds the already accepted directions $\tilde S$; $\mathfrak{B}_g^\star$ also adds the full target. Quantities that depend on $S$, including the good$^\star$/bad$^\star$ split, are analysis objects that the algorithm does not know.}
    \label{tab:quantities-batch}
\end{table}

Two conditions govern whether the good$^\star$ samples carry enough useful information. Both analyses use the first condition; the clean analysis also uses the second, a purity-coverage condition. If both conditions can hold before total corruption by bad$^\star$ samples, then we are in the \emph{clean} regime; if only the first can hold, we are in the \emph{intermediate} regime. If even the first fails, we are in the \emph{noisy} regime, where these bounds no longer guarantee progress for the up-down algorithm.

\subsection{Good Samples}
\label{sec:good-samples}

Suppose our batch contains $M_g^\star$ good samples. Say these samples are $b_1, b_2, \ldots, b_{M_g^\star}$, and they span the subgroup $B_g = \spn\{b_1, \ldots, b_{M_g^\star}\}$. As mentioned in Sec.~\ref{sec:post-selection}, since the goal is to extend $\tilde{S}$ by some other group $\Delta \tilde{S}$ based on these samples, what matters are not these samples themselves but these samples modulo $\tilde{S}$. In other words, only the new directions provided by these samples matter. To capture this mathematically, we consider $B_g^\star = B_g + \tilde{S}$, or alternatively, $B_g^\star / \tilde{S} = \spn\{b_1^\star, b_2^\star, \ldots, b_{M_g^\star}^\star\}$ where $b_i^\star = b_i + \tilde{S}$.

We want $B_g^\star$ to be as large as possible. Since all the good$^\star$ samples are in $C(S^\star)$, ideally, the largest it can be is $B_g^\star = C(S^\star)$. In practice, however, this is too much to ask. At the bare minimum, progress is possible so long as $S^\star \subseteq B_g^\star$:

\begin{lemma}
\label{lemma:good-span-S}
If $M_g^\star \geqslant \tilde{\nu} + \nu^\star + \log_2 \frac{1}{\delta}$ ($0 < \delta < 1$), then with success probability at least $1 - \delta$, $S^\star \subseteq B_g^\star$.
\end{lemma}
\textit{Proof idea.} Think of each good$^\star$ sample as carrying two pieces of information: which $S^\star$-coset it lies in, and where it sits inside that coset. Only the second piece is useful, and since the Bell difference distribution is blind to shifts by $S$, it is uniformly random in $S^\star/\tilde S$. The first piece is junk living in $C(S^\star)/S^\star$, but that space has only $2\nu^\star$ dimensions, so the junk can eat up at most $2\nu^\star$ samples: out of $M_g^\star$ samples, we can always form $M_g^\star-2\nu^\star$ independent combinations whose junk cancels, each a fresh uniform vector in $S^\star/\tilde S$. Spanning this $(\tilde\nu-\nu^\star)$-dimensional space takes $\tilde\nu-\nu^\star$ such vectors, and each extra one halves the chance of missing a direction, so another $\log_2\frac{1}{\delta}$ vectors push the failure probability below $\delta$. Altogether, $2\nu^\star+(\tilde\nu-\nu^\star)+\log_2\frac{1}{\delta}=\tilde\nu+\nu^\star+\log_2\frac{1}{\delta}$.

\begin{proof}
Since $\tilde{S} \subseteq S^\star$, the target $S^\star \subseteq B_g^\star$ is equivalent to $S^\star/\tilde{S} \subseteq B_g^\star/\tilde{S} = \spn\{b_1^\star, \ldots, b_{M_g^\star}^\star\}$ for $b_i^\star = b_i + \tilde{S}$, so it suffices to show $S^\star/\tilde{S} \cap B_g^\star/\tilde{S} = S^\star/\tilde{S}$.

We draw each good$^\star$ sample $b_i$ in two steps: first, we reveal the coset of $S^\star$ that $b_i$ lies in, i.e., $\mathfrak{b}_i \equiv b_i + S^\star$. Then, we reveal the specific $b_i$ within this coset.

In this way, the first step determines the set of all possible linear combinations of good$^\star$ samples that fall within $S^\star$, i.e., $K = \left\{x \in \{0,1\}^{M_g^\star}: \sum_i x_i \mathfrak{b}_i = S^\star\right\} = \ker \phi_{S^\star}$ where $\phi_{S^\star}: x \mapsto \sum_{i=1}^{M_g^\star} x_i \mathfrak{b}_i$. Since $\mathfrak{b}_i \in C(S^\star)/S^\star$ and $\dim C(S^\star)/S^\star = 2 \nu^\star$, $k \equiv \dim K \geqslant M_g^\star - 2 \nu^\star$ (rank--nullity theorem).

The second step thus determines $(B_g^\star/\tilde{S}) \cap (S^\star/\tilde{S}) = \left\{\sum_{i=1}^{M_g^\star} x_i b_i^\star : x \in K\right\} = \phi_{\tilde{S}}(K)$ where $\phi_{\tilde{S}}: x \mapsto \sum_{i=1}^{M_g^\star} x_i b_i^\star$. By Observation~\ref{obs:stab-inv-diff}, the $\mathrm{BellDiff}$ distribution is invariant under translation by $S$. Since $S^\star = S + \tilde{S}$, that means within each $S^\star$-coset, the distribution of a sample's $\tilde{S}$-coset is uniform. Therefore, in the second step, each $b_i^\star$ is uniformly distributed, even though $b_i$ may not. 

Intuitively, this means $(B_g^\star/\tilde{S}) \cap (S^\star/\tilde{S}) = \phi_{\tilde{S}}(K)$ is generated by $k \equiv \dim K$ i.i.d.\ uniform samples from $S^\star/\tilde{S}$. As a formal proof, fix a basis $g_1, \ldots, g_k$ of $K$. We wish to show that $\phi_{\tilde{S}}(\{g\}) = (\phi_{\tilde{S}}(g_1), \ldots, \phi_{\tilde{S}}(g_k))$ are uniformly random when we pick $b_i^\star$ within their respective  $S^\star/\tilde{S}$-cosets according to the second step. Consider shifting each $b_i^\star$ by $\Delta b_i^\star \in S^\star/\tilde{S}$. Then, $(\Delta \phi_{\tilde{S}}({g}))^T = G \cdot (\Delta b_1^\star, \ldots, \Delta b_{M_g^\star}^\star)^T$ where $G = \left\{(g_i)_j\right\} \in \mathbb{F}_2^{k \times M_g^\star}$. Since $k \leqslant M_g^\star$ and $G$ has rank $k$ (because $g_1, \ldots, g_k$ are linearly independent), the map $\{\Delta b_i^\star\}_i \mapsto \Delta \phi_{\tilde{S}}({g})$ is a surjective linear map $(S^\star/\tilde{S})^{M_g^\star} \rightarrow (S^\star/\tilde{S})^k$. Because $\Delta b_i^\star$ are uniform on $S^\star/\tilde{S}$, $\Delta \phi_{\tilde{S}}({g})$ are uniform on $(S^\star/\tilde{S})^k$, and hence so are $\phi_{\tilde{S}}({g})$. So $\phi_{\tilde{S}}(K)$ is effectively generated by $k$ i.i.d.\ uniform samples in $S^\star/\tilde{S}$.

Then $\phi_{\tilde{S}}(K) \neq S^\star/\tilde{S}$ iff these samples are all clustered within a codimension-1 subspace. There are $2^{\dim(S^\star/\tilde{S})} - 1 = 2^{\tilde{\nu} - \nu^\star} - 1$ such subgroups, and each contains all these samples with probability $2^{-k}$. By the union bound,
\begin{equation}
\mathbb{P}[\text{fail}] \leqslant (2^{\tilde{\nu} - \nu^\star} - 1) \cdot 2^{-k} \leqslant 2^{\tilde{\nu} - \nu^\star - k} \leqslant 2^{\tilde{\nu} + \nu^\star - M_g^\star},
\end{equation}
so $\mathbb{P}[\text{success}] \geqslant 1 - \delta$ once $M_g^\star \geqslant \tilde{\nu} + \nu^\star + \log_2\frac{1}{\delta}$.
\end{proof}

In addition to Lemma~\ref{lemma:good-span-S}, we would like $B_g^\star$ to ``approximately'' span $C(S^\star)$, if possible: the good$^\star$ samples, together with $S^\star$, should capture almost all of the current projected purity. This can happen when $M_g^\star \geqslant \Omega(\frac{\nu^\star}{\Delta \varepsilon})$:
\begin{lemma}
\label{lemma:good-span-C(S)}
For $\nu^\star>0$, $0<\delta<1$, $0<\Delta\varepsilon\ll\gamma_{\tilde S}$, and $\ln(1/\delta)/\nu^\star=O(1)$, if $M_g^\star \geqslant \frac{\nu^\star}{\Delta \varepsilon} \left(1 + O\left(\sqrt{\frac{\ln\frac{1}{\delta}}{\nu^\star}}\right) + O\left(\frac{\Delta \varepsilon}{\gamma_{\tilde S}}\right)\right)$, then with probability at least $1 - \delta$, $\gamma_{C(\mathfrak{B}_g^\star)} = 2^n\sum_{i \in \mathfrak{B}_g^\star} a_i^2 \geqslant \gamma_{\tilde{S}} - \Delta \varepsilon$, where $\mathfrak{B}_g^\star = B_g + S^\star$. An exact sufficient sample bound is given in Eq.~\eqref{eq:good-span-exact-bound}. If $\nu^\star=0$ or $\Delta\varepsilon\geqslant\gamma_{\tilde S}$, this coverage condition already holds without samples.
\end{lemma}
\begin{proof}
Similar to the proof of Lemma~\ref{lemma:good-span-S}, for each good$^\star$ sample $b_i$, define $\mathfrak{b}_i = b_i + S^\star$ to be the $S^\star$-coset it lies in. Let $\mathfrak{B}_i = S^\star+\spn\{b_1,\ldots,b_i\}\subseteq C(S^\star)$, so $\mathfrak{B}_i/S^\star=\spn\{\mathfrak{b}_1,\ldots,\mathfrak{b}_i\}$ is the space spanned by good$^\star$ samples as we draw them (up to cosets of $S^\star$). Then $\mathfrak{B}_0 = S^\star$ and $\mathfrak{B}_{M_g^\star} = \mathfrak{B}_g^\star$; set $D_i = \dim(\mathfrak{B}_i/S^\star)$. Since $S^\star\subseteq C(\mathfrak{B}_i)\subseteq C(S^\star)$, Lemma~\ref{lemma:belldiff-and-gamma} applies with $H=S^\star$ and $H'=C(\mathfrak{B}_i)$, even when $H'$ is noncommuting. By Observation~\ref{obs:add-S-purity}, $\gamma_{S^\star}=\gamma_{\tilde S}$. Thus the failure event $\gamma_{C(\mathfrak{B}_i)} < \gamma_{\tilde{S}} - \Delta\varepsilon$ implies
\begin{equation}
\label{eq:fail-implies-low-coverage}
\mathbb{P}_{b \sim \mathrm{good}^\star}[b \in \mathfrak{B}_i] = \mathbb{P}_{b \sim \mathrm{good}^\star}[\mathfrak{b} \in \mathfrak{B}_i/S^\star] \leqslant 1 - 2\Delta\varepsilon + \frac{2(\Delta\varepsilon)^2}{\gamma_{\tilde S}},
\end{equation}
where $\mathrm{good}^\star$ is the distribution of Bell difference samples conditioned on being good$^\star$, and we used $C(C(\mathfrak B_i))=\mathfrak B_i$.

While the coverage condition fails, the next sample $b_{i+1}$ lands outside $\mathfrak{B}_i$ with probability at least $2\Delta\varepsilon(1-\Delta\varepsilon/\gamma_{\tilde S})$, which increases $D$ by one. In other words, while the failure condition holds, $\Delta D$ stochastically dominates $\mathrm{Bernoulli}(2\Delta\varepsilon(1-\Delta\varepsilon/\gamma_{\tilde S}))$.

We then apply the Chernoff bound. For bookkeeping, define another process $W$ parallel to $D$, with $W_0=0$, $\Delta W = \Delta D$ while the failure condition lasts, and $\Delta W = 1$ afterward. The coverage condition is preserved as the span grows. In this way:
\begin{enumerate}
  \item During the failure event, $W = D$. Since $\mathfrak B_i/S^\star\subseteq C(S^\star)/S^\star$, which has dimension $2\nu^\star$, $W \leqslant 2\nu^\star$ on failure.
  \item $W$ stochastically dominates a sum of i.i.d.\ $\mathrm{Bernoulli}(2\Delta\varepsilon(1-\Delta\varepsilon/\gamma_{\tilde S}))$ variables, so we can apply Chernoff to $W$.
\end{enumerate}

So, for $\mu\geqslant2\nu^\star$, we have:
\begin{equation}
\mathbb{P}[\text{failure}] \leqslant \mathbb{P}[W_{M_g^\star} \leqslant 2\nu^\star] \leqslant \exp\!\left(-\frac{(\mu - 2\nu^\star)^2}{2\mu}\right),
\end{equation}
where $\mu = 2\Delta\varepsilon(1-\Delta\varepsilon/\gamma_{\tilde S})M_g^\star$ is the mean of the binomial comparison. Setting the right-hand side $\leqslant \delta$ and solving gives the exact sufficient bound
\begin{equation}
\label{eq:good-span-exact-bound}
M_g^\star\geqslant
\frac{2\nu^\star+\ln(1/\delta)+\sqrt{\ln^2(1/\delta)+4\nu^\star\ln(1/\delta)}}
{2\Delta\varepsilon(1-\Delta\varepsilon/\gamma_{\tilde S})}
,
\end{equation}
valid for all $0<\Delta\varepsilon<\gamma_{\tilde S}$ and $0<\delta<1$. To leading order, this yields
\begin{equation}
\label{eq:good-span-C(S)-bound}
M_g^\star \geqslant \frac{\nu^\star}{\Delta\varepsilon}\left(1 + O\left(\sqrt{\frac{\ln(1/\delta)}{\nu^\star}}\right) + O\left(\frac{\Delta\varepsilon}{\gamma_{\tilde S}}\right)\right).
\end{equation}
Here, Chernoff gives a factor at most $1 + \sqrt{2\frac{\ln \frac{1}{\delta}}{2\nu^\star}} + 2\frac{\ln \frac{1}{\delta}}{2\nu^\star}$. Eq.~\eqref{eq:good-span-C(S)-bound}'s big-$O$ is intended for the regime $\sqrt{\ln \frac{1}{\delta}/\nu^\star} \leqslant O(1)$. The same can be said for all Chernoff bounds in this paper (Lemmas~\ref{lemma:per-batch-intermediate} and \ref{lemma:per-batch-clean}). The purity correction is $O(\Delta\varepsilon)$ whenever $\gamma_{\tilde S}$ has a constant lower bound.

When $\nu^\star=0$, $S^\star=C(S^\star)$, so the initial span already retains purity $\gamma_{\tilde S}$. The case $\Delta\varepsilon\geqslant\gamma_{\tilde S}$ follows from nonnegativity of purity.
\end{proof}

Lemma~\ref{lemma:good-span-S} establishes the intermediate-mode condition; Lemma~\ref{lemma:good-span-C(S)} establishes the additional condition for the clean mode.

At $\lambda=0$ every Bell difference sample is good (Proposition~\ref{prop:pure-state-bell-diff}), and these two lemmas alone give the sample complexity of the noise-free learner, Algorithm~\ref{alg:noise-free-learning}, in our accuracy measure:

\begin{corollary}[Sample complexity of Algorithm~\ref{alg:noise-free-learning}]
\label{lemma:noise-free-samples}
For a pure state $\psi$ with stabilizer nullity $\nu \leqslant t$, $0<\varepsilon<1/2$, and $0<\delta<1$, $M = O\!\left(n + \frac{t + \ln(1/\delta)}{\varepsilon}\right)$ Bell difference samples suffice for Algorithm~\ref{alg:noise-free-learning} to output a commuting group $\tilde S$ satisfying
\begin{equation}
S\subseteq\tilde S,\qquad 1-\gamma_{\tilde S}\leqslant\varepsilon,
\end{equation}
with probability at least $1-\delta$, where $S=\mathrm{Stab}(\psi)$.
\end{corollary}
\begin{proof}
The radical $\tilde S$ is commuting by construction. Lemmas~\ref{lemma:good-span-S} and~\ref{lemma:good-span-C(S)}, applied with initial group $\{I\}$ and $\Delta\varepsilon=\varepsilon$ and the exact bound of Eq.~\eqref{eq:good-span-exact-bound}, give $S\subseteq\hat{\mathcal G}$ and $\gamma_{C(\hat{\mathcal G})}\geqslant1-\varepsilon$ at the stated sample count, with total failure probability at most $\delta$. Since $\hat{\mathcal G}\subseteq C(S)$, also $S\subseteq C(\hat{\mathcal G})$, so $S\subseteq\tilde S$. Since $\tilde S\subseteq C(\hat{\mathcal G})$, Observation~\ref{obs:subgroup-purity} gives $\gamma_{\tilde S}\geqslant\gamma_{C(\hat{\mathcal G})}\geqslant1-\varepsilon$.
\end{proof}

\subsection{Considering Bad Samples}
\label{sec:bad-samples}

Sec.~\ref{sec:good-samples} established a lower bound on $M_g^\star$ needed for good$^\star$ samples to contain useful information. Here we bound $M_b^\star$ from above to control bad$^\star$ sample contamination. Within each batch, we have knowledge of bounds $q_g^\star, q_b^\star$ (Table~\ref{tab:quantities-batch}) satisfying $p_g^\star \geqslant q_g^\star$ and $p_b^\star \leqslant q_b^\star$, obtained by means discussed later. We also have an upper bound $t$ on $\nu$ (hence $\nu^\star \leqslant \nu \leqslant t$) from the promise of Problem~\ref{prob:main}. We prove two modes of success per batch: the \emph{intermediate} mode (Algorithm~\ref{alg:propose-intermediate}) and the \emph{clean} mode (Algorithm~\ref{alg:propose-clean}).

Both per-batch proofs condition on a batch with (strictly) more than $K_g$ good$^\star$ samples and (strictly) fewer than $K_b$ bad$^\star$ samples; each proof fixes its own thresholds $K_g$ and $K_b$. The following window on $M$ guarantees both.

\begin{lemma}[Sampling window]
\label{lemma:sampling-window}
Let $K_g, K_b > 0$ and $0 < \delta_1, \delta_2 < 1$, with $q_g^\star \leqslant p_g^\star$ and $q_b^\star \geqslant p_b^\star$. If a batch of $M$ Bell difference samples satisfies
\begin{equation}
\label{eq:constraint-on-M}
\frac{K_g}{q_g^\star}\left(1 + O\left(\sqrt{\frac{\ln \frac{1}{\delta_1}}{K_g}}\right)\right) \leqslant M \leqslant \frac{K_b}{q_b^\star}\max\left(1 - O\left(\sqrt{\frac{\ln \frac{1}{\delta_2}}{K_b}}\right), \frac{\delta_2^{\frac{1}{\ceil{K_b}}}}{e}\right),
\end{equation}
then $\mathbb{P}[M_g^\star \leqslant K_g] \leqslant \delta_1$ and $\mathbb{P}[M_b^\star \geqslant K_b] \leqslant \delta_2$.
\end{lemma}
\begin{proof}
Intuitively, we have $M_g^\star \gtrapprox M q_g^\star$ and $M_b^\star \lessapprox M q_b^\star$. Chernoff gives both bounds when $M$ lies in the window.\footnote{The big-$O$ term on the left-hand side of Eq.~\eqref{eq:constraint-on-M} comes from $\sqrt{2\frac{\ln 1/{\delta_1}}{K_g}}+2\frac{\ln 1/{\delta_1}}{K_g}$ and is applicable when $\frac{\ln 1/{\delta_1}}{K_g} \leqslant O(1)$, as mentioned in the proof of Lemma~\ref{lemma:good-span-C(S)}; the big-$O$ term on the right-hand side comes from $\sqrt{2\frac{\ln 1/{\delta_2}}{K_b}}$, similarly.} The second term in the max is a safeguard in the event that the first term becomes negative; it comes from $\mathbb{P}[M_b^\star \geqslant K_b] \leqslant \binom{M}{\ceil{K_b}} {q_b^\star}^{\ceil{K_b}} \leqslant \left(\frac{e M q_b^\star}{\ceil{K_b}}\right)^{\ceil{K_b}} \leqslant \left(\frac{e M q_b^\star}{K_b}\right)^{\ceil{K_b}} \leqslant \delta_2$.\footnote{We retain $\ceil{K_b}$ in the exponent and only replace it by $K_b$ in the denominator. Thus the safeguard factor is at least $\delta_2/e$ for every $K_b>0$. If the bad-sample bound is zero, there are no bad samples: omit the bad-sample constraint and take $K_b=1$. Formulas dividing by the bad-sample bound concern positive values; at zero, use the good-sample bounds with the positive accuracy floor below.}
\end{proof}

In either mode, our formulas contain relative corrections of order $\sqrt{\ln(1/\delta)/K}$, such as the $O(\cdot)$ terms in Eq.~\eqref{eq:constraint-on-M}. Here $\delta$ is the per-batch failure probability and $K$ is a sample count set by the current progress, such as $K_g\approx\tilde\nu+t$ and $K_b\approx\tilde\nu-t$ in the intermediate mode; since we don't know $\nu^\star$, these counts already use the bound $\nu^\star\leqslant t$. When $K \gg \ln(1/\delta)$, the corrections are negligible and the formulas simplify. Since we can always rerun a batch multiple times within one iteration to boost success probability, it is tolerable to take $\delta \leqslant O(1)$ per batch. The only remaining regime where the corrections cannot be made negligible is therefore $K = O(1)$, which for the intermediate mode happens near the end, when $\tilde\nu-t = O(1)$ (or negative). Accordingly, Lemma~\ref{lemma:per-batch-intermediate} distinguishes $\tilde\nu-t \gg \ln(1/\delta)$ from $\tilde\nu-t \leqslant O(1)$, and Lemma~\ref{lemma:per-batch-clean} collects the corresponding corrections into constants.

\begin{lemma}[Per-Batch Success Condition of Intermediate Mode]
\label{lemma:per-batch-intermediate}
There are two cases:
\begin{description}
  \item[$\pmb{\tilde{\nu} - t \gg \ln\frac{1}{\delta}}$] If $\frac{q_b^\star}{q_g^\star} \lessapprox \frac{\tilde{\nu} - t}{\tilde{\nu} + t}$, take
  \begin{equation}
    M \approx \frac{\tilde{\nu}  + t}{q_g^\star}, \qquad L \approx 2^{t}\, \ln \frac{1}{\delta}
  \end{equation}
  in Algorithm~\ref{alg:stab} with $\mathsf{propose}$ from Algorithm~\ref{alg:propose-intermediate}. Then with probability at least $1 - \delta$, $\tilde{S}$ is successfully incremented by one Pauli. 
  \item[$\pmb{\tilde{\nu} - t \leqslant O(1)}$ or negative] Let $C_1, C_2 > 1$ and let $0<c<1$ be a constant, where $C_1 \approx 1$ when $t \gg 1$. If $\frac{q_b^\star}{q_g^\star} \lessapprox c\frac{1}{\tilde{\nu} + t}$, take
  \begin{equation}
    M \approx C_1 \frac{\tilde{\nu}  + t}{q_g^\star}, \qquad L \approx C_2 2^{t}
  \end{equation}
  in Algorithm~\ref{alg:stab} with $\mathsf{propose}$ from Algorithm~\ref{alg:propose-intermediate}. 
  \begin{itemize}
    \item Either with constant probability, $\tilde{S}$ is successfully incremented by one Pauli,
    \item or we already have $S \subseteq \tilde{S}$.
  \end{itemize}
\end{description}
\end{lemma}

\textit{Proof idea.} We choose $M$ so that, with high probability, the batch contains more than $K_g \approx \tilde{\nu} + t$ good$^\star$ samples---enough for their span to contain $S^\star$ by Lemma~\ref{lemma:good-span-S}---and fewer than $K_b$ bad$^\star$ samples. Each bad$^\star$ sample imposes at most one linear constraint on the radical, so $R^\star \cap S^\star$ has codimension at most $\strictFloor{K_b}$ in $S^\star$; on the other hand, $R^\star \subseteq C(S^\star)$ forces at least a $2^{-\nu^\star} \geqslant 2^{-t}$ fraction of $R^\star$ to lie in $S^\star$. A uniformly random $r \in R \setminus \tilde{S}$ therefore lands in $S^\star \setminus \tilde{S}$ with probability $\gtrapprox 2^{-t}$, and any such $r$ is accepted by the verifier at no purity cost, so $L \approx 2^t \ln \frac{1}{\delta}$ draws suffice. The condition on $\frac{q_b^\star}{q_g^\star}$ is exactly what makes the lower and upper requirements on $M$ compatible. Figure~\ref{fig:radical-draw} shows the two bounds at work.

\begin{figure}[ht]
\centering
\includegraphics{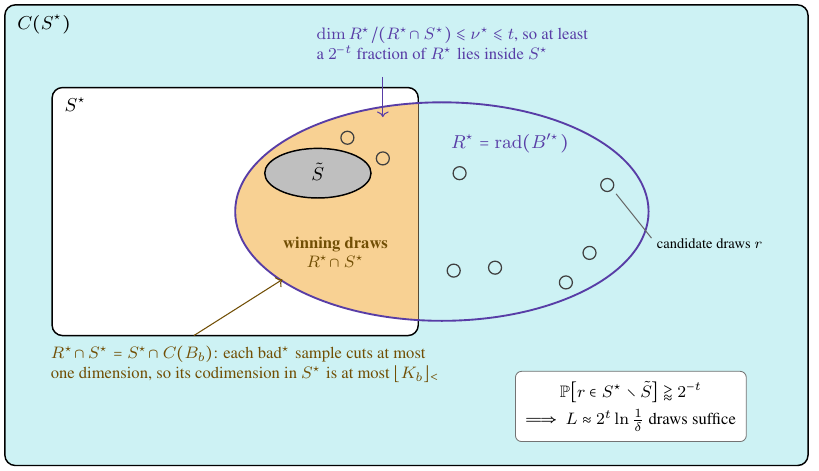}
\caption{Where the $2^t$ of the intermediate mode comes from, in the nesting convention of Fig.~\ref{fig:batch-groups}. Once $S^\star \subseteq B_g^\star$, the radical $R^\star$ lies inside $C(S^\star)$ but straddles $S^\star$. Candidates are drawn from $R \setminus \tilde S$, so the gray region $\tilde S$ of directions already accepted is omitted; a draw landing in the lens $R^\star \cap S^\star$ is guaranteed to be accepted by the verifier at no purity cost. Two bounds squeeze the lens. From above, $R^\star$ cannot escape $S^\star$ by much, so at least a $2^{-t}$ fraction of it lies inside $S^\star$; from below, each bad$^\star$ sample imposes at most one linear constraint, so the lens loses at most $\strictFloor{K_b}$ dimensions. Together they give the $L \approx 2^t \ln\frac{1}{\delta}$ attempts per batch that Algorithm~\ref{alg:propose-intermediate} makes. The regions are not drawn to scale: the guaranteed fraction $2^{-t}$ is far smaller than the lens shown.}
\label{fig:radical-draw}
\end{figure}

\begin{proof}
Our proof is organized into sections. In total, it relies on four different high-probability events: enough good$^\star$ samples ($\delta_1$), few enough bad$^\star$ samples ($\delta_2$), a good$^\star$ span containing $S^\star$ ($\delta_3$), and a useful draw among the $L$ attempts ($\delta_4$). Split the failure probability budget $\delta$ into $\delta_1 + \delta_2 + \delta_3 + \delta_4$; the exact allocation is fixed at the end.

\textit{Conditioning on the sample counts.} In each batch, we draw $M$ Bell difference samples, $M_g^\star$ of which are good$^\star$ and $M_b^\star$ of which are bad$^\star$. We want (strictly) more than $K_g$ good$^\star$ samples, enough for them to be useful, and (strictly) fewer than $K_b$ bad$^\star$ samples, few enough not to be too corrupting, where $K_b$ and $K_g$ will be fixed later. By Lemma~\ref{lemma:sampling-window}, this fails with probability at most $\delta_1+\delta_2$ whenever $M$ lies in the window of Eq.~\eqref{eq:constraint-on-M}. From now on, assume good$^\star$ and bad$^\star$ samples are already decided and $M_g^\star > K_g$, $M_b^\star < K_b$.

\textit{Good$^\star$ sample condition.} Similar to Sec.~\ref{sec:good-samples}, use the $^\star$ superscript to indicate structures coupled to $\tilde{S}$. That is, for each sample $b'_i$ passing post-selection, $b_i'{^\star} = b'_i + \tilde{S}$; for each group $G$, $G^\star = G + \tilde{S}$. Let $B_g$ and $B_b$ denote the spans of good and bad samples, respectively, and $B' = B_g + B_b$ the span of all post-selected samples. (So $B_g^\star = B_g + \tilde{S}$ and $B_g^\star / \tilde{S} = \spn\{b_{g1}^\star, \ldots, b_{gM_g^\star}^\star\}$ where $b_{gi}$ are the good samples; the same applies to $B_b$ and $B'$.)

We set $K_g = \tilde{\nu} + t + \log_2 \frac{1}{\delta_3}$. This way, Lemma~\ref{lemma:good-span-S} applies: with probability $\geqslant 1 - \delta_3$, $S^\star \subseteq B_g^\star$. From now on, assume the samples $b_i$ are drawn and $S^\star \subseteq B_g^\star$.

\textit{The radical in starred form.} Let $R^\star = R + \tilde{S}$. As discussed in Sec.~\ref{sec:post-selection}, uniform $r \in R \setminus \tilde{S}$ induces a uniform non-zero coset in $R^\star / \tilde{S}$. So it suffices to study $R^\star$ instead of $R$. Note $R^\star = (B' \cap C(B')) + \tilde{S} = (B'+\tilde{S}) \cap C(B') = \rad(B'{^\star})$. \footnote{This follows because $B'$ commutes with $\tilde{S}$ and because, for subgroups $A$, $B$, $C$ where $C \leqslant B$, $(A \cap B) + C = (A + C) \cap B$.}

\textit{Surviving target directions, $R^\star \cap S^\star$.} If the $\mathsf{propose}$ function proposes $r \in S^\star / \tilde{S}$, the verifier always accepts because there is no purity loss (Theorem~\ref{thm:stab-ver}, where $\Delta \gamma = 0$ because $\gamma_{r + \tilde{S}} = \gamma_{\tilde{S}} = \gamma_{S^\star}$, by Observation~\ref{obs:add-S-purity}). So we need to bound $(R^\star \cap S^\star) \setminus \tilde{S}$ by bounding $\dim \frac{R^\star \cap S^\star}{\tilde{S}}$. To that end, using $S^\star \subseteq B_g^\star \subseteq B'{^\star}$ and $S^\star \subseteq C(B_g^\star)$,
\begin{equation}
\label{eq:intermediate-bound-intersection}
\begin{aligned}
R^\star \cap S^\star & = B'{^\star} \cap C(B_g+\tilde{S}+B_b) \cap S^\star \\
& =  B'{^\star} \cap C(B_g^\star) \cap C(B_b) \cap S^\star \\
& = S^\star \cap C(B_b) = \{x: x\in S^\star, \expval{x,b} = 0 \; \forall b \in B_b\}
\end{aligned}
\end{equation}
which is $S^\star$ with additional $\dim B_b$ additional constraints, so each bad$^\star$ sample removes at most one independent target direction. Since $\dim B_b < K_b$, we have $\dim B_b \leqslant \strictFloor{K_b}$, where $\strictFloor{K_b} = \ceil{K_b} -1$ is the largest integer smaller than $K_b$. Hence, $\dim (R^\star \cap S^\star) \geqslant \dim S^\star - \strictFloor{K_b}$ and $\dim (R^\star \cap S^\star / \tilde{S}) \geqslant  \tilde{\nu} - \nu^\star - \strictFloor{K_b}$.

\textit{The radical stays close to the target, $\dim R^\star / (R^\star \cap S^\star)$.} Given that $S^\star \subseteq B_g^\star \subseteq B'{^\star}$, $R^\star = \rad(B'{^\star}) \subseteq C(S^\star)$. \footnote{This does not imply $B'^\star \subseteq C(S^\star)$; bad$^\star$ samples may not commute with $S^\star$.} Because any commuting group in $C(S^\star) / S^\star$ has dimension at most $\nu^\star$, $\dim \frac{R^\star + S^\star}{S^\star} \leqslant \nu^\star$ and hence $\dim \frac{R^\star}{R^\star \cap S^\star} = \dim \frac{R^\star + S^\star}{S^\star} \leqslant \nu^\star$ (second isomorphism theorem). Therefore, at least a $2^{-\nu^\star} \geqslant 2^{-t}$ fraction of $R^\star$ lies in $S^\star$.

\textit{Conclusion.} Therefore, each random $r \in R \setminus \tilde{S}$ (or equivalently $r \in R^\star \setminus \tilde{S}$) has success probability:
\begin{equation}
\label{eq:intermediate-bound-delta-4}
\begin{aligned}
\mathbb{P}[r \in S^\star \setminus \tilde{S}] & = \frac{|R^\star \cap S^\star| - |\tilde{S}|}{|R^\star| - |\tilde{S}|} = \frac{\frac{|R^\star \cap S^\star|}{|\tilde{S}|} - 1}{\frac{|R^\star|}{|R^\star \cap S^\star|}\frac{|R^\star \cap S^\star|}{|\tilde{S}|} - 1} \\
& \geqslant \frac{\frac{|R^\star \cap S^\star|}{|\tilde{S}|} - 1}{2^{\nu^\star}\frac{|R^\star \cap S^\star|}{|\tilde{S}|} - 1} \geqslant 2^{-\nu^\star} (1 - 2^{-(\tilde{\nu} - \nu^\star - \strictFloor{K_b})}) 
\end{aligned}
\end{equation}
Two situations:
\begin{description}
  \item[$\pmb{\tilde{\nu} - t \gg \ln \frac{1}{\delta}}$] We will choose the allocation at the end of the proof so that $\ln \frac{1}{\delta_1}, \ln \frac{1}{\delta_2} \ll \tilde{\nu} - t$. Since $K_g$ was set to be $\approx \tilde{\nu} + t \gg \ln \frac{1}{\delta}$, the lower bound of Eq.~\eqref{eq:constraint-on-M} becomes $M \gtrapprox \frac{\tilde{\nu} + t}{q_g^\star}$. So we set $M \approx \frac{\tilde{\nu} + t}{q_g^\star}$ for efficiency. On the bad$^\star$ sample side, we want $K_b$ as large as possible to relax the upper bound on $M$ in Eq.~\eqref{eq:constraint-on-M}. This is done by $K_b = \tilde{\nu} - t - \xi$ for some scale $\xi$ with $1 \ll \xi \ll \tilde{\nu} - t$.\footnote{We substitute $\nu^\star \leqslant t$ already in $K_g$ and $K_b$ rather than at the final step, so that the resulting bound on $q_b^\star/q_g^\star$ depends only on quantities the algorithm observes ($\tilde{\nu}$ and $t$).} Then, by Eq.~\eqref{eq:intermediate-bound-delta-4}, $\mathbb{P}[r \in S^\star \setminus \tilde{S}] \geqslant 2^{-t} (1 - 2^{-\xi}) \approx 2^{-t}$. So setting $L \approx 2^t \ln \frac{1}{\delta_4}$ gives a guaranteed update with probability at least $1 - \delta_4$. Taking the first term of the max, Eq.~\eqref{eq:constraint-on-M} simplifies to $M \approx \frac{\tilde{\nu} + t}{q_g^\star} \lessapprox \frac{\tilde{\nu} - t}{q_b^\star}$, which requires $\frac{q_b^\star}{q_g^\star} \lessapprox \frac{\tilde{\nu} - t}{\tilde{\nu} + t}$.
   
  \item[$\pmb{\tilde{\nu} - t = O(1)}$ or negative] We set $\delta$ (and $\delta_1, \ldots, \delta_4$) to be appropriate constants. Then $K_g = \tilde{\nu} + t + O(1)$, so we take $M \approx C_1\frac{\tilde{\nu} + t}{q_g^\star}$, where constant $C_1 > 1$ and $C_1 \approx 1$ when $\tilde{\nu} + t \gg 1$, i.e., $t \gg 1$.
   
  Unlike in the case $\tilde{\nu} - t \gg \ln \frac{1}{\delta}$, we set $K_b = 1$ such that $\strictFloor{K_b} = 0$, i.e., the batch contains no bad$^\star$ samples. This way, if $\nu^\star < \tilde{\nu}$, by Eq.~\eqref{eq:intermediate-bound-delta-4}, $\mathbb{P}[r \in S^\star \setminus \tilde{S}] \geqslant 2^{-(t+1)}$. Therefore it suffices to take $ L = \left\lceil 2 \times 2^t \ln \frac{1}{\delta_4} \right\rceil = 2^t \times \Theta(1)$; if $\nu^\star = \tilde{\nu}$, then $\dim S^\star = \dim \tilde{S}$, so $S^\star = \tilde{S}$ and we are already done. \footnote{Keeping a nonzero $K_b$ in the proof gives a slightly more general bookkeeping statement that may be useful for some edge-case input regimes. In that case, the final statement becomes: Either we make progress, or we are already $K_b$-close to final success in the sense that $\dim \left(S^\star / \tilde{S}\right) \leqslant K_b$.}
   
  It remains to check that Eq.~\eqref{eq:constraint-on-M} is compatible with this choice. Since the first term in the max in the upper bound of Eq.~\eqref{eq:constraint-on-M} can become negative, we use the second term as a safeguard. With $\delta_2 = \Theta(1)$, this gives the sufficient condition
  \begin{equation}
   M = C_1 \frac{\tilde{\nu} + t}{q_g^\star} \lessapprox C' \frac{1}{q_b^\star}, \qquad 0 < C' < 1
  \end{equation}
  Equivalently, $\frac{q_b^\star}{q_g^\star}\lessapprox c\frac{1}{\tilde{\nu}+t}$ for $c=\frac{C'}{C_1}<1$.
\end{description}

\textit{Failure budget.} In the case $\tilde{\nu} - t \gg \ln \frac{1}{\delta}$, allocate
\begin{equation}
\begin{cases}
  \delta_1 = \delta_2 = \delta_3 = \delta_4 = \frac{\delta}{4}, & \text{if } \ln \frac{1}{\delta} \gg 1 \\
  \delta_4 = (1 - \eta) \delta, \;\; \delta_1 = \delta_2 = \delta_3 = \frac{\eta}{3} \delta, & \text{if } \ln \frac{1}{\delta} = O(1)
\end{cases}
\end{equation}
where $\eta$ is some scale with $1 \ll \ln \frac{1}{\eta} \ll \tilde{\nu} - t$. This keeps $\ln \frac{1}{\delta_1}, \ln \frac{1}{\delta_2} \ll \tilde{\nu} - t$, as used above, and gives $L \approx 2^t \ln \frac{1}{\delta}$. In the case $\tilde{\nu} - t = O(1)$, the $\delta_i$ are the constants chosen there. Either way, $\delta_1+\delta_2+\delta_3+\delta_4 \leqslant \delta$.
\end{proof}

The lemma for the clean mode is similar. Note that in addition to the cases determined by $t$ versus $\ln \frac{1}{\delta}$, the parameters also depend on $\frac{q_b^\star}{q_g^\star}$ versus $\frac{\ln \frac{1}{\delta}}{\chi}$. As discussed above, we take $\delta$ to be a $\Theta(1)$ constant to delay the transition between regimes, thereby absorbing the $\ln \frac{1}{\delta}$ terms into constants on one side.

Assume the verifier is always correct. In the sample bounds, $\gamma_{\tilde S}$ can be replaced by any known lower bound, such as $\gamma-\varepsilon_{\mathrm{acc}}$, provided the accuracy budget is smaller than that bound.

\begin{lemma}[Per-Batch Success Condition of Clean Mode]
\label{lemma:per-batch-clean}
For $0<\Delta\varepsilon\ll\gamma_{\tilde S}$, define $\chi \equiv \max\left\{C_1\frac{t}{\Delta\varepsilon},C_2(\tilde{\nu}+t)\right\}$. We take $M = \frac{\chi}{q_g^\star}$ in Algorithm~\ref{alg:stab} with $\mathsf{propose}$ from Algorithm~\ref{alg:propose-clean}. With probability at least $1 - \delta$, the verifier accepts the whole proposed radical $R$ if queried. If the
update is denoted by $\tilde S\leftarrow \tilde R^\star = R + \tilde{S}$, then
\begin{equation}
\label{eq:clean-mode-bound}
\dim\frac{S^\star}{S^\star \cap \tilde R^\star} \leqslant \strictFloor{K_b}
\end{equation}
where $K_b = C_3 \frac{q_b^\star}{q_g^\star}\chi$ for $q_b^\star>0$ and $\strictFloor{K_b} = \ceil{K_b} - 1$ is the largest integer smaller than $K_b$. In other words, we are $\strictFloor{K_b}$ away from final success. Here, $\delta \leqslant \Theta(1)$, constants $C_1, C_2, C_3 > 1$, $C_1 \approx 1$ when $t \gg \ln \frac{1}{\delta}$, $C_2 \approx 1$ when $\tilde{\nu} + t \gg \ln \frac{1}{\delta}$, and $C_3 \approx 1$ when $\frac{q_b^\star}{q_g^\star} \chi \gg \ln \frac{1}{\delta}$.
\end{lemma}

\textit{Proof idea.} The bookkeeping is the same as in Lemma~\ref{lemma:per-batch-intermediate}, with one more requirement on the good$^\star$ samples: their span must be large enough that its centralizer retains purity $\gamma_{\tilde S} - \Delta\varepsilon_-$ (Lemma~\ref{lemma:good-span-C(S)}). This costs $K_g \approx \frac{t}{\Delta\varepsilon}$ samples and is where $\chi$ comes from. Since the radical sits inside that centralizer, its purity is at least as large, so the verifier accepts all of $R$ at once. The bad$^\star$ samples again cut at most $\strictFloor{K_b}$ dimensions off $S^\star$, which is the distance to final success in Eq.~\eqref{eq:clean-mode-bound}; $K_b$ is then the smallest value compatible with the batch size $M = \chi / q_g^\star$.

\begin{proof}
\textit{Probability bounds.} Split the failure probability as $\delta=\delta_1+\delta_2+\delta_3+\delta_4$. As in the intermediate case, Lemma~\ref{lemma:sampling-window} gives $M_g^\star > K_g$ and $M_b^\star < K_b$ with probability at least $1-\delta_1-\delta_2$ whenever $M$ lies in the window of Eq.~\eqref{eq:constraint-on-M}, where $\delta_1$ and $\delta_2$ are failure probability budgets for the respective conditions, and $K_b$ and $K_g$ will be fixed later. Assume we have a good$^\star$/bad$^\star$/selected-out assignment with such a property.

\textit{Good$^\star$ sample conditions.} We use the same quantities and their properties as in the intermediate mode proof. This time, however, 
\begin{equation}
\label{eq:K_g-clean}
K_g \approx \max\left\{\frac{t}{\Delta \varepsilon}\left(1 + O\left(\sqrt{\frac{\ln \frac{1}{\delta_4}}{t}}\right) + O\left(\frac{\Delta\varepsilon}{\gamma_{\tilde S}}\right)\right), \tilde{\nu} + t + \log_2\frac{1}{\delta_3}\right\},
\end{equation}
where we used $\Delta\varepsilon_-\approx\Delta\varepsilon$ as in Eq.~\eqref{eq:verifier-relative-slack}. For an exact choice of $K_g$, replace the first term of the maximum by the right-hand side of Eq.~\eqref{eq:good-span-exact-bound}, with $\nu^\star=t$, $\delta=\delta_4$, and $\Delta\varepsilon=\Delta\varepsilon_-$; omit that term if $t=0$. Then, conditioned on the good$^\star$/bad$^\star$/selected-out assignment, Lemma~\ref{lemma:good-span-S} applies with failure probability budget $\delta_3$ and Lemma~\ref{lemma:good-span-C(S)} applies with failure budget $\delta_4$ (where $\Delta \varepsilon$ is replaced with $\Delta \varepsilon_-$). Assume these hold; then $S^\star \subseteq B_g^\star$. Once this containment holds, $\mathfrak{B}_g^\star = B_g + S^\star = B_g^\star$, so the coverage bound reads $\gamma_{C(B_g^\star)} \geqslant \gamma_{\tilde S} - \Delta\varepsilon_-$.

\textit{The verifier accepts the radical.} As in the intermediate-mode proof, $R^\star = \rad(B'{^\star})$, so $R^\star \subseteq C(B_g^\star)$. Therefore, by Observation~\ref{obs:subgroup-purity}, $\gamma_{C(B_g^\star)} \leqslant \gamma_{R^\star}$, so $\gamma_{R+\tilde{S}} = \gamma_{R^\star} \geqslant \gamma_{C(B_g^\star)} \geqslant \gamma_{\tilde S} - \Delta \varepsilon_-$. Hence the verifier accepts the whole proposed radical $R$.

\textit{Distance from $S^\star$.} Moreover, Eq.~\eqref{eq:intermediate-bound-intersection} shows $R^\star \cap S^\star = S^\star \cap C(B_b)$, so $\dim \frac{S^\star}{R^\star \cap S^\star} \leqslant \dim B_b \leqslant M_b^\star < K_b$. Or, $\dim \frac{S}{R^\star \cap S} = \dim \frac{S^\star}{R^\star \cap S^\star} \leqslant K_b$ by applying the second isomorphism theorem twice. Since $\dim \frac{S^\star}{R^\star \cap S^\star}$ is an integer,  $\dim \frac{S^\star}{R^\star \cap S^\star} \leqslant \strictFloor{K_b}$. This matches Eq.~\eqref{eq:clean-mode-bound}, with $K_b$ left for us to determine. 

\textit{Choosing $K_b$.} Lastly, we simplify Eqs.~\eqref{eq:constraint-on-M} and~\eqref{eq:K_g-clean} and choose appropriate $\delta_1, \ldots, \delta_4$. Generally speaking, we pick $\delta_1, \delta_2, \delta_3, \delta_4 = \frac{\delta}{O(1)}$, set $M$ to the left-hand side of Eq.~\eqref{eq:constraint-on-M}, and $\text{LHS} \leqslant \text{RHS}$ of Eq.~\eqref{eq:constraint-on-M} gives the smallest $K_b$ we can take. We can write $M = \chi / q_g^\star$ where $\chi = \max\{C_1 \frac{t}{\Delta \varepsilon}, C_2(\tilde{\nu} + t)\}$, for constants $C_1, C_2$, where, by the corrections in Eq.~\eqref{eq:K_g-clean}, $C_1 \approx 1$ when $t \gg \ln \frac{1}{\delta}$ and $\Delta\varepsilon\ll\gamma_{\tilde S}$, and $C_2 \approx 1$ when $\tilde{\nu} + t \gg \ln \frac{1}{\delta}$. So it suffices to solve:
\begin{equation}
\frac{\chi}{q_g^\star} \leqslant \frac{K_b}{q_b^\star}\max\left(1 - O\left(\sqrt{\frac{\ln \frac{1}{\delta_2}}{K_b}}\right), \frac{\delta_2^{\frac{1}{\ceil{K_b}}}}{e}\right)   
\end{equation}
and find the smallest admissible $K_b$. While the maximum on the right-hand side depends on $K_b$, notice it is always $\Theta(1)$ in the regime of interest. Indeed, as discussed in Sec.~\ref{sec:loop-structure}, our philosophy dictates that we set $\delta = \Theta(1)$, so even when the first term in the max becomes too small, the second, $\delta_2^{1/\ceil{K_b}}/e$, is still $\Theta(1)$. Therefore, $\text{LHS} \leqslant \text{RHS}$ of Eq.~\eqref{eq:constraint-on-M} simplifies to $K_b \gtrapprox C_3 \frac{q_b^\star}{q_g^\star} \chi$, where $C_3 > 1$ and $C_3 \approx 1$ when $\frac{q_b^\star}{q_g^\star} \chi \gg\ln \frac{1}{\delta}$. Deciding exact values of $C_1, C_2, C_3$ is not needed for our purposes; for an explicit choice, Appendix~\ref{app:clean-parameters} fixes $\delta_1=\cdots=\delta_4=\frac{1}{16}$ and applies the concentration bounds directly, without optimizing.
\end{proof}
\begin{remark}
\label{remark:per-batch-clean-done}
When $\frac{q_b^\star}{q_g^\star} \chi \ll 1$, $\strictFloor{K_b} = 0$. Then, Eq.~\eqref{eq:clean-mode-bound} implies $S^\star \subseteq R^\star$ and we are done.
\end{remark}

In conclusion, Lemma~\ref{lemma:per-batch-intermediate} describes the intermediate mode of success, where $O(2^t)$ verifier queries are needed to increment $\tilde S$ by one. In contrast, when the clean-mode conditions of Lemma~\ref{lemma:per-batch-clean} hold, the whole radical can be accepted at once, allowing $\tilde S$ to jump directly to within a controlled remaining distance of the true weak stabilizer.

Both lemmas are single-batch statements: they succeed with probability $1-\delta$ for a per-batch budget $\delta$ that we will take to be a constant, and they are phrased in terms of the bounds $q_g^\star$ and $q_b^\star$ rather than the true probabilities $p_g^\star$ and $p_b^\star$. The next section removes both restrictions---boosting the per-batch probability by repetition and computing $q_g^\star, q_b^\star$ from the noise model---and arrives at the main theorems.

\section{End-to-End Guarantees}
\label{sec:end-to-end}

The per-batch lemmas of Sec.~\ref{sec:up-down-per-batch} hold only with constant probability and are stated relative to bounds $q_g^\star$, $q_b^\star$ that we have not yet computed. This section closes both gaps and delivers the main theorems. Sec.~\ref{sec:batch-to-iter} boosts a batch into an \emph{iteration} that updates $\tilde S$ with high probability; for the clean mode this requires splitting the run into a main stage and a terminal stage, giving the final pseudocode in Algorithm~\ref{alg:stab-clean}. Sec.~\ref{sec:success-modes} explains how iterations are chained and how the accuracy budget $\varepsilon$ is spent across them. The last two subsections then instantiate $q_g^\star$, $q_b^\star$ for depolarizing noise (Sec.~\ref{sec:up-down-depo}, Theorems~\ref{thm:main-intermediate-depo} and~\ref{thm:main-clean-depo}) and for arbitrary Pauli noise (Sec.~\ref{sec:up-down-arbi}, Theorems~\ref{thm:main-intermediate-arbi} and~\ref{thm:main-clean-arbi}).

\subsection{From Batches to Iterations}
\label{sec:batch-to-iter}

Lemmas~\ref{lemma:per-batch-intermediate} and \ref{lemma:per-batch-clean} in Sec.~\ref{sec:up-down-per-batch} described success conditions for each batch. As discussed in Sec.~\ref{sec:loop-structure},  an \emph{iteration} consists of running $T$ batches in order to boost the success probability. For the intermediate mode, this is straightforward because any update to $\tilde{S}$ must increase its dimension by 1.

\begin{observation}[Per-Iteration Success Condition of Intermediate Mode]
\label{obs:per-iter-intermediate}
Let $\delta_0$ be the failure probability budget in Lemma~\ref{lemma:per-batch-intermediate}. By running at most $T = \left\lceil\frac{\ln \frac{1}{\delta}}{\ln \frac{1}{\delta_0}}\right\rceil$ batches and taking the first successful update, the probability of success is boosted to $\geqslant 1 - \delta$.     
\end{observation}

The clean mode, however, introduces additional complications. The success event for the clean mode is $\dim \frac{S^\star}{S^\star \cap \tilde{S}} \leqslant \strictFloor{K_b}$. However, we cannot observe $\dim \frac{S^\star}{S^\star \cap \tilde{S}}$ directly. Instead, we can only see the consequence $\dim \tilde{S} \geqslant n - t - \strictFloor{K_b}$. This observable consequence can be boosted:
\begin{observation}[Per-Iteration Success Condition of Clean Mode, Main Stage]
\label{obs:per-iter-clean-main}
Let $\delta_0$ be the failure probability budget in Lemma~\ref{lemma:per-batch-clean}. By running at most $T = \left\lceil \frac{\ln \frac{1}{\delta}}{\ln \frac{1}{\delta_0}} \right\rceil$ batches, we obtain a candidate $\tilde{S}'$ approved by the verifier with $\dim \tilde{S}' \geqslant n -t - \strictFloor{K_b}$ with probability at least $1 - \delta$.
\end{observation}
To implement this idea, in each iteration, we may take the proposed candidate with the largest dimension, or take the first proposed update with dimension $\geqslant n - t - \strictFloor{K_b}$. 

Admittedly, $\dim \tilde{S} \geqslant n - t - \strictFloor{K_b}$ is a weaker statement than $\dim \frac{S^\star}{S^\star \cap \tilde{S}} \leqslant \strictFloor{K_b}$. As a result, we may accept an update with dimension $n - t - \strictFloor{K_b}$, but not necessarily one that is $\strictFloor{K_b}$-close to containing $S^\star$. This issue becomes the most severe at the end of the algorithm, when we need to take $K_b = 1$ so that $\strictFloor{K_b} = 0$. Here, it follows from Lemma~\ref{lemma:per-batch-clean} (and Remark~\ref{remark:per-batch-clean-done}) that $S^\star \subseteq \tilde{S}$ with constant probability. Nevertheless, boosting will weaken this statement to $\dim \tilde{S} \geqslant n -t$, which may be far from success. The solution is to run the iteration, but keep accepting nontrivial approved updates, such that one of them leads to final success with high probability. This comes with a tradeoff: we could accept up to $\min\{t, T\}$ updates in the last iteration ($t$ coming from the starting dimension being at least $n-t$ and each update increasing $\dim \tilde{S}$ by at least 1), so purity loss $\Delta \varepsilon$ accumulates.

\begin{observation}[Per-Iteration Success Condition of Clean Mode, Terminal Stage]
\label{obs:per-iter-clean-terminal}
Suppose the terminal stage starts with $\dim\tilde S\geqslant n-t$, and the clean-mode conditions with $K_b=1$ hold at each attempt. Let $\delta_0$ be the failure probability budget in Lemma~\ref{lemma:per-batch-clean}. By running at most $T = \left\lceil \frac{\ln \frac{1}{\delta}}{\ln \frac{1}{\delta_0}} \right\rceil$ fresh batches and accepting all nontrivial proposed updates that pass the verifier, we reach $S^\star \subseteq \tilde{S}$ with probability at least $1 - \delta$. This iteration may apply at most $\min\{t,T\}$ updates.
\end{observation}

\begin{algorithm}[ht]
\caption{\label{alg:stab-clean} Up-Down Stabilizer Finding Algorithm (Clean Mode)}
\KwIn{Accuracy parameter $\varepsilon$, confidence parameters $\delta$ and $\delta_{\mathrm{ver}}$}
\KwOut{An approximate weak stabilizer $\tilde{S}$ for the density matrix.}
Initialize the adaptive verifier, Algorithm~\ref{alg:stab-ver}, with failure probability $\delta_{\mathrm{ver}}$ and minimum query gap $\Delta\varepsilon_{\mathrm{gap}}$, the smallest gap required by the accuracy schedule specified below; let its guess be $\tilde{S} = \{I\}$\;
\tcp{Main Stage:}
\While{$\dim\tilde S<n-t$}{
    Set $\Delta \varepsilon$, $T$\;
    \For{$i = 1$ \KwTo $T$}{
      Set $M$\;
      Draw $M$ Bell difference samples, say $b = \{b_1, b_2, \ldots, b_M\}$\;
      \textbf{Post-selection:} keep all samples in $C(\tilde{S})$ and remove the others; now we have $b' = \{b_1', b_2', \ldots, b_{M'}'\}$\;
      Set $B' = \spn{b'}$, the span of the selected samples\;
      Run $\mathsf{propose}(B',\Delta \varepsilon)$, which rejects or approves {\bfseries\boldmath candidate update $\tilde{S}'$}\;
      \If{$\mathsf{propose}$ approves {\bfseries\boldmath and $\dim \tilde{S}' \geqslant n - t - \strictFloor{K_b}$}, where $K_b$ is from Lemma~\ref{lemma:per-batch-clean}}{apply the update $\tilde{S} \leftarrow \tilde{S}'$ and break}
    }
    \lIf{no update was applied}{\Return{failure}}
}
\tcp{Terminal Stage:}
Set $\Delta \varepsilon$, $T$\;
\For{$i = 1$ \KwTo $T$}{
    Set $M$\;
    Draw $M$ Bell difference samples, say $b = \{b_1, b_2, \ldots, b_M\}$\;
    \textbf{Post-selection:} keep all samples in $C(\tilde{S})$ and remove the others; now we have $b' = \{b_1', b_2', \ldots, b_{M'}'\}$\;
    Set $B' = \spn{b'}$, the span of the selected samples\;
    Run $\mathsf{propose}(B',\Delta \varepsilon)$\;
    \lIf{verifier approves}{apply the approved update {\bfseries but do not break}}
}
\Return{$\tilde{S}$}\;
\end{algorithm}

Therefore, the up-down algorithm in the clean mode is partitioned into two stages. In the \emph{main stage}, we iteratively grow $\tilde{S}$ using Observation~\ref{obs:per-iter-clean-main}. Once $\dim \tilde{S} \geqslant n - t$, we choose parameters so that $\strictFloor{K_b} = 0$ and reach final success in one iteration. This way, we go to the \emph{terminal stage} and finish our task. The pseudocode for the clean mode is displayed in Algorithm~\ref{alg:stab-clean}.

Allocate failure budgets $\delta_j$ to the main-stage iterations
and $\delta_{\mathrm{term}}$
to the terminal stage, with
$\sum_j\delta_j+\delta_{\mathrm{term}}\leqslant\delta$.
For the asymptotic bounds, we may allocate a constant fraction
of $\delta$ to the terminal stage and comparable shares of
the remaining budget to the main-stage iterations.

To keep the cost polynomial, we impose a lower bound on $\Delta\varepsilon$. If $0<\varepsilon_{\mathrm{tot}}\leqslant\varepsilon$ is the chosen total execution budget, reserve an $o(1)$ fraction of it and arrange that the original update budgets sum to at most the remaining $(1-o(1))\varepsilon_{\mathrm{tot}}$. Since there are at most $n$ nontrivial updates, the choice
\begin{equation}
\label{eq:general-clean-accuracy-floor}
\Delta\varepsilon=\max\left\{\Delta\varepsilon_{\mathrm{original}},
\frac{o(1)\,\varepsilon_{\mathrm{tot}}}{n}\right\},
\end{equation}
with the same reserved fraction, keeps the total loss within $\varepsilon_{\mathrm{tot}}$.

With per-iteration success in hand, what is left is to choose $q_g^\star$, $q_b^\star$, and $\Delta\varepsilon$ at each iteration. These depend on the noise model, which is where we go next: Secs.~\ref{sec:up-down-depo} and~\ref{sec:up-down-arbi} make these choices for depolarizing and arbitrary Pauli noise, and for depolarizing noise, Eq.~\eqref{eq:clean-verifier-gap} gives the resulting verifier gap.

\subsection{Full Learning Algorithm}
\label{sec:success-modes}

We now piece together iterations into a full stabilizer learning algorithm. At the beginning of each iteration, the algorithm observes $\dim \tilde{S}$ and calculates a lower bound $q_g^\star$ on $p_g^\star$ and an upper bound $q_b^\star$ on $p_b^\star$, along with the allocated purity drop $\Delta \varepsilon$. The per-iteration success conditions from Sec.~\ref{sec:batch-to-iter} then ensure that each iteration grows $\tilde{S}$ with high probability, leading to the next.

The exact calculation of $q_g^\star$, $q_b^\star$, and $\Delta \varepsilon$ depends on our assumptions on the input state. For our purposes, suppose that the sampled state $\rho$ is a noisy implementation of an ideal pure state $\rho' = \ket{\psi}\bra{\psi}$:
\begin{equation}
\rho = \mathcal{E}(\rho')
\end{equation}
for some noise channel $\mathcal{E}$.

We analyze two representative settings, in Sec.~\ref{sec:up-down-depo} and Sec.~\ref{sec:up-down-arbi} respectively. The first is depolarizing noise. In this case, the Pauli coefficients are suppressed in a uniform way, and post-selection becomes increasingly effective as $\tilde S$ grows. In particular, the bad$^\star$ probability $p_b^\star$ is suppressed as $\tilde{S}$ grows, allowing progress to become exponentially easier after the first iteration.

The second is arbitrary Pauli noise. Here $\mathcal E$ is any Pauli channel, and we do not assume any favorable structure in the distribution of errors. Consequently, the bad$^\star$ probability need not decay in the same way. This leads to more restrictive noise thresholds for the intermediate and clean modes.

The total time and sample costs of the algorithm have two contributions: the cost of the verifier, Algorithm~\ref{alg:stab-ver}, and the cost of the main up-down algorithm, Algorithms~\ref{alg:stab} and~\ref{alg:stab-clean}. The verifier sample cost follows from Lemma~\ref{lemma:uniform-bell-estimates} and is independent of the number of queries; its runtime still accumulates over those queries. We therefore suppress polynomial factors and focus on the possible exponential dependence on $t$.

\subsection{Guarantees under Depolarizing Noise}
\label{sec:up-down-depo}

This subsection carries out the program of Sec.~\ref{sec:success-modes} for depolarizing noise: we compute $q_b^\star$ and $q_g^\star$, observe that $q_b^\star$ halves with every generator added to $\tilde S$ while $q_g^\star$ remains within a constant factor of its initial value, and turn this into Theorems~\ref{thm:main-intermediate-depo} and~\ref{thm:main-clean-depo}.

Depolarizing noise is one of the simplest and most common noise models. It is defined by
\begin{equation}
\label{eq:depo-noise}
\rho = \mathcal{E}(\rho') \equiv (1 - \lambda) \rho' + \lambda \frac{I}{2^n}
\end{equation}
where the parameter $\lambda$ is given. \footnote{If not, we can estimate $\lambda$ by measuring the purity of $\rho$ with Bell sampling (Lemma~\ref{lemma:bell-sampling-known-quantites}), since $\gamma = \tr(\rho^2) = (1 - \lambda)^2 + \frac{2\lambda - \lambda^2}{2^n} \approx (1 - \lambda)^2$. If this estimate is used to set parameters, its approximation and estimation errors should stay below the available success margin.}

Denote the Pauli coefficients for $\rho$ and $\rho'$ by $a_i$ and $a_i'$, respectively. In other words, $\rho = \sum_{i \in \mathcal{P}^n} a_i \sigma_i$ and $\rho' = \sum_{i \in \mathcal{P}^n} a_i' \sigma_i$. Eq.~\eqref{eq:depo-noise} implies $a_i = (1 - \lambda) a_i' + \lambda \frac{\delta_{i0}}{2^n}$. To paraphrase, $a_i = (1-\lambda) a_i'$ except for the identity element, which remains $a_0 = \frac{1}{2^n}$. Therefore, every bad sample $i \in \mathcal{P}^n \setminus C(S)$ has the same Bell-difference probability:
\begin{equation}
\begin{aligned}
  P_{\mathrm{BellDiff}}(i) &  = 4^n \sum_{j \in \mathcal{P}^n} (- 1)^{\expval{i,j}}a_j^4 && \text{by Eq.~\eqref{eq:bell-diff-dist}} \\ 
   &  = (1-\lambda)^4 \left(4^n \sum_{j \in \mathcal{P}^n} (-1)^{\expval{i,j}} {a'_j}^4\right)  + 4^n \left( \left( \frac{1}{2^n} \right)^4 - \left( \frac{1 - \lambda}{2^n} \right)^4 \right) \\ 
   &  = (1-\lambda)^4 P'_{\mathrm{BellDiff}}(i) + \frac{1 - {(1-\lambda)^4}}{2^{2n}}  && \text{by Eq.~\eqref{eq:bell-diff-dist}} \\ 
   &  = \frac{1 - (1-\lambda)^4}{2^{2n}} && \text{by Proposition~\ref{prop:pure-state-bell-diff}}\\ 
\end{aligned}
\end{equation}
where $P'_{\mathrm{BellDiff}}$ denotes the Bell difference distribution for $\rho'$. Thus, 
\begin{equation}
p_b = |\mathcal{P}^n \setminus C(S)| \times \frac{1 - (1-\lambda)^4}{2^{2n}} \leqslant 1 - (1-\lambda)^4 
\end{equation}
because $|\mathcal{P}^n \setminus C(S)| \leqslant |\mathcal{P}^n|$. So we may use $q_b \equiv 1 - (1- \lambda)^4 \leqslant 4 \lambda$ as an upper bound on $p_b$, the probability of bad samples before post-selection. Then, by the same logic:
\begin{equation}
\label{eq:depo-q-b}
p_b^\star = |C(\tilde{S}) \setminus C(S)| \times \frac{1 - (1-\lambda)^4}{2^{2n}} \leqslant \frac{q_b}{2^{\dim \tilde{S}}}
\end{equation}
where we used $|C(\tilde{S}) \setminus C(S)| \leqslant |C(\tilde{S})| = 2^{2n - \dim \tilde{S}}$. In other words, the upper bound $q_b^\star$ for $p_b^\star$ is set as $q_b^\star = \frac{q_b}{2^{\dim \tilde{S}}} = \frac{1 - (1-\lambda)^4}{2^{\dim \tilde{S}}}$.

So $q_b^\star$ decreases exponentially in the dimension of $\tilde{S}$. At the same time, $q_g^\star$, the lower bound on $p_g^\star$, largely remains the same:
\begin{equation}
p_g^\star \geqslant \gamma_{S^\star}^2 = \gamma_{\tilde{S}}^2 \geqslant (\gamma - \varepsilon_{\text{acc}})^2 \geqslant ((1 - \lambda)^2 - \varepsilon_{\text{acc}})^2 \equiv q_g^\star
\end{equation}
Here, the first inequality follows from Lemma~\ref{lemma:prob-diff-lower}, $\gamma_{S^\star} = \gamma_{\tilde{S}}$ by Observation~\ref{obs:add-S-purity}, and $\varepsilon_{\text{acc}}$ is the accumulated purity-loss budget from all previous iterations: an upper bound on their accumulated purity loss, obtained by summing the budgets of accepted updates. 

In a nutshell, as we increment $\tilde{S}$, the concentration of bad$^\star$ samples drastically decreases, while good$^\star$ samples remain abundant. This works strongly in our favor: as long as we are able to make the first move towards constructing $\tilde{S}$, everything that follows becomes much easier. Hence, we prove Theorem~\ref{thm:main-intermediate-depo} for the intermediate mode, and Theorem~\ref{thm:main-clean-depo} for the clean mode.

For convenience of proof, assume $n \gg 1$. 

\begin{theorem}[Intermediate Regime Under Depolarizing Noise]
\label{thm:main-intermediate-depo}
Assume $n - t \gg 1$. When $\frac{1-(1-\lambda)^4}{(1-\lambda)^4} \lessapprox \frac{n-t}{n+t}$, which simplifies to $\lambda \lessapprox 1 - 2^{-1/4} \approx 0.159$ when $n \gg t$, the intermediate mode of the up-down algorithm succeeds. The algorithm uses $O(\mathrm{poly})$ samples and $O(2^t \cdot \mathrm{poly})$ time to return an approximate weak stabilizer.
\end{theorem}

\textit{Proof idea.} The first iteration sets the noise threshold in our analysis; later iterations become easier. Plugging $q_b^\star = 1-(1-\lambda)^4$, $q_g^\star = (1-\lambda)^4$, and $\tilde{\nu} = n$ into Lemma~\ref{lemma:per-batch-intermediate} gives the stated threshold. Every later iteration faces a weaker condition, because $q_b^\star$ halves with each accepted generator while $q_g^\star$ and $\frac{\tilde{\nu}-t}{\tilde{\nu}+t}$ barely move; by the time the lemma switches to its $\tilde{\nu} - t = O(1)$ case and loses a polynomial factor, the exponential suppression $2^{-(n-t)}$ has long paid for it. Figure~\ref{fig:qb-race} plots both effects.

\begin{figure}[ht]
\centering
\includegraphics{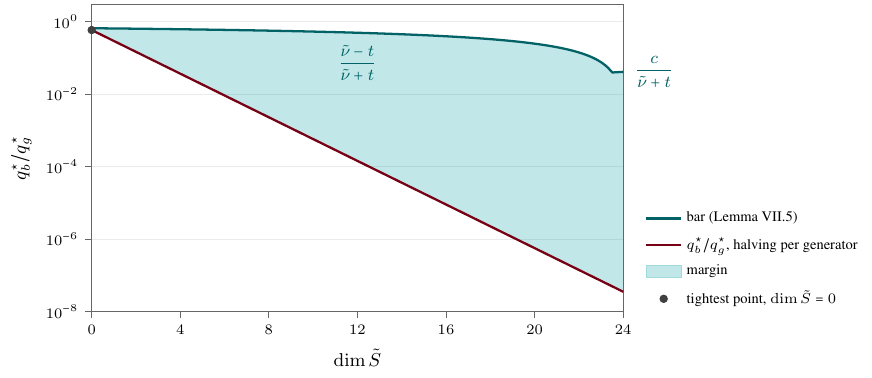}
\caption{The intermediate mode under depolarizing noise, drawn for illustrative $n=30$, $t=6$, $\lambda=0.11$ and $c=0.5$. Every accepted generator halves $q_b^\star$ by Eq.~\eqref{eq:depo-q-b} while $q_g^\star$ stays within a constant factor of its initial value (the plot holds it there, as in the leading-order analysis), so the ratio the finder must clear falls by a factor of two per step---a straight line on this logarithmic axis. The bar, below which Lemma~\ref{lemma:per-batch-intermediate} guarantees success, is the envelope of that lemma's two cases: $\frac{\tilde\nu-t}{\tilde\nu+t}$ through the main stage, flattening onto the second case's floor $\frac{c}{\tilde\nu+t}$, which is what keeps the requirement from collapsing to zero as $\tilde\nu\to t$. Those two cases differ in their conclusions as well as their thresholds, so the bar is an envelope rather than a single condition. The initial step, $\dim\tilde S=0$, imposes the strongest condition: there it reduces to the hypothesis of Theorem~\ref{thm:main-intermediate-depo}, and every later step has more room. The clean mode races differently: there the same ratio enters through $K_b$ in Lemma~\ref{lemma:per-batch-clean}, measuring how far an accepted update still falls short of $S^\star$ rather than whether one is possible. The shape here does not depend on the chosen parameters---changing $\lambda$ slides the line up or down without bending it.}
\label{fig:qb-race}
\end{figure}

\begin{proof}
\leavevmode\par
\textit{First iteration.} At the first iteration, $q_b^\star = q_b = 1 - (1 - \lambda)^4$, $q_g^\star = (1 - \lambda)^4$, and $\tilde{\nu} = n$. So, the first case in Lemma~\ref{lemma:per-batch-intermediate} requires
\begin{equation}
\label{eq:intermediate-condition-first}
\frac{1-(1-\lambda)^4}{(1-\lambda)^4} \lessapprox \frac{n-t}{n+t}.
\end{equation}
When $n \gg t$, $\text{RHS} \approx 1$, so the condition simplifies to $\lambda \lessapprox 1 - 2^{-1/4} \approx 0.159$. 

\textit{All later iterations.} Once the first iteration is satisfied, the conditions for all later iterations are weaker and therefore the algorithm succeeds. Indeed, when $\tilde{\nu} - t = \omega(1)$, the condition needed is:
\begin{equation}
\label{eq:intermediate-condition-main}
\frac{q_b^\star}{q_g^\star} \lessapprox \frac{\tilde{\nu}-t}{\tilde{\nu}+t} 
\end{equation}
where $q_b^\star = (1-(1-\lambda)^4) 2^{-\dim \tilde{S}}$ and $q_g^\star = \left((1-\lambda)^2-\varepsilon_{\text{acc}}\right)^2$. After each successful iteration,
\begin{itemize}
  \item $q_b^\star$ decreases by a factor of $2$;
  \item since the intermediate-mode condition does not depend on $\Delta \varepsilon$, we can limit $\Delta \varepsilon$ to sufficiently small values so that $q_g^\star$ does not change significantly. For example, take $\Delta \varepsilon = \frac{\varepsilon_{\text{tot}}}{n}$ with total maximum purity loss $\varepsilon_{\text{tot}} = \min\{\varepsilon,o((1-\lambda)^2)\} = \min\{\varepsilon,o(1)\}$, where $\varepsilon$ is the input accuracy parameter. 
  \item the RHS $\frac{\tilde{\nu}-t}{\tilde{\nu}+t}$ does not change significantly because of the assumption $\tilde{\nu} - t = \omega(1)$.
\end{itemize}
So the condition becomes weaker at this stage. 

Near the end, when $\tilde{\nu} - t = O(1)$ and we switch to the second case in Lemma~\ref{lemma:per-batch-intermediate}, the condition becomes:
\begin{equation}
\frac{q_b^\star}{q_g^\star} \lessapprox c\frac{1}{\tilde{\nu}+t}, \qquad c=\Theta(1)
\end{equation}
The RHS of the inequality is reduced to $\Omega(\frac{1}{t})$. Compared to Eq.~\eqref{eq:intermediate-condition-first}, this is reduced by a factor at most $O\left(\frac{\frac{n-t}{n+t}}{\frac{1}{t}}\right) = O\left((n-t)\frac{t}{n+t}\right) \leqslant O(n-t)$. Fortunately, by this point $\dim \tilde{S} \geqslant n-t-O(1)$, so the LHS has already been suppressed by a factor of $\Omega(2^{-(n-t)})$. Since we assume $n - t \gg 1$, this exponential suppression beats the additional polynomial loss. Thus the terminal condition remains weaker than the first iteration condition.  

\textit{Cost analysis.} In each batch, the up-down algorithm takes polynomially many fresh Bell difference samples and makes $\tilde{O}(2^t)$ queries to the verifier, where $\tilde{O}$ omits polylogarithmic factors. With $\Delta\varepsilon=\varepsilon_{\mathrm{tot}}/n$, Lemma~\ref{lemma:uniform-bell-estimates} lets the verifier handle all these queries with one polynomial-size dataset. Since there are at most $n$ updates and $O(n\ln(n/\delta))$ batch attempts, the total sample complexity is $O(\mathrm{poly})$, while the time complexity is $O(2^t \cdot \mathrm{poly})$.
\end{proof}

\begin{theorem}[Main Theorem for Clean Regime Under Depolarizing Noise]
\label{thm:main-clean-depo}
Assume $n - t - (\log_2 t + \log_2 \log_2 t) \gg 1$.\footnote{For $t=1$, where $\log_2\log_2 t$ is undefined, interpret this assumption as $n-t\gg1$, here and in the proof.} When $\frac{1-(1-\lambda)^4}{(1-\lambda)^4} \lessapprox \frac{n-t}{n+t}$, and $C \frac{t}{n-t} \frac{1-(1-\lambda)^4}{(1-\lambda)^4} \lessapprox \varepsilon \ll 1$, the clean mode of the up-down algorithm succeeds, where constant $C>1$ and $C\approx 1$ when $t\gg 1$. The algorithm uses $O(\mathrm{poly})$ samples and $O(\mathrm{poly})$ time to return an approximate weak stabilizer.
\end{theorem}

\textit{Proof idea.} The clean mode spends accuracy budget: an iteration that jumps $\dim \tilde{S}$ from $\tilde{s}'$ to $\tilde{s}$ costs $\Delta\varepsilon \approx C_1 t / \chi$, where $\chi$ is the batch size the jump needs. Since $\varepsilon\ll1$ while the purity stays above a constant, each update's budget is small compared with the purity, as Lemma~\ref{lemma:per-batch-clean} requires, and together with $t\gg1$ this gives $C\approx1$. The first jump, $0 \rightarrow \Lambda$ with $1 \ll \Lambda \ll n - t$, costs $\approx C_1 \frac{t}{n-t} \frac{q_b}{q_g}$---the accuracy threshold in the statement---and keeps $q_g^\star$ within a constant factor of its initial value. Later jumps are exponentially cheaper because $q_b^\star$ has been suppressed by $2^{-\Lambda}$, and the terminal stage is cheaper still provided $t(n-t)2^{-(n-t)} \ll 1$, which is the assumption on $n - t$. So the first jump dominates the total. Finally, we raise the exponentially small $\Delta\varepsilon$'s to a floor $o(1)\,\varepsilon/N_{\mathrm{tot}}$ so that the sample cost stays polynomial.

\begin{proof}
We want the clean mode to succeed at all steps under the conditions from Lemma~\ref{lemma:per-batch-clean}. Due to boosting as discussed in Sec.~\ref{sec:batch-to-iter}, we first increment $\tilde{S}$ until $\dim \tilde{S}$ is sufficiently large (main stage, Observation~\ref{obs:per-iter-clean-main}), and then finish it off in one iteration containing multiple updates (terminal stage, Observation~\ref{obs:per-iter-clean-terminal}).

\textit{Cost of one jump.} Each iteration spends accuracy budget, and we must show that the total stays within $\varepsilon$. Let $\Delta \varepsilon (\tilde{s}' \rightarrow \tilde{s})$ be the per-update budget that certifies a jump from $\dim \tilde{S} = \tilde{s}'$ to $\dim \tilde{S} \geqslant \tilde{s}$ in one iteration. In Lemma~\ref{lemma:per-batch-clean}, a smaller budget requires a larger batch scale $\chi$, which raises $K_b = C_3\frac{q_b^\star}{q_g^\star}\chi$. The jump needs $\strictFloor{K_b}\leqslant\tilde\nu-t$ with $\tilde\nu=n-\tilde s$ (Observation~\ref{obs:per-iter-clean-main}), i.e., $K_b\leqslant\tilde\nu-t+1$, so we take the largest $\chi$ this allows. This gives:
\begin{equation}
\label{eq:clean-mode-purity-loss-per-iter}
\Delta \varepsilon(\tilde{s}' \rightarrow \tilde{s}) = \begin{cases}
C_1 \frac{t}{\chi_{\text{m}}(\tilde{s}' \rightarrow \tilde{s})} & \text{if } C_2 (\tilde{\nu}' + t) \leqslant \chi_{\text{m}}(\tilde{s}' \rightarrow \tilde{s}) \\
\infty & \text{otherwise}  
\end{cases} 
\end{equation}
where $\chi_{\text{m}}(\tilde{s}' \rightarrow \tilde{s}) = \frac{1}{C_3} K_{b\text{m}}(\tilde{s})\frac{q_g^\star(\tilde{s}')}{q_b^\star(\tilde{s}')}$ for $K_{b\text{m}}(\tilde{s}) = \tilde{\nu} - t + 1$, $q_g^\star(\tilde{s}') = \left((1-\lambda)^2 - \varepsilon_{\text{acc}}(\tilde{s}')\right)^2$ with $\varepsilon_{\text{acc}}(\tilde{s}')$ the purity-loss budget accumulated before the jump, and $q_b^\star(\tilde{s}') = \frac{q_b}{2^{\tilde{s}'}}= \frac{1-(1-\lambda)^4}{2^{\tilde{s}'}}$. Here, $\tilde{\nu} = n - \tilde{s}$ and $\tilde{\nu}' = n - \tilde{s}'$. The constants $C_1,C_2,C_3$ are from Lemma~\ref{lemma:per-batch-clean}. The $\infty$ case means that even the largest allowed $\chi$ is below the lemma's minimum $C_2(\tilde\nu'+t)$, so our bounds certify no such jump. The first condition of the theorem forces $\lambda\lessapprox1-2^{-1/4}$, so $\gamma\geqslant(1-\lambda)^2\gtrapprox1/\sqrt2$. Every budget we use is at most $\varepsilon\ll1$, and $\gamma_{\tilde S}\geqslant\gamma-\varepsilon\approx\gamma$, so $\Delta\varepsilon\ll\gamma_{\tilde S}$ throughout: Lemma~\ref{lemma:per-batch-clean} applies at every step, with $C_1\approx1$ when $t\gg1$. The same bound keeps the total, floor included, below the cap $(1-\lambda)^2/4\gtrapprox0.18$ of Assumption~\ref{asp:accuracy}.

An iteration may overshoot its target dimension $\tilde s$. This can only help the clean-mode analysis, since $q_b^\star$ decreases, the number of good${}^\star$ samples needed in Lemma~\ref{lemma:good-span-S} decreases, and the remaining conditions, including the budget-based bound $q_g^\star$, are unchanged. Remark~\ref{remark:clean-bookkeeping} shows how to keep the bookkeeping exact.

\textit{The first jump.} We now show that under appropriate assumptions, the total purity loss is dominated by the first update, just as in the intermediate mode. Consider the transition $\tilde{s} = 0 \rightarrow \Lambda$:
\begin{equation}
\Delta \varepsilon(0 \rightarrow \Lambda) = \begin{cases}
C_1 C_3 \frac{t}{n-t+1-\Lambda} \frac{q_b}{q_g} & \text{if } \frac{q_b}{q_g} \leqslant \frac{1}{C_2C_3} \frac{n-t+1-\Lambda}{n+t} \\
\infty & \text{otherwise} \\
\end{cases}
\end{equation}
where $q_g = (1-\lambda)^4$, $q_b = 1-(1-\lambda)^4$, and we used $C_2 \approx 1$ because $\tilde{\nu}' + t = n + t \gg 1$. Since $n - t \gg 1$, we may pick $\Lambda$ such that $1 \ll \Lambda \ll n - t$. In this way, $C_3 \approx 1$, and the first jump is possible so long as $\frac{q_b}{q_g} \lessapprox \frac{n-t}{n+t}$, where
\begin{equation}
\Delta \varepsilon(0 \rightarrow \Lambda) \approx C_1 \frac{t}{n-t} \frac{q_b}{q_g}
\end{equation}
To leading order, this is also the tightest first-jump condition these bounds give.

\textit{$q_g^\star$ stays within a constant factor.} Before turning to later jumps, we check that $q_g^\star$ loses at most a constant factor whenever $\varepsilon_{\text{acc}} \approx C_1 \frac{q_b}{q_g} \frac{t}{n-t}$ is maintained. We have
\begin{equation}
\left\{\begin{gathered}
\varepsilon_{\text{acc}} \lessapprox C_1 \frac{t}{n+t} \\
(1-\lambda)^2 \gtrapprox \sqrt{\frac{n+t}{2n}} \\
\end{gathered}\right.
\end{equation}
Both lines utilize $\frac{q_b}{q_g} = \frac{1-(1-\lambda)^4}{(1-\lambda)^4} \lessapprox \frac{n-t}{n+t}$. Substituting $\xi \equiv \frac{t}{n}$ gives
\begin{equation}
\frac{\varepsilon_{\text{acc}}}{(1-\lambda)^2} \lessapprox C_1 \sqrt{\frac{2\xi^2}{(1+\xi)^3}}
\end{equation}
If $t \gg 1$, then $C_1 \approx 1$, and since $0 < \xi \leqslant 1$, the RHS is at most $\frac{1}{2}$ to leading order. If $t=O(1)$, then $\xi\ll 1$, so $\text{RHS} \ll 1$. Hence $q_g^\star=((1-\lambda)^2-\varepsilon_{\text{acc}})^2 \gtrapprox \frac{1}{4}(1-\lambda)^4=\frac{1}{4}q_g$.

\textit{Later jumps are negligible.} To make the domination by the first jump precise, view the main stage as choosing target dimensions $0 = \tilde{s}_0 < \tilde{s}_1 < \tilde{s}_2 < \cdots$, where $\tilde{s}_1 = \Lambda$, and write $K_j \equiv K_{b\text{m}}(\tilde{s}_j) = n - t - \tilde{s}_j + 1$. The first jump has size $\Lambda \gg 1$, so before the second jump the bad$^\star$ rate has already been suppressed by a factor $2^{-\Lambda} \ll 1$. Therefore, when choosing $\tilde{s}_2$, we may decrease $K_2$ from $K_1$ by any factor $R_1$ satisfying $1 \ll R_1 \ll 2^\Lambda$. This still makes $\chi_{\text{m}}(\tilde{s}_1 \rightarrow \tilde{s}_2)$ larger than $\chi_{\text{m}}(0 \rightarrow \Lambda)$ by a factor $\gg 1$, because the exponential suppression of $q_b^\star$ more than compensates for the decrease in $K_{b\text{m}}$.

The same argument then repeats. As long as $K_j \gg 1$, choose $K_{j+1}$ smaller than $K_j$ by a factor $R_j \gg 1$, but still much smaller than $2^{\tilde{s}_{j} - \tilde{s}_{j-1}}$, the previous reduction factor in $q_b^\star$. Since $K_j - K_{j+1} = \tilde{s}_{j+1} - \tilde{s}_j$, this choice also makes the next dimension jump $\gg 1$, and hence prepares another suppression factor $2^{-(\tilde{s}_{j+1} - \tilde{s}_j)} \ll 1$ for the following step. The remaining quantities in Eq.~\eqref{eq:clean-mode-purity-loss-per-iter}, such as $q_g^\star$, $C_1, C_2, C_3$, and $\tilde{\nu}' + t$, change only by constant factors or in the favorable direction. Thus each later $\chi_{\text{m}}$ increases by a factor $\gg 1$, so each corresponding $\Delta\varepsilon$ decreases by a factor $\gg 1$. Hence the remaining purity-loss tail is negligible, and $\sum \Delta\varepsilon$ is dominated by the first jump from $\tilde{s} = 0$ to $\tilde{s} = \Lambda$.

\textit{Terminal stage.} This continues until $K_{b\text{m}} = O(1)$, i.e., $n - t - \dim \tilde{S} = O(1)$. At this point, we may take $n-t$ as the final main-stage target, i.e., $K_{b\text{m}} = 1$. A successful update then gives $\dim\tilde S\geqslant n-t$, which ends the main stage, and after this update we enter the terminal stage. Indeed, compared to $\chi_\text{m}(0 \rightarrow \Lambda)$, $\chi_\text{m}$ for this final main-stage jump and at the terminal stage is larger by a factor $\propto \frac{2^{n-t}}{n-t}$, while the $(\tilde{\nu}'+t)$ term only decreases. So $n-t \gg 1$ ensures success in both---the same reasoning as in the intermediate mode.

\textit{Total accuracy expenditure.} Adding up, the total purity loss is at most
\begin{equation}
\label{eq:clean-mode-total-purity-loss}
\varepsilon_{\text{acc}}(n-t)  + \min\left\{t, O\left(\ln \frac{1}{\delta}\right)\right\} \Delta \varepsilon_{\text{terminal}},
\end{equation}
where $\varepsilon_{\text{acc}}(n-t)$ is the main-stage purity-loss budget accumulated to reach $\dim\tilde S\geqslant n-t$, and $\Delta\varepsilon_{\text{terminal}}$ is the max purity loss for one batch in the terminal stage. The terminal multiplier counts its updates, at most $\min\{t,T\}$ with $T=O(\ln\frac1\delta)$ (Observation~\ref{obs:per-iter-clean-terminal}). The main-stage part is dominated by the first jump. The terminal part carries the extra factor $\min\{t,O(\ln \frac{1}{\delta})\}$, so to make it negligible in the leading-order purity loss, we need the terminal $\chi_{\text{m}}$ to be larger than the first-jump $\chi_{\text{m}}$ by a factor $\gg t$. Since the terminal $\chi_{\text{m}}$ is larger by a factor proportional to $\frac{2^{n-t}}{n-t}$, it is enough to require $t(n-t)2^{-(n-t)}\ll 1$, which is guaranteed by $n-t-(\log_2 t+\log_2\log_2 t)\gg 1$. Under this assumption, both the later main-stage jumps and the terminal-stage contribution are asymptotically negligible, so the first jump dominates.

\textit{Keeping the cost polynomial.} If we strictly follow the above strategy, the cost may be exponential: the number of samples required by each batch and each verifier query grows polynomially with $\frac{1}{\Delta\varepsilon}$. As $\Delta\varepsilon$ becomes exponentially small, the cost grows exponentially. Let $N_{\mathrm{tot}}$ be the upper bound on nontrivial accepted updates given by the number of main-stage targets plus the terminal-stage bound $\min\{t,T\}$, capped at $n$. Then
\begin{equation}
N_{\mathrm{tot}}=O\!\left(\min\left\{n,\log(n-t+1)+\ln\frac{1}{\delta}\right\}\right).
\end{equation}
Reserve an $o(1)$ fraction of the input accuracy for a threshold floor: choose the original budgets to sum to at most the remaining $(1-o(1))\varepsilon$, and replace each by
\begin{equation}
\label{eq:clean-accuracy-floor}
\Delta\varepsilon=\max\left\{\Delta\varepsilon_{\mathrm{original}},
\frac{o(1)\,\varepsilon}{N_{\mathrm{tot}}}\right\},
\end{equation}
with the same reserved fraction. The total budget increases by at most the reserve and remains within $\varepsilon$. This requires the corresponding vanishing margin in the accuracy condition. With an inverse-logarithmic reserve, both the verifier and finder costs are polynomial. Appendix~\ref{app:clean-parameters} gives one explicit schedule along these lines, with an additive floor instead of the maximum, and shows that it recovers the conditions of the theorem given a fixed margin.
\end{proof}
\begin{remark}
When $1 \ll t \ll n$ and $\lambda \ll 1$, the conditions in Theorem~\ref{thm:main-clean-depo} simplify to:
\begin{equation}
\lambda \lessapprox 1 - 2^{-1/4}, \qquad 4 \lambda \frac{t}{n} \lessapprox \varepsilon \ll 1
\end{equation}
In the weak-noise limit, the window for $\varepsilon$ in Theorem~\ref{thm:main-clean-depo} is open when $\lambda t/(n-t)\ll1$. For $t\ll n$ this is automatic. The theorem does not need $t\ll n$, though: suitably small noise also accommodates $t=\Theta(n)$. Bare $\lambda\ll1$, however, does not give a uniform guarantee as $t$ approaches $n$.
\end{remark}

\begin{remark}[Sufficient-budget bookkeeping]
\label{remark:clean-bookkeeping}
The proof of Theorem~\ref{thm:main-clean-depo} follows one feasible schedule of target dimensions. One can also ask for the smallest main-stage budget that Lemma~\ref{lemma:per-batch-clean} and Observation~\ref{obs:per-iter-clean-main} certify for reaching $\dim\tilde S\geqslant\tilde s$. Taking $\varepsilon_{\text{acc}}(\tilde s)$ to be the smallest cumulative purity-loss budget certified by such a schedule, it satisfies
\begin{equation}
\label{eq:clean-mode-acc-purity-loss}
\varepsilon_{\text{acc}}(0) = 0, \qquad
\varepsilon_{\text{acc}}(\tilde{s}) = \min_{\tilde{s}' < \tilde{s}} \left\{\varepsilon_{\text{acc}}(\tilde{s}') + \Delta \varepsilon (\tilde{s}' \rightarrow \tilde{s})\right\},
\end{equation}
with $\Delta\varepsilon(\tilde s'\rightarrow\tilde s)$ from Eq.~\eqref{eq:clean-mode-purity-loss-per-iter}, and the schedule in the proof bounds it from above. This is sufficient-budget bookkeeping, not an optimality statement about the algorithm: ``smallest'' refers only to what our bounds certify. A schedule certified for one budget also works for any larger $\varepsilon$, although larger accuracy budgets may allow more cost-efficient choices.

The bookkeeping treats each iteration as landing exactly on its target dimension, while Observation~\ref{obs:per-iter-clean-main} only guarantees dimension at least $\tilde s$. For the proof, we may replace the update by an $\tilde{s}$-dimensional subgroup that contains the previously accepted group and lies inside the updated accepted group; extending a basis of the former supplies one, and its incremental purity loss is no larger than that of the full update (Observation~\ref{obs:subgroup-purity}). This is only a proof device: the algorithm keeps the full update, which can only help.
\end{remark}

For completeness, the polynomial factors can be made explicit for the single-dataset verifier above. For the clean and intermediate modes, assume $1\leqslant t\leqslant n$, $0<\varepsilon\leqslant1$, constant lower bounds on retained purity and good-sample probability, and sufficient margin in the depolarizing success conditions for the stated schedule and inverse-logarithmic slacks. The following pairs give upper bounds on (state copies, classical bit operations), respectively:
\begin{equation}
\label{eq:explicit-learning-costs}
\begin{aligned}
\text{clean:}\quad&
\left(\widetilde O(n/\varepsilon^2),\ \widetilde O(n^3/\varepsilon^2)\right),\\
\text{intermediate:}\quad&
\left(\widetilde O(n^3/\varepsilon^2),\ \widetilde O(2^t n^5/\varepsilon^2)\right),\\
\text{top-down:}\quad&
\left(\widetilde O(4^t n+n^3/\varepsilon^2),\ \widetilde O(4^t n^5/\varepsilon^2)\right).
\end{aligned}
\end{equation}
Here $\widetilde O$ suppresses logarithmic factors in $n$, $1/\varepsilon$, $1/\delta$, and $1/\delta_{\mathrm{ver}}$. These are bounds for this implementation; Appendix~\ref{app:verifier-complexity} gives the accounting. The top-down bound only requires the purity lower bound, not the depolarizing success conditions.

\subsection{Guarantees under Arbitrary Pauli Noise}
\label{sec:up-down-arbi}

With depolarizing noise, the bad$^\star$ probability $p_b^\star$ is suppressed exponentially as we increase $\tilde{S}$. For an arbitrary Pauli channel, however, this is not necessarily the case. In particular, $p_b^\star$ may remain the same even as $\tilde{S}$ grows. A particularly stark example is \emph{single-Pauli noise}, a Pauli channel whose noise is concentrated on a single Pauli error:
\begin{equation}
\rho = \mathcal{E}(\rho') \equiv (1 - \lambda) \rho' + \lambda \sigma_k \rho' {\sigma_k}^\dagger
\end{equation}
for some $k \in \mathcal{P}^n \setminus C(S)$. A direct calculation gives
\begin{equation}
P_{\mathrm{BellDiff}}(i) = (1-p_b) P'_{\mathrm{BellDiff}}(i) + p_b P'_{\mathrm{BellDiff}}(i+k)  
\end{equation}
where $p_b = \frac{1 - (1-2\lambda)^4}{2}$ and hence $p_g = 1 - p_b = \frac{1 + (1-2\lambda)^4}{2}$. In other words, the resulting distribution is a mixture of the pure-state distribution and a shifted version of it. This means that $p_b^\star$ and $p_g^\star$ remain largely unchanged as $\tilde{S}$ grows while commuting with $\sigma_k$. If an element of $S$ that anticommutes with $\sigma_k$ is added to $\tilde{S}$, $p_b^\star$ drops to zero.

More concretely, when $\tilde{S}$ commutes with $\sigma_k$, we find:
\begin{subequations}
\label{eq:arbi-bounds}
\begin{empheq}[left=\empheqlbrace]{gather}
p_g^\star \geqslant p_g\left(1-\frac{\varepsilon_{\mathrm{acc}}}{\gamma}\right)^2 \label{eq:arbi-good}
\\
p_b\left(1-\frac{\varepsilon_{\mathrm{acc}}}{\gamma}\right)^2 \leqslant p_b^\star \leqslant p_b \label{eq:arbi-bad}
\end{empheq}
\end{subequations}
because $|a'_{i+s}|=|a'_i|$ for $s\in S$ (each $s\in S$ stabilizes $\psi$ up to sign, so coefficients on a coset of $S$ can differ only in sign), and each coset of $S$ (a) either entirely survives or gets entirely filtered by $\mathcal{T}_{\tilde{S}}$, and (b) exactly half commutes with $k$ and half does not. So $\frac{\gamma_{\tilde{S}}}{\gamma_{\tilde{S}}'} = \frac{1}{2}\left(1+(1-2\lambda)^2\right) = \gamma$. Therefore, a maximum purity loss of $\varepsilon_{\text{acc}}$ in $\rho$ translates to a maximum purity loss of $\frac{\varepsilon_{\text{acc}}}{\gamma}$ in $\rho'$. Then, Lemma~\ref{lemma:prob-diff-lower} gives $\mathbb{P}'_{\mathrm{BellDiff}}(C(S^\star)) \geqslant (1 - \frac{\varepsilon_{\text{acc}}}{\gamma})^2$, which gives Eq.~\eqref{eq:arbi-good}. Since $\frac{p_b^\star}{p_g^\star}$ is a constant, this gives the LHS of Eq.~\eqref{eq:arbi-bad}.

So $p_g^\star$ and $p_b^\star$ only change by $O(\varepsilon_{\text{acc}})$ to leading order. In other words, a larger $\tilde{S}$ has no impact on the problem. If $\tilde{S}$ adds an element of $S$ that anticommutes with $\sigma_k$, $p_b^\star = 0$ while Eq.~\eqref{eq:arbi-good} still holds for $p_g^\star$. This effectively reduces our problem to the much simpler pure-state case in Sec.~\ref{sec:pure-state-learning}, and a solution can be found immediately.

Single-Pauli noise shows why increasing $\tilde{S}$ provides no guaranteed gradual improvement in $\frac{q_b^\star}{q_g^\star}$. For a general Pauli channel acting on the output state, with total nonidentity error probability $\lambda<1/2$, let $p_g,p_b$ denote its good/bad probabilities before post-selection. We have
\begin{equation}
\begin{aligned}
p_b^\star &\leqslant p_b\leqslant1-(1-\lambda)^4\leqslant4\lambda,\\
p_g^\star &\geqslant\gamma_{\tilde S}^2\geqslant(\gamma-\varepsilon_{\text{acc}})^2.
\end{aligned}
\end{equation}
The first bound follows because a Bell difference uses four copies, and if all four Pauli errors are identity, the sample is good. The second follows from Lemma~\ref{lemma:prob-diff-lower} and Observation~\ref{obs:add-S-purity}. We can therefore take
\begin{equation}
q_g^\star=(\gamma-\varepsilon_{\text{acc}})^2,\qquad q_b^\star=p_b.
\end{equation}
Since $\gamma\geqslant(1-\lambda)^2$ from the identity branch of the channel, keeping $\gamma-\varepsilon_{\text{acc}}$ bounded below by a positive constant gives $q_b^\star/q_g^\star=O(p_b/p_g)=O(\lambda)$. The conditions below in terms of $p_b/p_g$ assume that $p_b$, or a constant-factor upper bound, is available; if only $\lambda$ is known, take $q_b^\star=1-(1-\lambda)^4$ instead. A known purity lower bound may replace $\gamma$ in $q_g^\star$, with the same retained-purity margin.
This leads to Theorem~\ref{thm:main-intermediate-arbi} for the intermediate regime and Theorem~\ref{thm:main-clean-arbi} for the clean regime.

\begin{theorem}[Intermediate Regime Under Arbitrary Pauli Noise]
\label{thm:main-intermediate-arbi}
When $\frac{p_b}{p_g} \leqslant \Theta\left(\frac{1}{t}\right)$, or $\lambda \leqslant \Theta\left(\frac{1}{t}\right)$, the intermediate mode of the up-down algorithm succeeds. The algorithm uses $O(\mathrm{poly})$ samples and $O(2^t \cdot \mathrm{poly})$ time to return an approximate weak stabilizer.
\end{theorem}
\begin{proof}
Similar to the proof of Theorem~\ref{thm:main-intermediate-depo}, we stitch together iterations and slowly increment $\tilde{S}$, while allocating a sufficiently small $\Delta \varepsilon$ for each step. Unlike with depolarizing noise, without exponentially decreasing $q_b^\star$, the condition needed becomes stricter with each batch update, since the ratio $\frac{\tilde{\nu}-t}{\tilde{\nu}+t}$ always decreases as $\tilde{\nu}$ decreases. The strictest condition occurs when $\tilde{\nu}-t = O(1)$ or negative, where we'll need $\frac{p_b}{p_g} \leqslant \Theta\left(\frac{1}{t}\right)$, which is ensured by $\lambda \leqslant \Theta\left(\frac{1}{t}\right)$. This stage will always be reached, even without assumptions such as $n \gg 1$ or $n - t \gg 1$.
\end{proof}

\begin{theorem}[Clean Regime Under Arbitrary Pauli Noise]
\label{thm:main-clean-arbi}
When
\begin{equation}
\frac{p_b}{p_g} \leqslant \Theta\left(\frac{1}{t}\right), \qquad \varepsilon \geqslant \Theta\left(t\frac{p_b}{p_g} \min\{t, \ln \frac{1}{\delta}\} \right) 
\end{equation}
or,
\begin{equation}
\lambda \leqslant \Theta\left(\frac{1}{t}\right), \qquad \varepsilon \geqslant \Theta\left(t \lambda \min\{t, \ln \frac{1}{\delta}\} \right)
\end{equation}
the clean mode of the up-down algorithm succeeds. The algorithm uses $O(\mathrm{poly})$ samples and $O(\mathrm{poly})$ time to return an approximate weak stabilizer.
\end{theorem}

\textit{Proof idea.} Without the exponential suppression of $q_b^\star$, we make each iteration jump as far as it can. Lemma~\ref{lemma:per-batch-clean} then shrinks $\tilde{\nu}_i + t$ geometrically, down to $2t + O(1)$, once $\frac{p_b}{p_g} \leqslant \Theta(\frac{1}{t})$; taking each batch as large as its landing allows, the per-iteration accuracy costs form a series dominated by its last terms, totaling $O(t \frac{p_b}{p_g})$ for the main stage. The terminal stage adds up to $\min\{t, \ln \frac{1}{\delta}\}$ updates at $\Theta(t \frac{p_b}{p_g})$ each, which gives the accuracy threshold in the statement.

\begin{proof}
Unlike depolarizing noise, a general Pauli channel need not suppress bad$^\star$ samples gradually as $\tilde{S}$ grows, so in each iteration, we aim to increment $\dim \tilde{S}$ to the largest extent. The proof logic then follows similar steps as in Theorem~\ref{thm:main-clean-depo}. 

\textit{Main stage.} We run the main stage when $\tilde{\nu} - t \gg 1$. Suppose at iteration $i$ we have nullity $n - \dim \tilde{S} \equiv \tilde{\nu}_i$. By Lemma~\ref{lemma:per-batch-clean}, with the smallest batch $\chi = C_2(\tilde{\nu}_i+t)$ we get $K_b = C_2C_3\frac{q_b^\star}{q_g^\star}(\tilde{\nu}_i+t)$ and can reach $\tilde{\nu}_{i+1} \leqslant \strictFloor{K_b} + t$, or:
\begin{equation}
\tilde{\nu}_{i+1} + t \lessapprox C_2 C_3 \frac{q_b^\star}{q_g^\star} (\tilde{\nu}_i+t) + 2t - O(1)  
\end{equation} 
where the $O(1)$ term accounts for rounding. At the terminal stage, for the same reason as in Theorem~\ref{thm:main-intermediate-arbi}, we will need $\left(\frac{q_b^\star}{q_g^\star}\right)_{\text{terminal}} \leqslant \Theta\left(\frac{1}{t}\right) < 1$ to pass the main obstacle. This means that at the $i$-th stage, $\frac{q_b^\star}{q_g^\star} \leqslant \left(\frac{q_b^\star}{q_g^\star}\right)_{\text{terminal}} \leqslant O(\frac{1}{t}) < 1$. Furthermore, $C_2, C_3 \approx 1$ until the $\frac{q_b^\star}{q_g^\star} (\tilde{\nu}_i+t)$ term is $O(1)$. As a result, $\tilde{\nu}_i + t$ is suppressed at least geometrically until it is $2t + O(1)$. For the accuracy, though, the smallest batch is wasteful. Where we land only depends on $\strictFloor{K_b}$, which stays the same if we round $K_b$ up to $\ceil{K_b}$. So, as in the proof of Theorem~\ref{thm:main-clean-depo}, we take the largest batch with this landing, $\chi = \ceil{K_b}\frac{q_g^\star}{C_3 q_b^\star} \geqslant C_2(\tilde{\nu}_i+t)$. Step $i$ then requires
\begin{equation}
\label{eq:arbi-clean-step-cost}
\Delta\varepsilon_i = C_1\frac{t}{\chi} = \frac{C_1C_3}{\ceil{K_b}}\, t\frac{q_b^\star}{q_g^\star},
\end{equation}
which is dominated by the end terms as $\tilde{\nu}_i+t$, and with it $K_b$, decreases geometrically.

\textit{Entering the terminal stage.} Once $K_b \leqslant 1$, we have $\ceil{K_b} = 1$ and $\strictFloor{K_b} = 0$, so this iteration already reaches $\dim\tilde S\geqslant n-t$, which ends the main stage. This way, $\Delta \varepsilon$ in the main stage is $O\left(t\frac{q_b^\star}{q_g^\star}\right)$. Turning this into $O\left(t\frac{p_b}{p_g}\right)$ needs $\gamma-\varepsilon_{\text{acc}}$ bounded below by a constant, and Assumption~\ref{asp:accuracy} ensures this automatically: since $\varepsilon_{\text{acc}}\leqslant\varepsilon<(1-2\lambda)^2/4\leqslant\gamma/4$, using $\gamma\geqslant(1-\lambda)^2$, we have $\gamma-\varepsilon_{\text{acc}}>\frac34\gamma$. Hence $q_g^\star=(\gamma-\varepsilon_{\text{acc}})^2=\Theta(\gamma^2)$, so $\frac{q_b^\star}{q_g^\star} = O\left(\frac{p_b}{p_g}\right)$. More generally, this step only requires a constant retained-purity margin and does not require the full strength of Assumption~\ref{asp:accuracy}. The same cap bounds each per-update budget away from $\gamma_{\tilde S}$: all budgets, the floor added below included, total at most $\varepsilon<\gamma/4$, and $\gamma-\varepsilon_{\text{acc}}\leqslant\gamma_{\tilde S}\leqslant\gamma$, so $\Delta\varepsilon<\frac{\gamma}{4}-\varepsilon_{\text{acc}}\leqslant\gamma_{\tilde S}-\frac34\gamma\leqslant\frac14\gamma_{\tilde S}$. This is not necessarily $\ll\gamma_{\tilde S}$, but Lemma~\ref{lemma:per-batch-clean} still applies through the exact bound Eq.~\eqref{eq:good-span-exact-bound}, valid for any $\Delta\varepsilon<\gamma_{\tilde S}$: its purity factor is at most $4/3$, so $C_1=O(1)$ uniformly rather than $C_1\approx1$, which the $\Theta$-level statements absorb. Therefore, $\varepsilon_{\text{acc}} = O\left(t \frac{p_b}{p_g}\right)$.

\textit{Total purity loss.} In addition to $\lambda < \Theta\left(\frac{1}{t}\right)$ as discussed, the terminal stage requires $\Delta \varepsilon = \Theta\left(t \frac{q_b^\star}{q_g^\star}\right) = \Theta\left(t \frac{p_b}{p_g}\right)$ per successful update. Therefore, the total cost of the algorithm is $\varepsilon \geqslant \Theta\left(t \min\{t, \ln \frac{1}{\delta}\} \frac{p_b}{p_g}\right)$. As in the depolarizing case, we can then relax $\Delta \varepsilon$ to keep the cost polynomial.
\end{proof}

Let us take stock. Under depolarizing noise, the intermediate mode succeeds whenever $\lambda \lessapprox 0.159$ (for $n \gg t$) at a cost of $O(2^t \cdot \mathrm{poly})$ time, and the clean mode brings this down to polynomial time once the accuracy budget additionally satisfies $\varepsilon \gtrapprox C \frac{t}{n-t}\frac{1-(1-\lambda)^4}{(1-\lambda)^4}$ (Theorems~\ref{thm:main-intermediate-depo} and~\ref{thm:main-clean-depo}). Under arbitrary Pauli noise, post-selection may no longer suppress bad samples as $\tilde S$ grows, and the thresholds tighten to $\lambda \leqslant \Theta(1/t)$ and $\varepsilon \geqslant \Theta\left(t \lambda \min\{t, \ln\frac{1}{\delta}\}\right)$ (Theorems~\ref{thm:main-intermediate-arbi} and~\ref{thm:main-clean-arbi}). Beyond these thresholds, the top-down algorithm of Sec.~\ref{sec:top-down} remains available at cost $O(4^t \cdot \mathrm{poly})$ as long as the purity $\gamma$ is bounded below by a constant. These boundaries are the regime diagram of Fig.~\ref{fig:phase-diagram}.

\section{Numerical Results}
\label{sec:numerics}

\def\NvAdevMax{0.001}
\def\NvAnBig{512}
\def\NvAnMax{2048}
\def\NvAratioMin{2}
\def\NvAratioMax{256}
\def\NvAhaarDev{0.008}
\def\NvAscanTrialsMin{100}
\def\NvAscanTrialsMax{300}
\def\NvAlamStarEight{0.134}
\def\NvVtrials{360}
\def\NvVL{1024}
\def\NvVdeps{0.01}
\def\NvVwrongMax{7\%}
\def\NvVwrongLaterCount{1}
\def\NvVsettingsLater{55}
\def\NvVdevMax{0.08}
\def\NvVdevLast{0.01}
\def\NvVnLast{10^{4}}
\def\NvVqgDev{0.1\%}
\def\NvLlamDone{0.32}
\def\NvLlamDdeep{0.16}
\def\NvLlamGlobal{0.118}
\def\NvLlamStar{0.111}
\def\NvLlamDeepMin{0.16}
\def\NvLlamDeepMax{0.18}
\def\NvLinst{24}
\def\NvLtrials{100}
\def\NvLweightOne{1.6}
\def\NvLweightDeep{5.7}
\def\NvLharmlessOne{40\%}
\def\NvLharmlessDeep{12\%}
\def\NvLruns{120}
\def\NvLcap{60}
\def\NvLsweepChecked{25}
\def\NvLsweepOk{23}
\def\NvLsweepViol{2}
\def\NvLsweepViolMax{1\%}
\def\NvLsweepSettings{60}
\def\NvLsweepReached{47}
\def\NvLsweepDip{17}
\def\NvLsweepDipMax{0.10}
\def\NvBlamDev{0.004}
\def\NvBpredFactor{4}
\def\NvBpredFrac{98\%}
\def\NvBpredMedian{1.1}
\def\NvBequalOnset{0.100}
\def\NvBequalLargeNone{True}
\def\NvBnuEqual{77}
\def\NvBnuDecr{64}
\def\NvBcoresEqualSmall{77}
\def\NvBcoresEqualLarge{37}
\def\NvBcoresEqualLargeDense{77}
\def\NvBcoresDecrSmall{64}
\def\NvBcoresDecrLarge{32}
\def\NvBcoresDecrLargeDense{64}
\def\NvBtrials{32}
\def\NvBepsFloor{2\times10^{-6}}
\def\NvCinst{16}
\def\NvCtrials{100}
\def\NvCplaceDev{0.011}
\def\NvCpairShallowDev{0.011}
\def\NvCpairDeep{0.126}
\def\NvCbrickDeep{0.156}

Having established the guarantees of the up-down finder, we now test them numerically. Sec.~\ref{sec:numerics-method} describes how we generate the samples and what we measure. Sec.~\ref{sec:numerics-threshold} studies the first batch under global depolarizing noise, where Sec.~\ref{sec:up-down-depo} places the noise threshold, and repeats the study with a verifier that works from Bell samples. Sec.~\ref{sec:numerics-regimes} adds near-stabilizers to the input state and maps out the clean and intermediate regimes. Sec.~\ref{sec:numerics-local} replaces global noise by two-qubit gate noise and runs the intermediate mode to completion.

\subsection{Simulation Method}
\label{sec:numerics-method}

Every simulation starts from an ideal state $\ket{\psi}=C\bigl(\ket{\phi}\otimes\ket{0}^{\otimes(n-t)}\bigr)$, where $C$ is a Clifford encoder and the core $\ket{\phi}$ has $t$ qubits and stabilizer nullity $\nu\leqslant t$. A Pauli channel $\mathcal E$ then gives the noisy state $\rho=\mathcal E(\psi)$. Rather than simulating Bell measurements on pairs of copies, we sample Bell differences directly. One Bell difference sample consumes four copies of $\rho$, two for each Bell sample, and each copy carries an independent error drawn from $\mathcal E$. Denote the Pauli coefficients of $\psi$ by $a'_j$ and the fidelities of $\mathcal E$ by $f_j$, as in Eq.~\eqref{eq:pauli-fidelities}. Eq.~\eqref{eq:bell-diff-dist} then gives
\begin{equation}
    P^{\rho}_{\rm BellDiff}(i)=4^n\sum_{j\in\mathcal P^n}(-1)^{\expval{i,j}}f_j^4\,{a'_j}^4
    =\sum_{k\in\mathcal P^n}q^{(4)}_k\,P^{\psi}_{\rm BellDiff}(i+k),
    \label{eq:numerics-convolution}
\end{equation}
where $q^{(4)}$ is the distribution of $k=k_1+k_2+k_3+k_4$, the sum of the four error labels. In other words, a noisy sample is an ideal sample shifted by four independent errors, some of which may be the identity. For global depolarizing noise, the shift vanishes with probability $(1-\lambda)^4$ and is otherwise uniform, as in Sec.~\ref{sec:up-down-depo}. We draw the ideal samples from the $4^t$ Pauli expectation values of the core. A Clifford encoder acts on Pauli labels by a symplectic map, which preserves spans, radicals and projected purities. We therefore work in the unencoded frame, where the target contains $S_0=\expval{Z_{t+1},\ldots,Z_n}$, and map gate errors into this frame.

For each batch $B$, we compute the radical $R=\rad(B)$ and the number $\ell=\dim(R\cap S_0)$ of target directions it contains. When $\nu<t$, the target also contains stabilizers of the core, so counting only $S_0$ is conservative. In the intermediate mode, a batch with $\ell\geqslant1$ yields an accepted update as soon as a uniform draw from $R$ lands in $S$, and we call such a batch successful. In the clean mode, we also compute the exact purity loss $\Delta\gamma(R)$ from the Pauli expectation values of the core. Thus, except in Fig.~\ref{fig:numerics-threshold}(d), the verifier is replaced by exact projected purities. Fig.~\ref{fig:numerics-threshold}(d) instead runs Algorithm~\ref{alg:stab-ver} on $N$ Bell samples of $\rho$, drawn independently of the samples used by the finder.

The batch sizes of Lemmas~\ref{lemma:per-batch-intermediate} and~\ref{lemma:per-batch-clean} depend on the bound $q_g^\star$. We scan $M$ and report the best value, which can only overstate the success of a fixed rule. Each trial draws one independent stream of samples, and its batch of size $M$ consists of the first $M$ samples of that stream. Different trials are therefore independent, while the batches of one trial at different $M$ are nested; only in Fig.~\ref{fig:numerics-threshold}(d) does each $M$ use fresh batches. A success probability at fixed $M$ is the fraction of successful trials. The rule itself is available in practice: under global depolarizing noise, $q_g^\star=(1-\lambda)^4$ equals $\gamma^2$ up to $O(2^{-n})$, and the purity $\gamma$ can be estimated from the verifier's Bell samples (Sec.~\ref{sec:up-down-depo}). Averaged over the verifier datasets of Fig.~\ref{fig:numerics-threshold}(d) with $N=10^4$, this estimate of $q_g^\star$ agrees with $(1-\lambda)^4$ to within $\NvVqgDev$ at every $\lambda$; the scans themselves use the exact value. All error bars and bands are $1\sigma$ ($68\%$) intervals computed from the simulations themselves. For success probabilities, which include cases with no observed events, we use the Wilson interval: for an estimate $\hat p$ from $N$ trials, it contains every $p$ with $|\hat p-p|\leqslant\sqrt{p(1-p)/N}$, and so keeps a nonzero width when none or all of the trials succeed. Where we maximize over $M$ or locate a threshold, we resample the success count at each $M$ independently from its observed rate and repeat the maximization; this parametric bootstrap ignores the correlation between the nested batches of one trial. In Figs.~\ref{fig:numerics-regimes} and~\ref{fig:numerics-local}(a,b), we bootstrap over the random cores or encoders. Where a value of $M$ or an iteration is selected from the data, as in Figs.~\ref{fig:numerics-threshold}(d) and~\ref{fig:numerics-local}(c,d), the interval is that of the selected cell and does not account for the selection.

\subsection{First Batch under Global Depolarizing Noise}
\label{sec:numerics-threshold}

At the first iteration, $\tilde S=\{I\}$, $q_g^\star=(1-\lambda)^4$ and $q_b^\star=1-(1-\lambda)^4$. Without the confidence margins of Eq.~\eqref{eq:constraint-on-M}, Lemma~\ref{lemma:per-batch-intermediate} with $t=\nu$ requires a batch size in the window
\begin{equation}
    \mathcal W=\left[\frac{n+\nu}{q_g^\star},\,\frac{n-\nu}{q_b^\star}\right],\qquad
    \lambda^\star(n,\nu)=1-\left(\frac{n+\nu}{2n}\right)^{1/4}.
    \label{eq:numerics-window}
\end{equation}
The window is nonempty exactly when $\lambda<\lambda^\star(n,\nu)$. This is the intermediate--noisy boundary of Fig.~\ref{fig:phase-diagram}, and it tends to $1-2^{-1/4}\approx0.159$ as $\nu/n\to0$. We test it on product-magic states,
\begin{equation}
    \ket{\psi}=C\left[(T\ket{+})^{\otimes t}\otimes\ket{0}^{\otimes(n-t)}\right],\qquad \nu=t.
    \label{eq:numerics-product-family}
\end{equation}

Fig.~\ref{fig:numerics-threshold}(a) shows how the success probability depends on the batch size at $n/\nu=8$ and $\lambda=0.11$. With too few samples, the span has not yet reached $S$ and the radical contains no target direction. With too many, bad samples reduce the radical to $\{I\}$. In between, success is close to one. The window $\mathcal W$ lies inside this range, but success already begins below its lower edge, so $\mathcal W$ is a sufficient range rather than the exact one.

Fig.~\ref{fig:numerics-threshold}(b) shows the success probability, maximized over $M$, as a function of $\lambda$. At fixed $n/\nu=8$, the threshold $\lambda^\star=\NvAlamStarEight$ is the same for every $n$. We find that the curves cross near $\lambda^\star$ and steepen as $n$ grows. Let $\lambda_c(n,\nu)$ be the noise strength at which the maximized success probability equals $\tfrac12$. Fig.~\ref{fig:numerics-threshold}(c) compares $\lambda_c$ with $\lambda^\star$ over a wide range of $n/\nu$. For $n\geqslant\NvAnBig$ and $\NvAratioMin\leqslant n/\nu\leqslant\NvAratioMax$, the two agree to within $\NvAdevMax$. Haar-random cores agree with $\lambda^\star$ to within about $\NvAhaarDev$, so the threshold is not specific to product cores.

Finally, Fig.~\ref{fig:numerics-threshold}(d) replaces the exact verifier by Algorithm~\ref{alg:stab-ver} with noisy verification samples, at $n=64$ and $\nu=8$, with $\Delta\varepsilon_-=0$ and $\Delta\varepsilon_+=\NvVdeps$. The finder draws up to $L=\NvVL$ elements of $R$ and stops at the first one the verifier approves. Two kinds of error can occur. A false acceptance approves a candidate outside $S$ and adds a generator that the algorithm cannot remove later. A false rejection discards a correct candidate and only costs further draws. With $N=10$ Bell samples, false acceptances occur in up to about $\NvVwrongMax$ of batches. For $N\geqslant30$, we ran $\NvVsettingsLater$ settings, one for each combination of $N$ and $\lambda\in\{0.06,0.07,\ldots,0.16\}$, with $\NvVtrials$ batches each, and observed $\NvVwrongLaterCount$ false acceptance in total at the values of $M$ shown. Compared with the exact verifier, the probability of a correct update is lower by no more than about $\NvVdevMax$ for every $N$, and about $\NvVdevLast$ at $N=\NvVnLast$. This gap includes false rejections and, at small $N$, false acceptances that preempt a correct update. For this family, the exact verifier is therefore a mild idealization. This is expected: every wrong candidate costs a purity loss of at least $(1-\lambda)^2/4$, far above the threshold $\Delta\varepsilon_+$.

\begin{figure}[t]
    \centering
    \includegraphics[width=\textwidth]{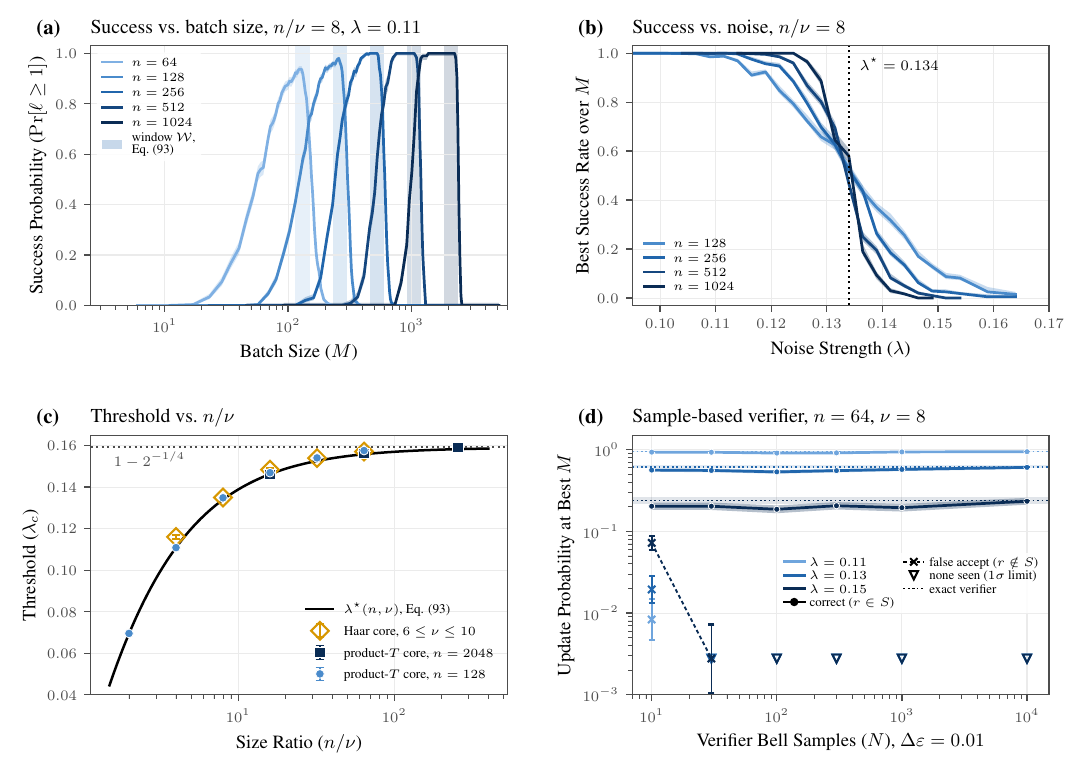}
    \caption{\textbf{First batch under global depolarizing noise.} Product-state-$T$ cores, Eq.~\eqref{eq:numerics-product-family}, with $\nu=t$. Success means that the radical of one batch contains a target direction, $\ell\geqslant1$. (a)~Success probability against the batch size $M$ at $n/\nu=8$ and $\lambda=0.11$, with the window $\mathcal W$ of Eq.~\eqref{eq:numerics-window} shaded; $\NvAscanTrialsMin$--$\NvAscanTrialsMax$ independent batches per point. (b)~Success rate maximized over $M$, against $\lambda$, at $n/\nu=8$. (c)~Noise strength $\lambda_c$ at which the maximized success probability equals $\tfrac12$. The solid curve is $\lambda^\star(n,\nu)$ and the dotted line its limit $1-2^{-1/4}$; diamonds are Haar-random $\nu$-qubit cores with $6\leqslant\nu\leqslant10$. (d)~First update with Algorithm~\ref{alg:stab-ver} on $N$ Bell samples, at $n=64$ and $\nu=8$, with $\Delta\varepsilon_+=\NvVdeps$, $L=\NvVL$ and $\NvVtrials$ batches per point. For each $N$, $M$ maximizes the probability of a correct update. Solid lines give that probability and dashed lines the probability of a false acceptance at the same $M$; dotted lines give the exact verifier. Open triangles mark settings without any false acceptance, plotted at the $1\sigma$ upper limit. Bands and bars are $1\sigma$ intervals (Sec.~\ref{sec:numerics-method}).}
    \label{fig:numerics-threshold}
\end{figure}

\subsection{Clean versus Intermediate Updates}
\label{sec:numerics-regimes}

For product-$T$ cores, every Pauli outside $S$ costs a purity loss of at least $(1-\lambda)^2/4$. An exact verifier then accepts a whole radical exactly when $R\subseteq S$, and the accuracy condition of Theorem~\ref{thm:main-clean-depo} never triggers. The condition matters when the core has near-stabilizers, that is, Paulis outside $S$ whose imposition costs little purity.

We build such cores on $t=8$ qubits from nine Clifford blocks interleaved with eight single-qubit rotations,
\begin{equation}
    \ket{\phi}=C_9\,\mathrm R_z^{(q_8)}(\theta_8)\,C_8\cdots C_2\,\mathrm R_z^{(q_1)}(\theta_1)\,C_1\ket{0}^{\otimes 8},
    \label{eq:numerics-doped-core}
\end{equation}
where each rotation acts on a uniformly random qubit $q_k$. Each block $C_k$ consists of three random all-to-all CNOT layers. In each layer, every qubit first receives a single-qubit Clifford word $H^aS^bH^c$ with independent uniformly random bits $a,b,c$, and four randomly oriented CNOTs then act on randomly chosen disjoint pairs of qubits. We use two angle schedules. For equal angles, $\theta_k=0.15$, near-stabilizer losses are typically about $\sin^2(0.15)/2\approx0.011$. For decreasing angles, $\theta_k=\pi/2^{k+1}$ with $k=1,\ldots,8$, they shrink roughly geometrically, by about a factor of four per near-stabilizer. We keep only cores with $\nu=7$, which have one exact stabilizer and seven near-stabilizers; we compute $\nu$ exactly from the $4^8$ Pauli expectation values. Each core is an independently drawn circuit. The cores are encoded into $n=64$ and $n=256$ qubits; as explained in Sec.~\ref{sec:numerics-method}, the choice of Clifford encoder does not affect any quantity we measure under global depolarizing noise.

For the first update, we set $\Delta\varepsilon=\varepsilon$. For each core and noise strength, we run $\NvBtrials$ independent trials and average the outcomes over trials and cores. We call $(\lambda,\varepsilon)$ clean if, for some $M$, the whole radical is accepted and contains a target direction with probability at least $\tfrac12$. We call it intermediate if this fails but $\ell\geqslant1$ with probability at least $\tfrac12$ for some $M$, and noisy otherwise. This classification tests a single batch of Lemma~\ref{lemma:per-batch-clean} with $M$ scanned, not the explicit clean-mode schedule of Appendix~\ref{app:clean-parameters}.
\begin{figure}[t]
    \centering
    \includegraphics[width=\textwidth]{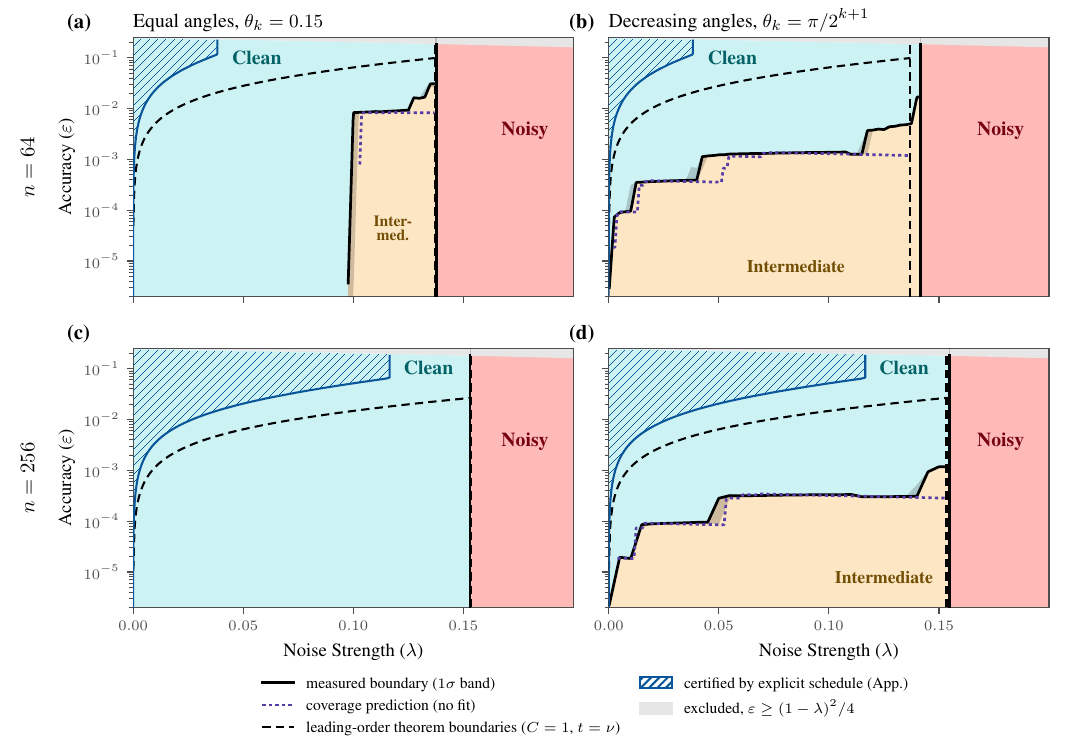}
    \caption{\textbf{Measured first-batch regimes with near-stabilizers} under global depolarizing noise; compare Fig.~\ref{fig:phase-diagram}, noting the logarithmic $\varepsilon$ axis. The cores have $t=8$ qubits and $\nu=7$, and are built from blocks of three random all-to-all CNOT layers, interleaved with eight $\mathrm R_z(\theta_k)$ insertions, Eq.~\eqref{eq:numerics-doped-core}: (a,c)~equal angles, $\theta_k=0.15$; (b,d)~decreasing angles, $\theta_k=\pi/2^{k+1}$. Rows: $n=64$ and $n=256$. Regions follow the classification in the text with $\Delta\varepsilon=\varepsilon$. Solid lines are the measured boundaries, with $1\sigma$ bands from a bootstrap over cores, and dotted lines the coverage prediction of the text. Dashed lines are the leading-order boundaries of Theorems~\ref{thm:main-intermediate-depo} and~\ref{thm:main-clean-depo} with $C=1$ and $t=\nu$. The hatched region is where the explicit clean-mode schedule of Appendix~\ref{app:clean-parameters} passes its check, Eq.~\eqref{eq:clean-check}, with $t=\nu$ and $\delta=0.1$; gray marks the region excluded by Assumption~\ref{asp:accuracy}. Each core is an independently drawn circuit, simulated with $\NvBtrials$ independent trials per noise strength. Equal angles use $\NvBcoresEqualSmall$ cores at $n=64$ and $\NvBcoresEqualLarge$ at $n=256$ ($\NvBcoresEqualLargeDense$ for $\lambda\geqslant0.09$, where we added cores near the intermediate and noisy boundaries); decreasing angles use $\NvBcoresDecrSmall$ and $\NvBcoresDecrLarge$ ($\NvBcoresDecrLargeDense$).}
    \label{fig:numerics-regimes}
\end{figure}
Fig.~\ref{fig:numerics-regimes} shows the measured regimes. The noisy boundary follows $\lambda^\star(n,7)$ to within about $\NvBlamDev$ for both families. For equal angles, an intermediate region appears at $n=64$ above $\lambda\approx\NvBequalOnset$ and is absent at $n=256$. For decreasing angles, the intermediate region lies below an edge that rises with $\lambda$ in steps, one per near-stabilizer, and it moves down as $n$ grows.

A simple coverage argument predicts these edges without fitted parameters. A near-stabilizer of loss $\Delta$ leaves the radical once a good sample anticommutes with it, which happens with probability about $2\Delta$ per good sample. The largest useful batch is the upper edge of $\mathcal W$, so a batch contains about $q_g^\star(n-\nu)/q_b^\star$ good samples. The whole radical is accepted if the near-stabilizers that remain in it have a total loss of at most $\varepsilon$, after rescaling by $(1-\lambda)^2$. We evaluate this condition with the measured near-stabilizer losses of each core and obtain the dotted lines in Fig.~\ref{fig:numerics-regimes}. The prediction falls below the measured edge just under $\lambda^\star$, where the window $\mathcal W$ closes and the largest useful batch is overestimated, and it places the onset of the equal-angle region slightly late. The argument also explains the dependence on $n$: larger $n$ allows larger batches, which remove weaker near-stabilizers. Finally, the measured clean region extends well beyond the boundaries of Theorem~\ref{thm:main-clean-depo}, which assumes the worst-case loss spectrum. The explicit schedule of Appendix~\ref{app:clean-parameters} certifies a smaller region still: for any $\varepsilon$, only $\lambda\lesssim0.04$ at $n=64$ and $\lambda\lesssim0.12$ at $n=256$, while the measured noisy boundary lies near $\lambda^\star$. The difference comes from the finite-size factors that the leading-order lines drop: the constant $C=\bar t/t$ of the appendix, about $1.9$ at $t=7$, the verifier slack $\kappa$, and the confidence margins of each batch, the last two of which shrink only slowly with $n$.

\subsection{Two-Qubit Gate Noise}
\label{sec:numerics-local}

The guarantees for arbitrary Pauli noise, Theorems~\ref{thm:main-intermediate-arbi} and~\ref{thm:main-clean-arbi}, assume nothing about the error distribution. Gate noise in a shallow circuit, by contrast, is a structured Pauli channel. We prepare $(T\ket{+})^{\otimes t}\otimes\ket{0}^{\otimes(n-t)}$ with $(n,t)=(16,4)$, $(32,4)$ and $(32,8)$, so that $\nu=t$, and apply a Clifford encoder of $D\in\{1,2,4,8,16\}$ brickwork layers. Each layer applies random single-qubit Clifford gates and then CNOT gates, each with a random control, on alternating nearest-neighbor pairs of an open chain whose first $t$ qubits carry the core. After each CNOT, two-qubit depolarizing noise applies one of the 15 nonidentity two-qubit Paulis with total probability $p$. Single-qubit gates, state preparation and the Bell measurement are ideal. We report the noise by the total error probability of the Pauli channel that each copy undergoes, as defined in Sec.~\ref{sec:noise-model}: the probability that the faults of one copy, propagated to its output, multiply to a nonidentity Pauli. For global depolarizing noise this is $\lambda(1-4^{-n})$, so the two noise models are compared at equal total error probability. Each point in Figs.~\ref{fig:numerics-local}(a,b) averages over $\NvLinst$ random encoders with $\NvLtrials$ independent trials each. The complete runs of Figs.~\ref{fig:numerics-local}(c,d) and the connectivity study below instead set the noise by $1-e^{-pL}$, where $L$ is the number of noisy gates per copy, which differs from the total error probability by less than $0.02$ over the range shown.
\begin{figure}[t]
    \centering
    \includegraphics[width=\textwidth]{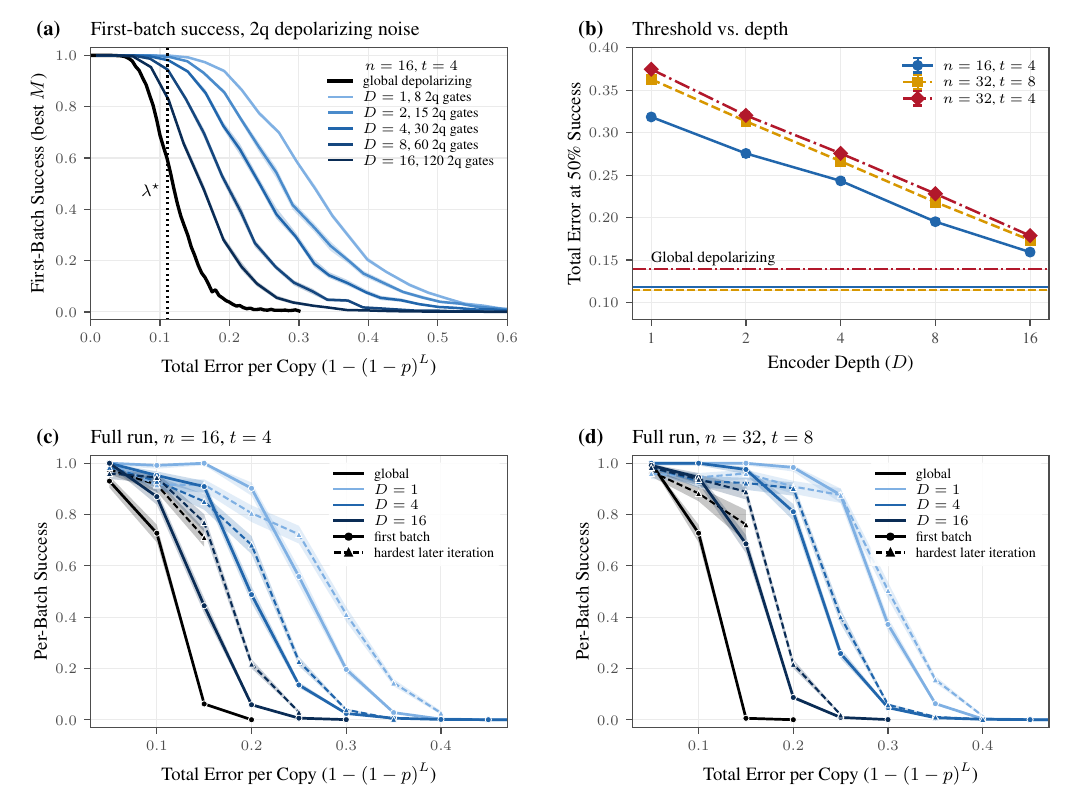}
    \caption{\textbf{Two-qubit gate noise.} Brickwork Clifford encoders of depth $D$ on an open chain, with two-qubit depolarizing noise of rate $p$ after every CNOT and $L$ noisy gates per copy; black lines show global depolarizing noise. (a)~Best first-batch success at $n=16$ and $t=4$ against the total error probability per copy; the dotted line marks $\lambda^\star(16,4)$. (b)~Total error probability at which the best first-batch success equals $\tfrac12$, against $D$; horizontal lines show global depolarizing noise. (c,d)~Complete intermediate-mode runs at $(n,t)=(16,4)$ and $(32,8)$. Solid lines give the per-batch success of the first iteration, the fraction of its batches that succeed, pooled over runs; dashed lines give the lowest such rate among later iterations. The horizontal axis is $1-e^{-pL}$ (Sec.~\ref{sec:numerics-local}). Later iterations are shown only where at least 10 runs reached them, and an iteration that needs more than $\NvLcap$ batches ends the run. Bands and bars are $1\sigma$ intervals, from a bootstrap over the $\NvLinst$ encoders in (a,b) and, in (c,d), from Wilson intervals over the pooled batches, which omit variation between encoders.}
    \label{fig:numerics-local}
\end{figure}

Fig.~\ref{fig:numerics-local}(a) shows the first-batch success at $n=16$ and $t=4$. For this encoder, gate noise is less harmful than global depolarizing noise with the same total error probability. The best first-batch success equals $\tfrac12$ at a total error probability of $\NvLlamDone$ for depth $1$ and $\NvLlamDdeep$ for depth $16$, compared with $\NvLlamGlobal$ for global noise and $\lambda^\star(16,4)=\NvLlamStar$. Fig.~\ref{fig:numerics-local}(b) shows the same trend for all three sizes; at depth $16$, the thresholds lie between $\NvLlamDeepMin$ and $\NvLlamDeepMax$. The advantage shrinks with depth because a deeper encoder spreads each fault over more qubits. In the unencoded frame, a single fault acts on $\NvLweightOne$ qubits on average at depth $1$ and on $\NvLweightDeep$ at depth $16$. Over the same range, the fraction of faults that leave the sample in $C(S)$, and so do no harm, falls from $\NvLharmlessOne$ to $\NvLharmlessDeep$. The errors thus become closer to uniform, and the threshold approaches that of global noise.

The brickwork encoder spreads a fault by at most one site per layer and places the core at one end of the chain. Repeating the first-batch study at $n=16$ and $t=4$ ($\NvCinst$ encoders per depth, $\NvCtrials$ trials each), placing the core at random positions changes each threshold by at most $\NvCplaceDev$. Replacing the nearest-neighbor pairs by a uniformly random perfect matching of all $n$ qubits leaves the thresholds at $D\leqslant2$ within $\NvCpairShallowDev$ but lowers them at larger depth, to $\NvCpairDeep$ at depth $16$, compared with $\NvCbrickDeep$ for brickwork in the same runs and $\NvLlamGlobal$ for global noise. The depth dependence in Fig.~\ref{fig:numerics-local}(b) thus reflects one-dimensional connectivity: an encoder that spreads faults faster brings gate noise closer to global depolarizing noise, as expected for well-scrambling circuits~\cite{dalzell2024}. For every encoder we tested, gate noise was no more harmful than global depolarizing noise with the same total error probability.

Under global depolarizing noise, the first batch is the hardest. Later batches are post-selected on $C(\tilde S)$, and because a bad sample is uniformly random, each accepted generator removes half of the bad samples while keeping every good one (Sec.~\ref{sec:up-down-depo}). For an arbitrary Pauli channel this need not happen, because errors that commute with $\tilde S$ survive post-selection, so a later iteration can have as many bad samples as the first (Sec.~\ref{sec:up-down-arbi}). Gate noise lies between these two cases, so we test it directly. We run the intermediate mode to completion. In each iteration, it accepts a uniformly random valid update: a Pauli that lies in the target $S$, is not yet in $\tilde S$, and lies in the radical of the batch up to elements of $\tilde S$; i.e., a uniformly random element of $(R+\tilde S)\cap S$ outside $\tilde S$. We draw this update directly from the known target, which models an exact verifier with an unlimited proposal budget. Each batch contains $M=\lceil C(n-\dim\tilde S+t)/p_g\rceil$ samples before post-selection, where $p_g$ is the probability of a good sample, estimated from $3000$ samples classified against the known target; $C=1$ except in the end stage $\tilde\nu-t\leqslant2$ of Lemma~\ref{lemma:per-batch-intermediate}, where $C=1.5$. Figs.~\ref{fig:numerics-local}(c,d) compare the first iteration with the hardest later one, with $\NvLruns$ runs per point. Of the $\NvLsweepSettings$ settings (both sizes; $D\in\{1,4,16\}$ and global noise; noise in steps of $0.05$), $\NvLsweepReached$ have a later iteration that at least 10 runs reached. A later iteration succeeds less often than the first in only two cases: in $\NvLsweepViol$ settings where both succeed less than $\NvLsweepViolMax$ of the time, and in $\NvLsweepDip$ settings with first-iteration success above $0.8$, where one of the last four iterations falls short by up to $\NvLsweepDipMax$. This also happens once under global depolarizing noise, so we attribute it to the batch-size rule near the end of a run, when few target directions remain, rather than to errors that survive post-selection. Wherever the first iteration limits success, it remains the bottleneck under brickwork gate noise, as under global depolarizing noise. We do not repeat the complete runs for the other encoders.

\subsection{Scope}
\label{sec:numerics-scope}

Except for Fig.~\ref{fig:numerics-threshold}(d), our simulations use exact projected purities. They therefore test the finder rather than the sample cost of the verifier, Theorem~\ref{thm:stab-ver}. We test the sample-based verifier only for product-$T$ cores, whose wrong candidates carry a large purity loss. Near-stabilizer cores, for which the verifier must resolve losses of order $\varepsilon$, remain to be tested. The complete runs of Figs.~\ref{fig:numerics-local}(c,d) draw each update directly from the known target, so they test neither a finite proposal budget $L$ nor the number of verifier queries. Neither the worst-case channels behind Theorems~\ref{thm:main-intermediate-arbi} and~\ref{thm:main-clean-arbi} nor coherent errors or noisy Bell measurements are simulated. We also do not run the explicit clean-mode schedule of Appendix~\ref{app:clean-parameters}; Fig.~\ref{fig:numerics-regimes} shows only where it is certified.

\FloatBarrier

\section{Conclusion and Outlook}
\label{sec:conclusion}

For pure states, stabilizer learning takes a single step: span enough Bell difference samples and compute the radical. Under noise, each part of this step fails, and most of this paper is spent repairing them. First, the span can no longer be trusted, because bad samples shrink its radical, so every update must be certified. Our verifier does this from one set of Bell samples with a guarantee that holds uniformly over all queries, so the finder may choose its queries adaptively (Theorem~\ref{thm:stab-ver}). Second, certification alone is slow: proposing samples one at a time works at any noise strength but costs $4^t$ (Theorem~\ref{thm:top-down}). Recovering polynomial cost requires spanning batches again, which works because bad samples degrade the radical one direction at a time, and the radical extends beyond the target by at most $t$ directions. Lemmas~\ref{lemma:per-batch-intermediate} and~\ref{lemma:per-batch-clean} identify batch sizes at which good samples reveal the target before bad samples corrupt the radical. Third, the verifier certifies purity while the finder sees only where samples fall. Lemma~\ref{lemma:belldiff-and-gamma} connects the two: any loss of purity shows up as samples leaving the centralizer, which sets how many good samples a batch needs. Finally, Section~\ref{sec:end-to-end} chains batches into a full algorithm. Under depolarizing noise, each accepted direction halves the rate of bad samples in later batches, so the first iteration sets the noise thresholds, $\lambda \lessapprox 0.159$ and, for the clean mode, $\lambda \lessapprox n\varepsilon/4t$.

The learned structure has several uses, and in our clean and intermediate regimes it is learned from polynomially many noisy copies. Here $\lambda$ is the total error probability of the channel; for independent single-qubit errors of probability $p$ it equals $1-(1-p)^n$, so our thresholds require $p$ of order $1/n$, or a channel whose total error probability is suppressed, such as the residual noise on an error-corrected code state. For encoded states, $t$ can be taken as the number of logical qubits, and the learner characterizes the stabilizer structure a device actually prepares, a natural diagnostic in the early fault-tolerant era. The learned Clifford frame identifies the approximate low-nullity structure of the underlying pure state and supplies a structural ingredient for compressed tomography~\cite{grewal2025, hangleiter2024}. Recovering the ideal state additionally requires information about the channel, which we leave to future work. Finally, the nullity $\tilde\nu\leqslant\nu$ bounds the number of $T$ gates needed to prepare the ideal state, a noise-robust version of the bound of Ref.~\cite{hangleiter2024} that can serve as a benchmark of a device's non-Clifford content.

Several questions remain open. Our results fall into three regimes, each served best by a different variant of the finder--verifier algorithm, which raises the issue of lower bounds and the optimality of learning weak stabilizers. Any lower bound exponential in $t$ must be computational: with runtime exponential in $n$, polynomially many copies always suffice, because the verifier can check all $4^n$ Paulis as candidates (Sec.~\ref{sec:ver-discussion}). A related question is whether a single algorithm can interpolate smoothly between the three regimes, and whether such an algorithm gives further insight into the expected lower bounds. More sophisticated finders may also use information from the Bell samples that ours discard, such as the Pauli overlaps $|\tr(\sigma\rho)|^2$. On average, $|\tr(\sigma\rho)|$ scales as $2^{-t/2}$, so post-selecting on these overlaps, measured directly on single copies rather than through their squares in the Bell samples, may lead to a $2^t$ algorithm for the noisy regime. Finally, simulations indicate that the first-batch threshold $\lambda^\star(n,\nu)$, which tends to $0.159$, is sharp for the up-down finder, with a transition that steepens as $n$ grows (Sec.~\ref{sec:numerics-threshold}). We do not know whether other finders can exceed it, or whether the $\Theta(1/t)$ thresholds for Pauli channels are fundamental.

The noise model can also be refined and generalized. Coherent noise breaks even the weak symmetries of a state, but a small rotation changes the approximate-stabilizer target only slightly. We therefore expect an approximate target to make learning more robust to coherent errors, not less. The finder would still need to be adapted, because a coherent rotation spreads the state's Pauli weight outside $C(\mathrm{Stab}(\psi))$, which our analysis of good and bad samples does not account for. Within Pauli noise, we analyzed two extremes: global depolarizing noise and worst-case Pauli channels. We expect realistic, average-case error distributions to fall between these two analyses, and the gap between them to narrow. For two-qubit depolarizing noise on the Clifford encoders we tested, the first batch even tolerates more total error than under global depolarizing noise, and the margin shrinks as the encoder spreads faults more widely. In complete runs under brickwork gate noise, later iterations were harder than the first batch only near the end of a run, where the same happens under global depolarizing noise (Sec.~\ref{sec:numerics-local}). For such distributions, post-selection may still work to a large extent, reducing bad samples down to the grain size of the error distribution, set by its heaviest errors that commute with every accepted generator. For some circuit families, such as random circuits, local noise is known to act on single-copy output distributions as global white noise~\cite{dalzell2024}. Whether the same holds for the two-copy statistics of Bell sampling is open.

\begin{acknowledgments}
The authors thank Michael Gullans and Dominik Hangleiter for their contributions to the early development of this project, including discussions that motivated the study of Bell sampling in the presence of noise. We acknowledge helpful discussions with Daniel Spencer, Emil Khabiboulline, Timothy Connor Mooney, Kaiyan Shi, Zongbo Bao, Guoding Liu, and Qi Ye. L.T. was supported in part by funding through NSF QLCI award OMA-2120757. T.S. was supported by the NSF STAQ program and the U.S. Department of Energy, Office of Science; National Quantum Information Science Research Centers, Quantum Systems Accelerator (Award No. DE-SCL0000121). 

The authors used OpenAI o3 to assist with equation-level derivations for an earlier perturbative argument for Lemma~\ref{lemma:uniform-bell-estimates}, and GPT-6-Astra (OpenAI) to develop it into an exact proof. GPT-5.6-Sol-Pro (OpenAI) was used to search the completed proofs for gaps and to develop the explicit clean-mode schedule of Appendix~\ref{app:clean-parameters}. AI-assisted coding tools, Codex (OpenAI; GPT-5 through GPT-6.1 models) and Claude Code (Anthropic; Claude Opus 5.5 and Sonnet 5.5), were used to run the numerical simulations, to modify and debug the routines that process Bell and Bell difference samples, to implement the algorithm loops studied in Sec.~\ref{sec:numerics}, and to check the reported numerical results against the code and stored data. The direct sampling of Bell differences, Eq.~\eqref{eq:numerics-convolution}, was suggested by and developed with Claude Opus 5.5. AI tools were also used as writing aids, to prepare draft figures, and to assist with literature searches; all cited works were checked against the original sources. The authors directed and reviewed all AI-assisted material and take full responsibility for the content of this manuscript.
\end{acknowledgments}

\bibliography{bib}

\appendix

\section{Bell-Sampling Identities and Deferred Proofs}
\label{app:bell-sampling}

This appendix collects the deferred Bell-sampling derivations from the main text. 

\subsection{Combinatorial Identities}

Before the main derivations, we collect four identities on $\pi_Y$ and the symplectic form $\expval{\cdot,\cdot}$ that we use repeatedly below.

\begin{lemma}
\label{lemma:useful-properties}
The following properties hold:
\begin{enumerate}
    \item $\forall i,j \in \{0,1\}^{2n}, \;\; (-1)^{\pi_Y(i) + \pi_Y(j) + \pi_Y({i + j})} = (-1)^{\expval{i,j}}$;
    \item Consider $i,j \in \{0,1\}^{2n}$ such that $\expval{i,j} = 0$. Then $\pi_Y$ is additive in $\mathbb{Z}_2$ in the sense that $(-1)^{\pi_Y({i+j})} = (-1)^{\pi_Y(i)+\pi_Y(j)} $
    \item $\forall i \in \{0,1\}^{2n},\;\;\sum_{j\in \{0,1\}^{2n}}(-1)^{\expval{i,j}} = 4^n\delta_{i0}$;
    \item $\forall i \in \{0,1\}^{2n},\;\;\sum_{j\in \{0,1\}^{2n}}(-1)^{\pi_Y(j) + \expval{i,j}} = 2^n(-1)^{\pi_Y(i)}$;
\end{enumerate}
\end{lemma}
\begin{proof}\leavevmode\par

\begin{description}
\item[\textit{Property 1}] First, we show that this statement is true for one qubit: $a,b \in \{0,1\}^2$. $\sigma_a$ and $\sigma_b$ anticommute if and only if $\{\sigma_a, \sigma_b, \sigma_{a+b}\} = \{X,Y,Z\}$; in this case, $\pi_Y(\sigma_a) + \pi_Y(\sigma_b) + \pi_Y(\sigma_{a+b}) = 1$. Otherwise $\pi_Y(\sigma_a) + \pi_Y(\sigma_b) + \pi_Y(\sigma_{a+b}) = 0 \text{ or } 2$. Therefore the statement is true for one qubit.

Hence, this is true for all qubits because we can add the contributions from all the qubits in a string together. Given $i \in \{0,1\}^{2n}$, denote the $k$-th qubit component by $i_k \in \{0,1\}^2$, consisting of the $k$-th and $(k+n)$-th bits of $i$. So $(-1)^{\pi_Y(i) + \pi_Y(j) + \pi_Y(i+j)} = \prod_{k=1}^{n} (-1)^{\pi_Y(i_k) + \pi_Y(j_k) + \pi_Y((i+j)_k)} = \prod_{k=1}^{n} (-1)^{\expval{i_k, j_k}} = (-1)^{\expval{i,j}}$.
\item[\textit{Property 2}] The second property is a special case of the first property in the event that $\sigma_i$ and $\sigma_j$ commute.
\item[\textit{Property 3}] The third property comes from the fact that for any Pauli $\sigma_i$ except the identity, exactly half of all the possible Paulis $\sigma_j$ commute with it, while the other half anticommute.
\item[\textit{Property 4}] First we establish:
\begin{equation}
\sum_{i \in \mathcal{P}^n} (-1)^{\pi_Y(i)} = \sum_{i \in \mathcal{P}^n} (-1)^{\sum_{k=1}^{n} \pi_Y(i_k)} = \prod_{k=1}^{n} \biggl( \sum_{i_k \in \{0,1\}^2} (-1)^{\pi_Y(i_k)} \biggr) = 2^n
\end{equation}
With this, we find:
\begin{equation}
\begin{aligned}
  \sum_{j \in \mathcal{P}^n} (-1)^{\pi_Y(j) + \expval{i,j}} &= \sum_{j \in \mathcal{P}^n} (-1)^{\pi_Y(j) + \pi_Y(i) + \pi_Y(j) + \pi_Y(i+j)} \\
  &= (-1)^{\pi_Y(i)} \sum_{j \in \mathcal{P}^n} (-1)^{\pi_Y(i+j)} \\
  &= 2^n (-1)^{\pi_Y(i)}
\end{aligned}
\end{equation}
\end{description}
\end{proof}

\subsection{Bell Sampling}

We first derive Eq.~\eqref{eq:p-bell}. Expanding both copies of $\rho$ in Paulis and contracting the resulting tensor network collapses the double sum over $(j,k)$ to a single sum over $j$:
\begin{equation}\label{eq:p-bell-derivation}
\begin{aligned}
  P_{\mathrm{Bell}}\left( i \right) &  = \tr \left( {\ket{{\sigma _{i}}}\bra{{\sigma _{i}}}\left( {\rho  \otimes \rho } \right)} \right) \\
   &  = \sum_{j,k \in \left\{ {0,1} \right\}^{2n}} {{a_{j}}{a_{k}}\tr \left( \ket{\sigma_i}\bra{\sigma_i}({\sigma _j}\otimes{\sigma_k}) \right)}  \\
   & = \frac{1}{2^n}\sum_{j,k \in \left\{ {0,1} \right\}^{2n}} \left(a_{j}a_{k}\vcenter{\hbox{\includegraphics[height=4em]{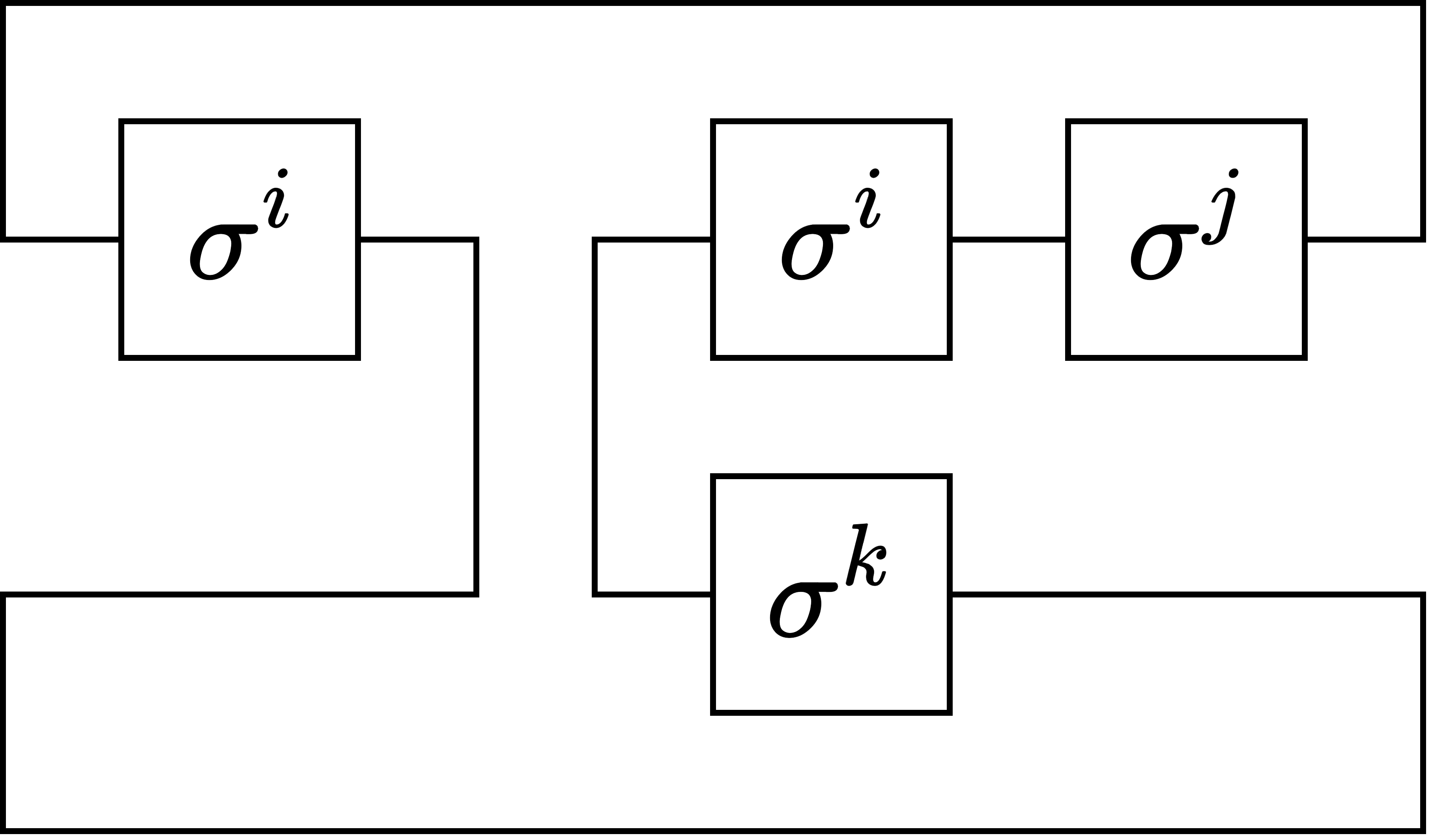}}}\right) \\
   & = \frac{1}{2^n}\sum_{j,k \in \left\{ {0,1} \right\}^{2n}} \left(a_{j}a_{k}\vcenter{\hbox{\includegraphics[height=2em]{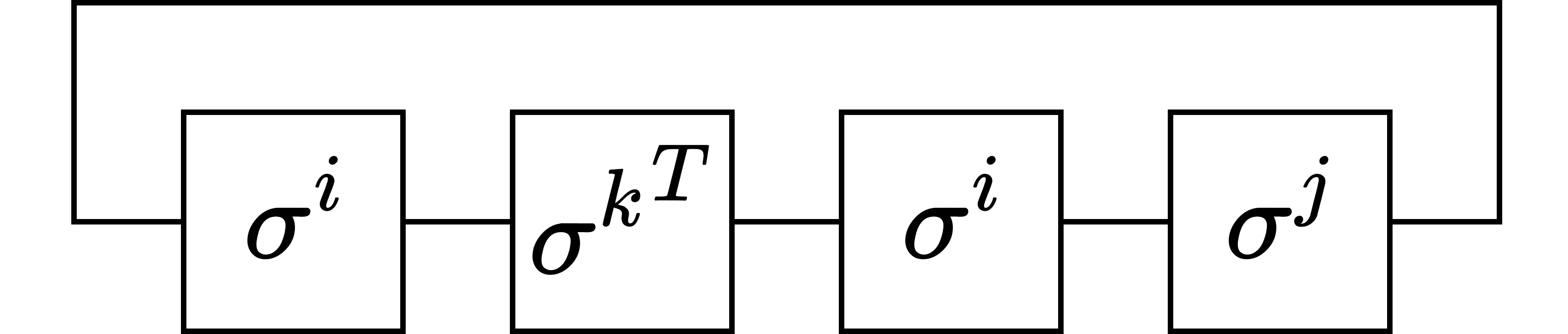}}}\right) \\
   &  = \frac{1}{2^n}\sum_{j,k \in \left\{ {0,1} \right\}^{2n}} {2^n{a_{j}}{a_{k}}{(-1)^{\pi_Y\left(k \right) + \left\langle {i,k} \right\rangle }}{\delta_{kj}}}  \\
   &  = \sum_{j \in \left\{ {0,1} \right\}^{2n}} {{(-1)^{\pi_Y\left(j \right) + \left\langle {i,j} \right\rangle }}a_{j}^{2}}.  \\
\end{aligned}
\end{equation}

Bell sampling, through a filtered expectation, gives access to the projected purity $\gamma_{\tilde{S}}$. We restate Lemma~\ref{lemma:projected-purity} from the main text and give its proof.

\begin{restated}{Lemma~\ref{lemma:projected-purity}}
For a commuting candidate group $\tilde{S} = \expval{g_1, g_2, \ldots, g_m}$, Bell sampling gives
\begin{equation}
\expval{(-1)^{\pi_Y(i)} \prod_{k=1}^{m} \left[(-1)^{\pi_Y(g_k) + \expval{g_k,i}} = 1\right]}_{i \sim \mathrm{Bell}} = \gamma_{\tilde{S}}
\end{equation}
\end{restated}
\begin{proof}
Property numbers below refer to Lemma~\ref{lemma:useful-properties}. Here $g_1,\ldots,g_m$ form an independent basis of $\tilde S$, so $|\tilde S|=2^m$ and expanding the product gives each $g\in\tilde S$ exactly once.
\begin{equation}
\label{eq:projected-purity}
\begin{aligned}
\expval{(-1)^{\pi_Y(i)} \prod_{k=1}^{m} \left[(-1)^{\pi_Y(g_k) + \expval{g_k,i}} = 1\right]}_{i \sim \mathrm{Bell}}
  &= \expval{(-1)^{\pi_Y(i)} \prod_{k=1}^{m} \tfrac{1}{2}\bigl(1 + (-1)^{\pi_Y(g_k) + \expval{g_k,i}}\bigr)}_{i \sim \mathrm{Bell}} \\
  &= \frac{1}{2^m} \sum_{g \in \tilde{S}} \expval{(-1)^{\pi_Y(i) + \pi_Y(g) + \expval{g,i}}}_{i \sim \mathrm{Bell}}
  && \text{\small[Prop.~2]} \\
  &= \frac{2^n}{2^m} \sum_{g \in \tilde{S}} \sum_{i \in \mathcal{P}^n} (-1)^{\expval{g,i}} a_i^2
  && \text{\small[Props.~1, 4]} \\
  &= 2^n \sum_{i \in \mathcal{P}^n} \biggl(\frac{1}{2^m} \sum_{g \in \tilde{S}} (-1)^{\expval{g,i}}\biggr) a_i^2 \\
  &= 2^n \sum_{i \in \mathcal{P}^n} \bigl[\sigma_i \in C(\tilde{S})\bigr]\, a_i^2 \\
  &= \gamma_{\tilde{S}}.
\end{aligned}
\end{equation}
\end{proof}

\subsection{Bell Difference Sampling}

To justify the second equality of Eq.~\eqref{eq:bell-diff-dist}, substitute $P_{\mathrm{Bell}}$ from Eq.~\eqref{eq:p-bell}:
\begin{equation}
\begin{aligned}
P_{\mathrm{BellDiff}}(i) &= \sum_{j \in \mathcal{P}^n} P_{\mathrm{Bell}}(j)\, P_{\mathrm{Bell}}(i+j) \\
&= \sum_{j,k,l \in \mathcal{P}^n} (-1)^{\pi_Y(k) + \pi_Y(l) + \expval{j,k} + \expval{i+j,l}} a_k^2 a_l^2 \\
&= \sum_{k,l \in \mathcal{P}^n} (-1)^{\pi_Y(k) + \pi_Y(l) + \expval{i,l}} a_k^2 a_l^2 \sum_{j \in \mathcal{P}^n} (-1)^{\expval{j,k+l}} \\
&= 4^n \sum_{k \in \mathcal{P}^n} (-1)^{\expval{i,k}} a_k^4
&& \text{\small[Prop.~3 of Lemma~\ref{lemma:useful-properties}]}.
\end{aligned}
\end{equation}

Analogous to Lemma~\ref{lemma:projected-purity} for Bell sampling, Bell difference sampling admits a closed form for the probability that a sample lands in $C(H)$:

\begin{lemma}
\label{lemma:projected-quartic-purity}
For any Pauli label subgroup $H \leqslant \mathcal{P}^n$, whether commuting or not,
\begin{equation}
\label{eq:bell-diff-prob-group}
\mathbb{P}_{b \sim \mathrm{BellDiff}}[b \in C(H)] = \frac{16^n}{|H|} \sum_{i \in H} a_i^4.
\end{equation}
\end{lemma}
\begin{proof}
Write $H = \expval{g_1, \ldots, g_h}$. Property numbers refer to Lemma~\ref{lemma:useful-properties}.
\begin{equation}
\begin{aligned}
\mathbb{P}_{b \sim \mathrm{BellDiff}}[b \in C(H)]
  &= \sum_{b \in C(H)} P_{\mathrm{BellDiff}}(b) \\
  &= 4^n \sum_{k \in \mathcal{P}^n} a_k^4 \sum_{b \in C(H)} (-1)^{\expval{k,b}} \\
  &= \frac{4^n}{|H|} \sum_{k \in \mathcal{P}^n} a_k^4 \sum_{m \in H} \sum_{b \in \mathcal{P}^n} (-1)^{\expval{k+m,b}} \\
  &= \frac{16^n}{|H|} \sum_{k \in \mathcal{P}^n} a_k^4 \sum_{m \in H} \delta_{k+m,0}
  && \text{\small[Prop.~3]} \\
  &= \frac{16^n}{|H|} \sum_{k \in H} a_k^4.
\end{aligned}
\end{equation}
\end{proof}

Applied to pure states, this pins the support of $P_{\mathrm{BellDiff}}$ inside $C(S)$.

\begin{restated}{Proposition~\ref{prop:pure-state-bell-diff}}
For a pure state $\psi$, $P_{\mathrm{BellDiff}}(i) \neq 0$ only if $i \in C(S)$.
\end{restated}
\begin{proof}
By Observation~\ref{obs:pure-state-stab}, $a_j = 0$ for $j \notin C(S)$, so Eq.~\eqref{eq:bell-diff-dist} restricts to a sum over $j \in C(S)$. Suppose $i \notin C(S)$; pick $g \in S$ with $\expval{i,g} = 1$. The map $j \mapsto j+g$ is a fixed-point-free involution on $C(S)$, and $a_j^4 = a_{j+g}^4$ since $\sigma_g$ stabilizes $\rho$. Since $\expval{i, j+g} = \expval{i, j} + 1$, terms cancel in pairs: $P_{\mathrm{BellDiff}}(i) = 0$.
\end{proof}

\section{Eigenspace-Weight Representation and \texorpdfstring{$\tau$}{tau} Coefficients}

\label{app:fourier}

\subsection{Simultaneous Eigenspaces and \texorpdfstring{$\tau$}{tau} Coefficients}

In the main paper, we have decomposed the density matrix in terms of the Paulis, $\rho=\sum_{i} a_i \sigma_i$. In this section, however, we introduce a more useful representation. Consider a commuting Pauli group $H = \expval{g_1, \ldots, g_h}$. For a bitstring $x \in \{0,1\}^h$, let $\Pi_x$ be the projector onto the simultaneous eigenspace in which $g_k$ has eigenvalue $(-1)^{x_k}$, and define its weight
\begin{equation}
\tau_x = \tr(\Pi_x \rho) = \tr(\Pi_x \rho \Pi_x).
\end{equation}
The $\Pi_x$ decompose the Hilbert space into the joint eigenspaces of $H$, and $\tau_x$ is the weight $\rho$ puts on each: the probability of the joint outcome $x$ when measuring $g_1, \ldots, g_h$. So $\tau_x \geqslant 0$ and $\sum_x \tau_x = 1$. These weights are defined for any $\rho$. Only when $\rho$ has Pauli support in $C(H)$ does it commute with every $\Pi_x$ and become block diagonal, $\rho = \sum_x \Pi_x \rho \Pi_x$; in general it is not.

Expanding $\Pi_x = \prod_{k=1}^h \frac12\bigl(I + (-1)^{x_k} g_k\bigr)$ writes each projector as a combination of the elements of $H$,
\begin{equation}
\Pi_x = \frac{1}{|H|} \sum_{i \in H} (-1)^{x(i)} \sigma_i .
\end{equation}
Here we choose the signs of $\sigma_i$, $i \in H$, so that $\sigma_i \sigma_j = \sigma_{i+j}$, which is possible because $H$ is commuting. Writing $\sigma_i = g_1^{y_1} \cdots g_h^{y_h}$, we set $x(i) = x \cdot y$, the bitwise dot product, thinking of $x$ as a function on $H$. Hence
\begin{equation}
\label{eq:a-to-tau}
\tau_x = \tr(\rho \Pi_x) = \frac{1}{|H|}\sum_{i \in H}(-1)^{x(i)} \tr(\rho \sigma_i) = \frac{2^n}{|H|}\sum_{i \in H}(-1)^{x(i)} a_i ,
\end{equation}
a Fourier transform over $H$. Since $\sum_x (-1)^{x(i+j)} = |H|\,\delta_{ij}$, it inverts to
\begin{equation}
a_i = \frac{1}{2^n}\sum_{x\in\{0,1\}^h} (-1)^{x(i)}\tau_x, \qquad i \in H.
\end{equation}
\subsection{BellDiff Mass in the \texorpdfstring{$\tau$}{tau} Representation}

With this expression, we can convert relations we found in the main text into this formalism. In particular, we are interested in the probability that a given Bell difference sample lies in $C(H)$ via Eq.~\eqref{eq:bell-diff-prob-group}:
\begin{equation}
\label{eq:bell-diff-prob-group-tau}
\begin{aligned}
\mathbb{P}_{b \sim \mathrm{BellDiff}}[b \in C(H)] &= \frac{16^n}{|H|} \sum_{i \in H} a_i^4 \\
  &= \frac{16^n}{|H|} \sum_{i \in H} \sum_{x,y,z,w \in \{0,1\}^h} \frac{1}{16^n} (-1)^{(x+y+z+w)(i)} \tau_x \tau_y \tau_z \tau_w \\
  &= \frac{1}{|H|} \sum_{x,y,z,w \in \{0,1\}^h} \biggl( \sum_{i \in H} (-1)^{(x+y+z+w)(i)} \biggr) \tau_x \tau_y \tau_z \tau_w \\
  &= \frac{1}{|H|} \sum_{x,y,z,w \in \{0,1\}^h} |H|\, \delta(x+y+z+w)\, \tau_x \tau_y \tau_z \tau_w \\
  &= \sum_{x,y,z,w \in \{0,1\}^h} \tau_x \tau_y \tau_z \tau_w\, \delta(x+y+z+w)
\end{aligned}  
\end{equation}
where the $\delta$ function takes the value 1 if and only if its argument is zero, and 0 otherwise.
Alternatively, we can also write:
\begin{equation}
\label{eq:bell-diff-prob-group-tau-conv}
\begin{aligned}
\mathbb{P}_{b \sim \mathrm{BellDiff}}[b \in C(H)] &= \sum_{x,z,s,t \in \{0,1\}^h} \tau_x \tau_{x+s} \tau_z \tau_{z+t} [s = t] \\
  &= \sum_{x,z,s \in \{0,1\}^h} \tau_x \tau_{x+s} \tau_z \tau_{z+s} \\
  &= \sum_{s \in \{0,1\}^h} \biggl( \sum_{x \in \{0,1\}^h} \tau_x \tau_{x+s} \biggr)^2 \\
  &= \norm{\tau * \tau}_2^2
\end{aligned}
\end{equation}
Here $*$ is ordinary XOR convolution, $(f*g)(s)=\sum_x f(x)g(s+x)$, with no normalization factor.

\subsection{Pauli Twirling and \texorpdfstring{$\gamma_H$}{gamma\_H}}

Recall from Sec.~\ref{sec:purity-norm} that the Pauli twirl $\mathcal{T}_H$ is a projector onto Paulis in $C(H)$: $\mathcal{T}_H(\rho) = \sum_{i \in C(H)} a_i \sigma_i$, so $\gamma_H = \tr(\mathcal{T}_H(\rho)^2)$. Here we describe $\mathcal{T}_H$'s action on the block structure of $\rho$ in the simultaneous eigenspace basis of $H$.
For a commuting Pauli group $H$, $\rho$ is not necessarily block diagonal in the joint eigenspace decomposition of its generators. In this case,
\begin{equation}
\rho = \sum_{x,y \in \{0,1\}^h} \Pi_x \rho \Pi_y
\end{equation}
where $\Pi_x \rho \Pi_y$ represents each submatrix in this form. What would $\mathcal{T}_H$ do to each term? For $i \in H$, $\sigma_i \Pi_x = \Pi_x \sigma_i = (-1)^{x(i)} \Pi_x$, since $\sigma_i$ acts as the sign $(-1)^{x(i)}$ on the eigenspace of $\Pi_x$. Averaging over $H$ then gives
\begin{equation}
\begin{aligned}
  \mathcal{T}_H(\Pi_x \rho \Pi_y) &= \frac{1}{|H|} \sum_{i \in H} \sigma_i \Pi_x \rho \Pi_y \sigma_i \\
  &= \frac{1}{|H|} \Pi_x \rho \Pi_y \sum_{i \in H} (-1)^{(x+y)(i)} \\
  &= \Pi_x \rho \Pi_y\, \delta_{xy}
\end{aligned} 
\end{equation}
In other words, $\mathcal{T}_H$ dephases all the off-diagonal blocks of the density matrix. This gives the following proposition:
\begin{proposition}
\label{prop:twirl-dephase}
$\mathcal{T}_H (\rho) = \sum_{x \in \{0,1\}^h} \Pi_x \rho \Pi_x$. 
\end{proposition}

The block form also compares projected purity with the collision probability of $\tau$:
\begin{proposition}
\label{prop:tau^2-lower-bound}
$\sum_{x \in \{0,1\}^h} \tau_x^2 \geqslant \gamma_H$.
\end{proposition}
\begin{proof}
$\gamma_H = \tr \left(\mathcal{T}_H (\rho)^2\right) = \sum_x \tr\left((\Pi_x\rho\Pi_x)^2\right) \leqslant \sum_x \left(\tr(\Pi_x\rho\Pi_x)\right)^2 = \sum_x\tau_x^2$.
\end{proof}
This directly leads to a simple lower bound for $\mathbb{P}_{b\sim{\text{BellDiff}}}\left[b \in C(H)\right]$:
\begin{proposition}
\label{prop:belldiff-lower-bound}
$\mathbb{P}_{b\sim{\text{BellDiff}}}\left[b \in C(H)\right] \geqslant \gamma_H^2$    
\end{proposition}
\begin{proof}
Note that $(\tau * \tau)(0) = \sum_{x \in \{0,1\}^h} \tau_x^2$. So $\mathbb{P}_{b\sim{\text{BellDiff}}}\left[b \in C(H)\right] = \norm{\tau * \tau}_2^2 \geqslant \left((\tau * \tau)(0)\right)^2 = \left(\sum_{x \in \{0,1\}^h} \tau_x^2\right)^2  \geqslant \gamma_H^2$, where the last inequality is Proposition~\ref{prop:tau^2-lower-bound}.
\end{proof}

\subsection{Purity Loss and Trace Distance}

To justify Eq.~\eqref{eq:purity-trace-distance}, let $\Delta\gamma_\psi(H)$ denote the purity loss of the pure state $\ket{\psi}$, with $\tau_x=\bra{\psi}\Pi_x\ket{\psi}$. Each block $\Pi_x\proj{\psi}\Pi_x$ has rank at most one, so $\Delta\gamma_\psi(H)=1-\sum_x\tau_x^2$. The closest pure state stabilized by $H$ is the normalized projection onto a most likely joint eigenspace. Since $(\max_x\tau_x)^2\leqslant\sum_x\tau_x^2\leqslant\max_x\tau_x$,
\begin{equation}
\frac{\Delta\gamma_\psi(H)}{2}
\leqslant\inf_{\phi:\,H\subseteq\mathrm{Stab}(\phi)}D_{\rm tr}(\psi,\phi)^2
=1-\max_x\tau_x
\leqslant\Delta\gamma_\psi(H).
\end{equation}
For the noisy state $\rho$, Eq.~\eqref{eq:pauli-fidelities} gives $1-2\lambda\leqslant f_i\leqslant1$ when the total Pauli error probability is $\lambda<1/2$. Applying these bounds to the squared Pauli coefficients outside $C(H)$ gives
\begin{equation}
(1-2\lambda)^2\Delta\gamma_\psi(H)
\leqslant\Delta\gamma(H)
\leqslant\Delta\gamma_\psi(H).
\end{equation}
Combining the two inequalities with $H=\tilde S$ proves Eq.~\eqref{eq:purity-trace-distance}. For depolarizing noise, the exact relation is $\Delta\gamma(H)=(1-\lambda)^2\Delta\gamma_\psi(H)$, so the upper bound instead has denominator $1-\lambda$.

\section{Proof of Lemma~\ref{lemma:belldiff-and-gamma}}
\label{app:density-matrix-property}

We relate the purity lost when extending $H$ to $H'$ to the probability that a Bell difference sample in $C(H)$ falls outside $C(H')$.

Let $H \subseteq H' \subseteq C(H)$ with $H$ commuting, and write $\Delta\gamma=\gamma_H-\gamma_{H'}$. For a Pauli label subgroup $G$, set
\begin{equation}
P_G=\Pr_{b\sim\mathrm{BellDiff}}[b\in C(G)]
=\frac{1}{|G|}\sum_{g\in G}\bigl(\tr(\rho\sigma_g)\bigr)^4,
\qquad R=\frac{P_{H'}}{P_H}.
\end{equation}
The fourth-moment identity is Lemma~\ref{lemma:projected-quartic-purity}; it applies even when $G$ is noncommuting. We prove the bound in two cases.

\subsection{Commuting extensions}

Suppose first that $H'$ is commuting. Choose its generators to extend those of $H$, and let $\Pi_x$ and $\Pi_{xy}$ be the coarse and refined eigenspace projectors. As in Appendix~\ref{app:fourier}, their weights satisfy
\begin{equation}
\tau_x=\tr(\Pi_x\rho),\qquad
\tau_{xy}=\tr(\Pi_{xy}\rho),\qquad
\sum_y\tau_{xy}=\tau_x.
\end{equation}
The two collision quantities we need are
\begin{equation}
\gamma_\tau=\sum_x\tau_x^2,\qquad
\Delta\gamma_\tau=\sum_x\tau_x^2-\sum_{x,y}\tau_{xy}^2.
\end{equation}
They need not equal the corresponding purities of a mixed state. Instead,
\begin{equation}
\label{eq:tau^2-bound}
\gamma_\tau\geqslant\gamma_H,\qquad
\Delta\gamma_\tau\geqslant\Delta\gamma.
\end{equation}
The first inequality is Proposition~\ref{prop:tau^2-lower-bound}. For the second, twirling by $H'$ removes the off-diagonal blocks within each $H$-eigenspace (Proposition~\ref{prop:twirl-dephase}), so
\begin{equation}
\label{eq:F-norm-sum}
\Delta\gamma
=\sum_x\sum_{y\ne z}\|\Pi_{xy}\rho\Pi_{xz}\|_{\mathrm{HS}}^2
\leqslant\sum_x\sum_{y\ne z}\tau_{xy}\tau_{xz}
=\Delta\gamma_\tau.
\end{equation}
To see the inequality, write $\rho=VV^\dagger$ and $V_{xy}=\Pi_{xy}V$. Then $\|\Pi_{xy}\rho\Pi_{xz}\|_{\mathrm{HS}}=\|V_{xy}V_{xz}^\dagger\|_{\mathrm{HS}}\leqslant\|V_{xy}\|_{\mathrm{HS}}\|V_{xz}\|_{\mathrm{HS}}=\sqrt{\tau_{xy}\tau_{xz}}$.

Write $\tau_H(x)=\tau_x$ and $\tau_{H'}(x,y)=\tau_{xy}$ to distinguish the convolution domains. Equation~\eqref{eq:bell-diff-prob-group-tau-conv} gives
\begin{equation}
\label{eq:ratio-form}
P_H=\|\tau_H*\tau_H\|_2^2,\qquad
P_{H'}=\|\tau_{H'}*\tau_{H'}\|_2^2.
\end{equation}
Marginalization commutes with convolution: $\sum_y(\tau_{H'}*\tau_{H'})(s,y)=(\tau_H*\tau_H)(s)$. Thus
\begin{equation}
\begin{aligned}
P_H-P_{H'}
&=\sum_s\left[\left(\sum_y(\tau_{H'}*\tau_{H'})(s,y)\right)^2
-\sum_y(\tau_{H'}*\tau_{H'})(s,y)^2\right]\\
&\geqslant\gamma_\tau^2-\sum_y(\tau_{H'}*\tau_{H'})(0,y)^2.
\end{aligned}
\end{equation}
Every summand is nonnegative, so we kept only $s=0$. In this sector, the total mass is $\gamma_\tau$, the origin entry is $\gamma_\tau-\Delta\gamma_\tau$, and the remaining entries sum to $\Delta\gamma_\tau$. Bounding the sum of their squares by the square of their sum gives
\begin{equation}
P_H-P_{H'}\geqslant
\gamma_\tau^2-(\gamma_\tau-\Delta\gamma_\tau)^2-(\Delta\gamma_\tau)^2
=2\Delta\gamma_\tau(\gamma_\tau-\Delta\gamma_\tau).
\end{equation}
We also need a bound that stays useful for large collision loss. By Cauchy--Schwarz, every entry satisfies
\begin{equation}
(\tau_{H'}*\tau_{H'})(0,y)
=\sum_{x,z}\tau_{xz}\tau_{x,z+y}
\leqslant\sum_{x,z}\tau_{xz}^2
=\gamma_\tau-\Delta\gamma_\tau.
\end{equation}
Multiplying this entrywise bound by the sector's total mass shows that
\begin{equation}
\label{eq:convolution-defect-bounds}
P_H-P_{H'}\geqslant
\max\left\{2\Delta\gamma_\tau(\gamma_\tau-\Delta\gamma_\tau),\,
\gamma_\tau\Delta\gamma_\tau\right\}.
\end{equation}

Now fix any $0\leqslant\Delta\varepsilon\leqslant\Delta\gamma$. If $\Delta\gamma_\tau\leqslant\gamma_\tau/2$, the quadratic in Eq.~\eqref{eq:convolution-defect-bounds} is increasing up to $\gamma_\tau/2$, and $\Delta\varepsilon\leqslant\Delta\gamma_\tau$ gives
$P_H-P_{H'}\geqslant2\Delta\varepsilon(\gamma_\tau-\Delta\varepsilon)$.
If $\Delta\gamma_\tau\geqslant\gamma_\tau/2$, the second bound gives the same conclusion, since
$\gamma_\tau\Delta\gamma_\tau\geqslant\gamma_\tau^2/2\geqslant2\Delta\varepsilon(\gamma_\tau-\Delta\varepsilon)$.

Finally, $\tau_H*\tau_H$ is a probability distribution, and each entry is at most its origin entry $\gamma_\tau$, again by Cauchy--Schwarz. Hence $P_H\leqslant\gamma_\tau$. Dividing the preceding bound by $P_H$, using $\Delta\varepsilon\leqslant\gamma_H\leqslant\gamma_\tau$, yields
\begin{equation}
1-R\geqslant2\Delta\varepsilon-\frac{2(\Delta\varepsilon)^2}{\gamma_\tau}
\geqslant2\Delta\varepsilon-\frac{2(\Delta\varepsilon)^2}{\gamma_H}.
\end{equation}
This proves Eq.~\eqref{eq:exact-escape-bound} for commuting $H'$. The same estimates also give the global bound $1-R\geqslant\Delta\gamma_\tau\geqslant\Delta\gamma$.

Figure~\ref{fig:tau-manifold} illustrates the relation between the equality sets and the positive-loss region. The bound above quantifies the escape throughout that region without expanding around an equality point.
\begin{figure}[ht]
\centering
\includegraphics[width=0.78\linewidth]{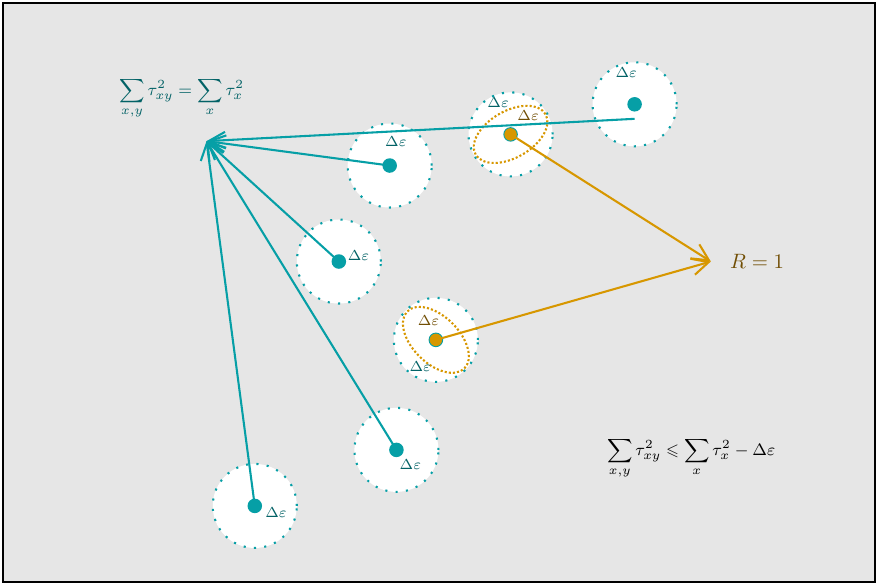}
\caption{Schematic of refined eigenspace weights for commuting $H\subseteq H'$, at fixed coarse weights $\tau_x$ and a small threshold $\Delta\varepsilon>0$. Teal points have $\Delta\gamma_\tau=0$; gold points also have $R=1$. Teal points without gold emphasize that zero collision loss need not imply $R=1$. Teal contours enclose $\Delta\gamma_\tau<\Delta\varepsilon$, while gold contours enclose $1-R<\Delta\varepsilon$ and lie within the teal contours because $1-R\geqslant\Delta\gamma_\tau$. The gray region has $\Delta\gamma_\tau\geqslant\Delta\varepsilon$. Contour labels indicate function values, not distances.}
\label{fig:tau-manifold}
\end{figure}

\FloatBarrier

\subsection{Noncommuting extensions}

Suppose now that $H'$ is noncommuting. Choose anticommuting Hermitian Paulis $\sigma_{i_1},\sigma_{i_2}$ in $H'$ and set $i_3=i_1+i_2$. The three observables $\sigma_{i_a}$ anticommute pairwise and commute with $H$. The label subgroup $L=H+\spn\{i_1,i_2\}$ has four $H$-cosets and is contained in $H'$.

We use only the eigenspaces of $H$. In each block, write
\begin{equation}
\rho_x=\Pi_x\rho\Pi_x,\qquad
\tau_x=\tr\rho_x,\qquad
m_a(x)=\tr(\rho_x\sigma_{i_a}).
\end{equation}
Their expectations obey $\sum_{a=1}^3m_a(x)^2\leqslant\tau_x^2$: for any real unit vector $u$, $(\sum_a u_a\sigma_{i_a})^2=I$, so its expectation in the positive block $\rho_x$ has magnitude at most $\tau_x$. H\"older's inequality therefore gives, for any four block labels,
\begin{equation}
\sum_{a=1}^3\prod_{j=1}^4m_a(x_j)
\leqslant\sum_{a=1}^3\prod_{j=1}^4|m_a(x_j)|
\leqslant\prod_{j=1}^4\left(\sum_{a=1}^3|m_a(x_j)|^4\right)^{1/4}
\leqslant\prod_{j=1}^4\tau_{x_j}.
\end{equation}
Choose compatible Hermitian representatives $\sigma_h$ for $h\in H$, as in Appendix~\ref{app:fourier}. Then
$\tr(\rho\sigma_h\sigma_{i_a})=\sum_x(-1)^{x(h)}m_a(x)$ and
$\tr(\rho\sigma_h)=\sum_x(-1)^{x(h)}\tau_x$.
Expanding their fourth powers and summing over $h$ retains exactly the tuples with $x_1+x_2+x_3+x_4=0$. The preceding inequality gives
\begin{equation}
\sum_{h\in H}\sum_{a=1}^3\bigl(\tr(\rho\sigma_h\sigma_{i_a})\bigr)^4
\leqslant\sum_{h\in H}\bigl(\tr(\rho\sigma_h)\bigr)^4.
\end{equation}
The fourth-moment formula, applied to the four $H$-cosets of $L$, now yields $P_L\leqslant P_H/2$. Since $C(H')\subseteq C(L)$,
\begin{equation}
\label{eq:noncommuting-escape}
R\leqslant\frac{P_L}{P_H}\leqslant\frac12.
\end{equation}
This implies Eq.~\eqref{eq:exact-escape-bound}, because
\begin{equation}
2\Delta\varepsilon-\frac{2(\Delta\varepsilon)^2}{\gamma_H}
=\frac{\gamma_H}{2}-\frac{2}{\gamma_H}\left(\Delta\varepsilon-\frac{\gamma_H}{2}\right)^2
\leqslant\frac12.
\end{equation}

The leading constant $2$ is sharp already for one qubit. For the pure state $\sqrt{p}\ket{0}+\sqrt{1-p}\ket{1}$, take $H=\{I\}$ and $H'=\expval{Z}$. Then $\gamma_H=1$, $\Delta\gamma=2p(1-p)$, and
\begin{equation}
1-R=\frac{1-(2p-1)^4}{2}=2\Delta\gamma-2(\Delta\gamma)^2,
\end{equation}
which attains the bound at $\Delta\varepsilon=\Delta\gamma$.

\section{Uniform Accuracy of the Reused Samples}
\label{app:uniform-verification}

We prove Lemma~\ref{lemma:uniform-bell-estimates} using only a count of filter patterns and Hoeffding's inequality. The point is that the verifier can only ever see finitely many distinct filter patterns on a fixed dataset, so a union bound over patterns stands in for a union bound over the (exponentially many) commuting groups. Write $d=2n+1$. For any $N$ Bell labels, the values of $Q_H$ select their intersection with the affine solution set
\begin{equation}
\{b:\expval{g_k,b}=\pi_Y(g_k)\text{ for every basis generator }g_k\in H\}.
\end{equation}
Each nonempty intersection is determined by at most $d$ of the selected points: take an affine basis of their affine hull. This hull lies in the solution set, so its intersection with the sampled labels is exactly the original selected set. Including the empty pattern, the number of possible filter patterns is therefore at most\footnote{To see the inequality, multiply the sum by $(d/N)^d$, bound it by $(1+d/N)^N$, and use $(1+d/N)^N\leqslant e^d$.}
\begin{equation}
\sum_{j=0}^{d}\binom Nj\leqslant(eN/d)^d\qquad(N\geqslant d).
\end{equation}
Each difference $Q_H-Q_{H'}$ is determined by two filter patterns, so there are at most $(eN/d)^{2d}$ difference patterns. Multiplication by the fixed parity signs does not change the number of patterns. The signed difference takes values in $\{-1,0,1\}$\footnote{In fact, $Q_H-Q_{H'}\in\{0,1\}$, since the constraints defining $Q_{H'}$ include those defining $Q_H$, and hence $Q_{H'}(b)\leqslant Q_H(b)$.} and has mean $\Delta\gamma_{H,H'}$.

Let $\widehat{\Delta\gamma}'_{H,H'}$ be the same estimate on an independent \emph{ghost} set of $M$ Bell samples---a standard symmetrization device that exists only inside this proof. If $M(\Delta\varepsilon_{\mathrm{gap}})^2\geqslant64$, then, for each fixed pair $H\subseteq H'$, Hoeffding's inequality gives
\begin{equation}
\Pr\!\left[|\widehat{\Delta\gamma}'_{H,H'}-\Delta\gamma_{H,H'}|>\Delta\varepsilon_{\mathrm{gap}}/4\right]
\leqslant2e^{-M(\Delta\varepsilon_{\mathrm{gap}})^2/32}\leqslant\tfrac12.
\end{equation}
On the event that some pair has $|\widehat{\Delta\gamma}_{H,H'}-\Delta\gamma_{H,H'}|>\Delta\varepsilon_{\mathrm{gap}}/2$, choose one such pair using only the first dataset. The second estimate is within $\Delta\varepsilon_{\mathrm{gap}}/4$ of its mean with probability at least $1/2$. Throughout, the suprema run over commuting groups. Hence
\begin{equation}
\Pr\!\left[\sup_{H\subseteq H'}|\widehat{\Delta\gamma}_{H,H'}-\Delta\gamma_{H,H'}|>\Delta\varepsilon_{\mathrm{gap}}/2\right]
\leqslant2\Pr\!\left[\sup_{H\subseteq H'}|\widehat{\Delta\gamma}_{H,H'}-\widehat{\Delta\gamma}'_{H,H'}|>\Delta\varepsilon_{\mathrm{gap}}/4\right].
\end{equation}
Now flip a fair coin for each of the $M$ sample pairs and swap the pair's two members on heads; this leaves the joint distribution unchanged. Conditional on the paired observations, each fixed group pair's difference of estimates is an average of independent, mean-zero random variables in $[-2,2]$. Hoeffding's inequality bounds its tail at $\Delta\varepsilon_{\mathrm{gap}}/4$ by $2e^{-M(\Delta\varepsilon_{\mathrm{gap}})^2/128}$. There are at most $(2eM/d)^{2d}$ value patterns on the $2M$ observations, so
\begin{equation}
\label{eq:uniform-bell-tail}
\Pr\!\left[\sup_{H\subseteq H'}|\widehat{\Delta\gamma}_{H,H'}-\Delta\gamma_{H,H'}|>\Delta\varepsilon_{\mathrm{gap}}/2\right]
\leqslant4(2eM/d)^{2d} e^{-M(\Delta\varepsilon_{\mathrm{gap}})^2/128},
\end{equation}
provided $M\geqslant\max\{d,64/(\Delta\varepsilon_{\mathrm{gap}})^2\}$. For example, the explicit choice
\begin{equation}
\label{eq:uniform-bell-sample-choice}
M=\left\lceil\frac{256}{(\Delta\varepsilon_{\mathrm{gap}})^2}
\left[2(2n+1)\ln\frac{2048}{(\Delta\varepsilon_{\mathrm{gap}})^2}+\ln\frac4\delta\right]\right\rceil
\end{equation}
makes the right-hand side at most $\delta$ for all $0<\Delta\varepsilon_{\mathrm{gap}}\leqslant1$ and $0<\delta<1$. This proves the lemma. The constants here are only convenient sufficient choices.

For Algorithm~\ref{alg:stab-ver}, assign the verifier failure probability $\delta_{\mathrm{ver}}$. Its dataset size is then
\begin{equation}
\label{eq:verifier-dataset-cost}
M_{\mathrm{ver}}=O\!\left(
\frac{n\ln(2/\Delta\varepsilon_{\mathrm{gap}})+\ln(1/\delta_{\mathrm{ver}})}
{(\Delta\varepsilon_{\mathrm{gap}})^2}\right).
\end{equation}
The algorithm never enumerates groups; the pattern count lives only in the proof. Each query evaluates the empirical average for its proposed generators.

\section{An Explicit Clean-Mode Schedule}
\label{app:clean-parameters}

The proof of Theorem~\ref{thm:main-clean-depo} keeps the clean mode's target dimensions, update budgets and batch sizes symbolic and works at leading order. Here we make one choice of all of them explicit, for global depolarizing noise with supplied $\lambda$ and $1\leqslant t<n$: the total budget, the target dimensions and attempt caps, the per-update budgets with their verifier thresholds, and the batch sizes. Whether this choice is certified is decided by a single inequality, Eq.~\eqref{eq:clean-check}, computed from the inputs before any data are taken. If it holds, the clean mode returns an $\varepsilon$-approximate weak stabilizer with probability at least $1-\delta-\delta_{\mathrm{ver}}$, with $\delta$ and $\delta_{\mathrm{ver}}$ the finder and verifier budgets of Sec.~\ref{sec:ver-discussion} (Proposition~\ref{prop:clean-schedule}). We then show that the check passes in the regime of the theorem, given a fixed margin. The constants are convenient rather than optimized, so a failed check only means that this particular choice is not certified by our bounds. Throughout, $D=n-t$, the four per-batch failure probabilities $\delta_1,\ldots,\delta_4$ in the proof of Lemma~\ref{lemma:per-batch-clean} are all $\frac{1}{16}$, and $\ell=\ln16$.

We first fix the total budget, the purity it guarantees, and the factor $\kappa$:
\begin{equation}
\label{eq:clean-param-basics}
\varepsilon_{\mathrm{tot}}=\min\left\{\varepsilon,\frac{(1-\lambda)^2}{8}\right\},
\qquad
\gamma_{\min}=(1-\lambda)^2-\varepsilon_{\mathrm{tot}},
\qquad
\kappa=\frac{1}{4+\ln(n+1)},
\end{equation}
where $\kappa$, inverse-logarithmic in $n$, will set both the verifier slack and the reserved floor, and the cap on $\varepsilon_{\mathrm{tot}}$ sits safely below that of Assumption~\ref{asp:accuracy}. As long as the accumulated purity loss stays within $\varepsilon_{\mathrm{tot}}$, Lemma~\ref{lemma:compatibility} applies and $\gamma_{\tilde S}\geqslant\gamma-\varepsilon_{\mathrm{tot}}\geqslant\gamma_{\min}$. By Lemma~\ref{lemma:prob-diff-lower} and Eq.~\eqref{eq:depo-q-b}, each raw Bell difference sample, before post-selection, is then good$^\star$ with probability $p_g^\star\geqslant\gamma_{\min}^2$ and bad$^\star$ with probability $p_b^\star\leqslant q_b\,2^{-\dim\tilde S}$, where $q_b=1-(1-\lambda)^4$.

The main stage aims at the target dimensions
\begin{equation}
\label{eq:clean-targets}
\tilde s_0=0,
\qquad
\tilde s_1=\left\lceil\sqrt{D}\,\right\rceil,
\qquad
\tilde s_{j+1}=D-\left\lfloor\frac{D-\tilde s_j}{2}\right\rfloor
\quad\text{until }\tilde s_J=D .
\end{equation}
The first target plays the role of $\Lambda$ in the proof of Theorem~\ref{thm:main-clean-depo}; after it, each target halves the remaining distance to $D$, so $J\leqslant1+\lceil\log_2D\rceil$. These are guaranteed dimensions, and an update may overshoot them. We count the terminal stage as stage $J+1$, with $\tilde s_{J+1}\equiv D$. Each stage gets at most $T=\lceil\log_4[(J+1)/\delta]\rceil$ batch attempts. The main stage applies at most one update per target and the terminal stage at most $m\equiv\min\{t,T\}$ (Observation~\ref{obs:per-iter-clean-terminal}), so at most $N_{\mathrm{tot}}=J+m$ updates are applied in total.

The per-update budgets follow from how many bad$^\star$ samples each stage can tolerate. To land at $\tilde s_i$, a batch may contain at most $K_{b,i}-1$ bad$^\star$ samples, with $K_{b,i}=D-\tilde s_i+1$ (Eq.~\eqref{eq:clean-mode-bound}); the terminal stage, with $K_{b,J+1}=1$, tolerates none. For an integer $K\geqslant1$, a binomial count $X$ with mean at most
\begin{equation}
\label{eq:clean-mu-max}
\mu_{\max}(K)=\max\left\{\frac{K}{16},\;K-\sqrt{2K\ell}\right\}
\end{equation}
reaches $K$ with probability at most $\frac{1}{16}$: by Markov's inequality for the first entry, and for the second by the Chernoff bound $\Pr[X\geqslant K]\leqslant e^{-(K-\mu)^2/2K}$, valid when the mean is at most $\mu\leqslant K$. Note that $K/16\leqslant\mu_{\max}(K)<K$, $\mu_{\max}(1)=\frac{1}{16}$, and $\mu_{\max}(K)=K[1-O(K^{-1/2})]$. We weight the stages by
\begin{equation}
w_i=\frac{2^{-\tilde s_{i-1}}}{\mu_{\max}(K_{b,i})},
\qquad
W=\sum_{j=1}^{J}w_j+m\,w_{J+1},
\end{equation}
and give each main-stage update ($i\leqslant J$) and each terminal update ($i=J+1$) the budget and thresholds
\begin{equation}
\label{eq:clean-allocation}
\Delta\varepsilon_i=\varepsilon_{\mathrm{tot}}\left[(1-\kappa)\frac{w_i}{W}+\frac{\kappa}{N_{\mathrm{tot}}}\right],
\qquad
\Delta\varepsilon_{i,+}=\Delta\varepsilon_i,
\qquad
\Delta\varepsilon_{i,-}=(1-\kappa)\,\Delta\varepsilon_i .
\end{equation}
The weighted part spends budget where bad$^\star$ samples are still frequent: a stage whose bad$^\star$ samples are rare, or which tolerates many of them, can afford a large batch and gets little. The additive part is the floor, added here rather than taken as a maximum as in Eq.~\eqref{eq:clean-accuracy-floor}. With the terminal budget counted $m$ times,
\begin{equation}
\label{eq:clean-allocation-sum}
\sum_{j=1}^{J}\Delta\varepsilon_j+m\,\Delta\varepsilon_{J+1}=\varepsilon_{\mathrm{tot}},
\qquad
\Delta\varepsilon_i\geqslant\frac{\kappa\,\varepsilon_{\mathrm{tot}}}{N_{\mathrm{tot}}} .
\end{equation}
Every query gap $\kappa\Delta\varepsilon_i$ is therefore at least $\Delta\varepsilon_{\mathrm{gap}}=\kappa^2\varepsilon_{\mathrm{tot}}/N_{\mathrm{tot}}$, and the verifier draws its dataset once, by Eq.~\eqref{eq:uniform-bell-sample-choice} with this gap and $\delta_{\mathrm{ver}}$.

For the batch sizes, let $\bar t=t+\frac12\bigl(\ell+\sqrt{\ell^2+4t\ell}\bigr)$, so that $\bar t/t\to1$ as $t\to\infty$. A batch that starts at dimension at least $\tilde s$ with budget $\Delta\varepsilon$, and hence lower threshold $\Delta\varepsilon_-=(1-\kappa)\Delta\varepsilon$, draws
\begin{equation}
\label{eq:clean-batch-size}
K_g=\left\lceil\max\left\{n-\tilde s+t+4,\;\frac{\bar t}{\Delta\varepsilon_-(1-\Delta\varepsilon_-/\gamma_{\min})}\right\}\right\rceil,
\qquad
M=\left\lceil\frac{K_g+\ell+\sqrt{\ell^2+2K_g\ell}}{\gamma_{\min}^2}\right\rceil
\end{equation}
raw Bell difference samples, and stage $i$ uses $M_i=M(\tilde s_{i-1},\Delta\varepsilon_i)$. The first stage starts before any update, with no purity loss yet, so $M_1$ uses $(1-\lambda)^4$ in place of $\gamma_{\min}^2$. The first entry of $K_g$ is Lemma~\ref{lemma:good-span-S}, since $\tilde\nu+\nu^\star\leqslant n-\tilde s+t$. The second is Eq.~\eqref{eq:good-span-exact-bound} with $\nu^\star\to t$, $\delta\to\frac{1}{16}$, $\gamma_{\tilde S}\to\gamma_{\min}$, and $\Delta\varepsilon\to\Delta\varepsilon_-$, the threshold below which the verifier must approve. These replacements only enlarge the bound, and $\Delta\varepsilon_-/\gamma_{\min}\leqslant\frac17$ because $\Delta\varepsilon_-\leqslant\varepsilon_{\mathrm{tot}}\leqslant\frac{(1-\lambda)^2}{8}\leqslant\frac{\gamma_{\min}}{7}$. Since $p_g^\star\geqslant\gamma_{\min}^2$, and $p_g^\star\geqslant(1-\lambda)^4$ in the first stage, the mean good$^\star$ count is at least $K_g+\ell+\sqrt{\ell^2+2K_g\ell}$, the larger root of $(\mu-K_g)^2=2\mu\ell$, so the Chernoff lower tail $e^{-(\mu-K_g)^2/2\mu}$ leaves fewer than $K_g$ good$^\star$ samples with probability at most $\frac{1}{16}$.

It remains to check the bad$^\star$ counts. In stage $i$, the bad$^\star$ count of a batch has mean at most $M_iq_b\,2^{-\tilde s_{i-1}}=q_bM_iw_i\,\mu_{\max}(K_{b,i})$. The choice is certified if this stays within $\mu_{\max}(K_{b,i})$ at every stage,
\begin{equation}
\label{eq:clean-check}
q_b\max_{1\leqslant i\leqslant J+1}M_iw_i\leqslant1 .
\end{equation}
At $\lambda=0$, $q_b=0$ and the check holds trivially.

We then run Algorithm~\ref{alg:stab-clean} with these parameters. The main stage visits the targets in order, skipping any that an overshoot has already reached. At target $j$, each attempt draws a fresh batch of $M_j$ samples, post-selects it onto $C(\tilde S)$ for the current $\tilde S$, and proposes the whole radical with thresholds $\Delta\varepsilon_{j,\pm}$. An approved update is applied, ending the stage, only if it reaches $\tilde s_j=n-t-\strictFloor{K_{b,j}}$, as in Algorithm~\ref{alg:stab-clean}; one that falls short spends no budget. After $T$ failed attempts we report failure. The terminal stage runs $T$ fresh batches of $M_{J+1}$ samples and applies every nontrivial approved update.

\begin{proposition}
\label{prop:clean-schedule}
If Eq.~\eqref{eq:clean-check} holds, the clean mode with the parameters above returns an $\varepsilon_{\mathrm{tot}}$-approximate, hence $\varepsilon$-approximate, weak stabilizer with probability at least $1-\delta-\delta_{\mathrm{ver}}$. It uses at most $(J+1)T$ batches and at most one verifier query per batch.
\end{proposition}
\begin{proof}
We begin with the budget. Fix the verifier dataset on the uniform-accuracy event of Lemma~\ref{lemma:uniform-bell-estimates}, which fails with probability at most $\delta_{\mathrm{ver}}$. Each applied update then loses at most its $\Delta\varepsilon_{i,+}$ (Theorem~\ref{thm:stab-ver}). At most one update is applied per target and at most $m$ in the terminal stage, so by Eq.~\eqref{eq:clean-allocation-sum} the accumulated loss stays within $\varepsilon_{\mathrm{tot}}$ on every execution history, and the bounds on $\gamma_{\tilde S}$, $p_g^\star$ and $p_b^\star$ above hold throughout.

Next, we bound the failure of a single batch. Condition on the verifier dataset and on the history before an attempt in stage $i$, so that $\dim\tilde S\geqslant\tilde s_{i-1}$. The attempt fails only if (a) it has fewer than $K_g$ good$^\star$ samples, (b) it has at least $K_{b,i}$ bad$^\star$ samples, or, given enough good$^\star$ samples, (c) their span misses $S^\star$ or (d) the coverage bound fails. Each has probability at most $\frac{1}{16}$: (a) by the choice of $M_i$, (b) by Eqs.~\eqref{eq:clean-mu-max} and~\eqref{eq:clean-check}, and (c) and (d) by the two lemmas behind the entries of $K_g$. A union bound, which needs no independence between the two counts, leaves success probability at least $\frac34$.

On a successful batch, the proof of Lemma~\ref{lemma:per-batch-clean} applies unchanged: the verifier approves the whole radical $R$, and $\tilde S'=R+\tilde S$ satisfies $\dim S^\star/(S^\star\cap\tilde S')\leqslant K_{b,i}-1$. Since $\dim S^\star\geqslant n-t$, in main stage $j$ this gives $\dim\tilde S'\geqslant n-t-(K_{b,j}-1)=\tilde s_j$, so the update is applied. In the terminal stage, $K_{b,J+1}=1$ gives $S\subseteq\tilde S'$; if $R$ adds nothing new, $S\subseteq\tilde S$ already.

Finally, until a stage reaches its goal, $\dim\tilde S\geqslant\tilde s_j$ in the main stage and $S\subseteq\tilde S$ in the terminal stage, each attempt succeeds with conditional probability at least $\frac34$. The stage thus fails with probability at most $4^{-T}$, and the $J+1$ stages together with probability at most $(J+1)4^{-T}\leqslant\delta$.
\end{proof}

To leading order, the check reproduces the conditions of Theorem~\ref{thm:main-clean-depo}. Let $D\to\infty$ with $tD2^{-D}\to0$, as in its proof, and $\varepsilon_{\mathrm{tot}}\to0$ with $(1-\lambda)^2$ bounded below. Since $K_{b,1}=D-\lceil\sqrt{D}\,\rceil+1$, the first weight is $w_1=[1+O(D^{-1/2})]/D$, with $1/w_1\leqslant D$. Every later stage starts at dimension at least $\lceil\sqrt{D}\,\rceil$, and $\mu_{\max}\geqslant\frac{1}{16}$, so
\begin{equation}
0\leqslant\frac{W-w_1}{w_1}\leqslant16D\left[(J-1)\,2^{-\lceil\sqrt{D}\rceil}+t\,2^{-D}\right]\to0 .
\end{equation}
Thus $W=(1+o(1))/D$: as in the proof of the theorem, the first jump takes almost all of the weighted budget. This only needs the later weights to be small in sum, not each smaller than the one before. Moreover, $\kappa\to0$, $\Delta\varepsilon_-/\gamma_{\min}\to0$, $\gamma_{\min}^2=(1-\lambda)^4(1+o(1))$, and $K_g\geqslant\bar t/\varepsilon_{\mathrm{tot}}\to\infty$, so the concentration and rounding corrections in Eq.~\eqref{eq:clean-batch-size} vanish uniformly over stages. With $\Delta\varepsilon_{i,-}\geqslant(1-\kappa)^2\varepsilon_{\mathrm{tot}}w_i/W$ in $M_i$, the left-hand side of Eq.~\eqref{eq:clean-check} is at most
\begin{equation}
\label{eq:clean-check-leading}
(1+o(1))\max\left\{\frac{n+t}{n-t},\;\frac{\bar t}{(n-t)\,\varepsilon_{\mathrm{tot}}}\right\}\frac{1-(1-\lambda)^4}{(1-\lambda)^4},
\end{equation}
where the first entry comes from the first stage, since later stages enter it only through the negligible weights $w_i\leqslant W-w_1$. Requiring each entry, times the last factor, to be at most one gives the two conditions of the theorem: the first is its noise condition, and the second its accuracy condition with $C=\bar t/t$, which tends to one for $t\gg1$ and is an explicit constant for bounded $t$; once $\varepsilon\ll1$, also $\varepsilon_{\mathrm{tot}}=\varepsilon$. So if both products stay below one by a fixed margin, the check passes for all large enough $D$. This certifies neither the boundary itself nor margins that vanish in the limit, and at fixed $\varepsilon$ it is Eq.~\eqref{eq:clean-check} itself that decides.

Appendix~\ref{app:verifier-complexity} counts the cost of this schedule; with retained purity bounded below, it gives the clean-mode bounds in Eq.~\eqref{eq:explicit-learning-costs}.

\section{Costs of the Three Finder Modes}
\label{app:verifier-complexity}

We count classical bit operations using ordinary binary linear algebra. A batch of $M_{\mathrm{batch}}$ Bell difference samples costs $O(n^2M_{\mathrm{batch}}+n^3)$ to post-select, form its span, and compute the radical. The cached verifier costs $O(nM_{\mathrm{ver}}+n^2)$ for a single-generator query and $O(n^2M_{\mathrm{ver}}+n^3)$ for a radical query. Its initialization costs $O(nM_{\mathrm{ver}})$. One Bell sample uses two state copies and one independent Bell difference sample uses four, so verifier and finder sample costs add. The main text leaves three vanishing factors as $o(1)$: the verifier slack, the reserved fraction of the accuracy budget, and the factor in the intermediate-mode budget $\varepsilon_{\mathrm{tot}}$. For explicit costs we set all three to one inverse-logarithmic value $\kappa$.

\textit{Top-down mode.} The gap is $\varepsilon/n$, giving
\begin{equation}
M_{\mathrm{ver}}=O\!\left(
\frac{n^2}{\varepsilon^2}
\left[n\ln(2n/\varepsilon)+\ln(1/\delta_{\mathrm{ver}})\right]\right).
\end{equation}
The proof of Theorem~\ref{thm:top-down}, keeping its confidence terms, requires
$T=O(4^t(n+\ln(1/\delta))/\gamma^2)$ Bell difference samples. There are at most $T$ single-generator queries. Thus the algorithm uses $2M_{\mathrm{ver}}+4T$ state copies and
$O(T(nM_{\mathrm{ver}}+n^2)+nM_{\mathrm{ver}})$ classical bit operations. For $\gamma=\Omega(1)$ and $0<\varepsilon\leqslant1$, these give the top-down bounds in Eq.~\eqref{eq:explicit-learning-costs}.

\textit{Intermediate mode.} There are at most $n$ accepted updates and at most $n+1$ iterations in total, each with $O(\ln(n/\delta))$ batch attempts. Under a fixed margin in the per-batch success conditions, each batch uses $O((n+t+1)/q_{\min})$ Bell difference samples and makes $O(2^t)$ single-generator queries, where $q_{\min}>0$ is a lower bound on $q_g^\star$ throughout the run. Consequently,
\begin{equation}
N_{\mathrm{diff}}=O\!\left(\frac{n(n+t+1)\ln(n/\delta)}{q_{\min}}\right),
\qquad
N_{\mathrm{query}}=O(2^t n\ln(n/\delta)).
\end{equation}
In the depolarizing analysis, $\Delta\varepsilon_{\mathrm{gap}}=\varepsilon_{\mathrm{tot}}/n$, with $\varepsilon_{\mathrm{tot}}=\min\{\varepsilon,\kappa(1-\lambda)^2\}$. If $(1-\lambda)^2$ is bounded below, Eq.~\eqref{eq:verifier-dataset-cost} gives $M_{\mathrm{ver}}=\widetilde O(n^3/\varepsilon^2)$. Summing the batch-processing and query costs gives
\begin{equation}
O\!\left(n^2N_{\mathrm{diff}}+n^4\ln(n/\delta)
+N_{\mathrm{query}}(nM_{\mathrm{ver}}+n^2)\right).
\end{equation}
For $q_{\min}=\Omega(1)$ and $t\leqslant n$, this is $\widetilde O(2^t n^5/\varepsilon^2)$, while the total state-copy cost is $\widetilde O(n^3/\varepsilon^2)$.

\textit{Clean mode.} Use the explicit schedule of Appendix~\ref{app:clean-parameters}. It has logarithmically many main-stage targets. Their total number of batch attempts, including the terminal stage, satisfies
\begin{equation}
N_{\mathrm{batch}}=O\!\left(
\log(n-t+1)\ln\frac{1+\log(n-t+1)}{\delta}
+\ln\frac1\delta\right).
\end{equation}
Both $N_{\mathrm{batch}}$ and the update bound $N_{\mathrm{tot}}$ are polylogarithmic in $n$ and $1/\delta$. Unsuccessful targets terminate with a reported failure, so these are also runtime caps. Redundant proposals are skipped.

The floor in Eq.~\eqref{eq:clean-allocation-sum} gives
\begin{equation}
\label{eq:clean-verifier-gap}
\Delta\varepsilon_{\min}=\frac{\kappa\varepsilon_{\mathrm{tot}}}{N_{\mathrm{tot}}},
\qquad
\Delta\varepsilon_{\mathrm{gap}}=\frac{\kappa^2\varepsilon_{\mathrm{tot}}}{N_{\mathrm{tot}}},
\end{equation}
the second being the verifier slack times the floor, with $\varepsilon_{\mathrm{tot}}=\min\{\varepsilon,(1-\lambda)^2/8\}$. Eq.~\eqref{eq:verifier-dataset-cost} then yields $M_{\mathrm{ver}}=\widetilde O(n/\varepsilon_{\mathrm{tot}}^2)$. The batch sizes of Eq.~\eqref{eq:clean-batch-size} have coverage thresholds between $(1-\kappa)\Delta\varepsilon_{\min}$ and $\gamma_{\min}/7$, so with $q_{\min}=\gamma_{\min}^2$ they give
\begin{equation}
N_{\mathrm{diff}}=O\!\left(
\frac{N_{\mathrm{batch}}}{q_{\min}}
\left[n+\frac{(t+1)N_{\mathrm{tot}}}{\kappa\varepsilon_{\mathrm{tot}}}\right]\right).
\end{equation}
There is at most one radical query per batch, including rejected proposals, so the total classical cost is
\begin{equation}
O\!\left(n^2N_{\mathrm{diff}}+
N_{\mathrm{batch}}(n^3+n^2M_{\mathrm{ver}})\right).
\end{equation}
With $q_{\min}=\Omega(1)$, $t\leqslant n$, and $0<\varepsilon\leqslant1$, so that $\varepsilon_{\mathrm{tot}}=\Theta(\varepsilon)$, these give $\widetilde O(n/\varepsilon^2)$ copies and $\widetilde O(n^3/\varepsilon^2)$ classical bit operations. The factors involving $q_{\min}$ and retained purity must be kept when those quantities are not bounded below by constants.

\end{document}